\documentclass[onecolumn,11pt]{article}  

\usepackage[intoc,refpage]{nomencl}
\usepackage[hcentering,vcentering,vmargin=2.7cm,hmargin=2.5cm%,showframe
]{geometry}
\usepackage[cmex10]{amsmath} % Use the [cmex10] option 
\usepackage{amsthm} 
\usepackage{amssymb}  % assumes amsmath package installed

\usepackage{bm}  % Define \bm{} to use bold math fonts

\usepackage[all]{xy} 

\usepackage{tikz} 
\usetikzlibrary{arrows,decorations.pathmorphing,backgrounds,positioning,fit,petri,topaths}
\usetikzlibrary{intersections}

\usepackage[pdftex]{hyperref}  % PDF hyperlinks

\theoremstyle{plain} 
\newtheorem{thm}{Theorem}
\newtheorem{lem}{Lemma}

\newtheorem{rmk}{Remark}

\newtheorem{cor}{Corollary}
\newtheorem{ex}{Example}

\newtheorem{obs}{Observation}

\newcommand{\RR}{\mathbb{R}}
\newcommand{\ZZ}{\mathbb{Z}}

\newcommand{\FF}{\mathbb{F}}

\newcommand{\III}{\mathcal{I}}

\newcommand{\SSS}{\mathbb{S}}
\newcommand{\BB}{\mathcal{B}}
\newcommand{\sss}{\mathfrak{s}}

\newcommand{\la}{\bm{\lambda}}
\newcommand{\al}{\bm{\alpha}}

\newcommand{\uu}{\underline}
\newcommand{\ms}{\mspace{2mu}}
\newcommand{\su}{\textnormal{supp}}
\newcommand{\vv}{\bm{v}}
\newcommand{\ww}{\bm{w}}

\newcommand{\sd}{\mathcal{S}}
\newcommand{\sal}{\SSS_{(a,l)}}

\newcommand{\lb}{\langle}
\newcommand{\rb}{\rangle}

\begin{document}

%%%%%%%%%%%%%%%%%%%%%%%%%%%%%%%%%%%%%%%%%%%%%%%%%%%

\title{On the Algebraic Structure of Linear Trellises}

\author{David~Conti %,~\IEEEmembership{Member,~IEEE,}
        and~Nigel~Boston%,~\IEEEmembership{Member,~IEEE}
\thanks{
%D. Conti 
%was supported by the Science Foundation Ireland grant 06/MI/006 (e-mail: david.conti@ucdconnect.ie). 
This work was part of the Ph.D. research of D. Conti (e-mail: david.conti@ucdconnect.ie), supported by the Science Foundation Ireland grant 06/MI/006, and supervised by N. Boston (e-mail: boston@math.wisc.edu).
%is with the Claude Shannon Institute, UCD, Ireland (e-mail: david.conti@ucdconnect.ie). 
%N. Boston is with the Departments of Mathematics and Electrical and Computer Engineering, UW, Madison, USA (e-mail: boston@math.wisc.edu). 
%This paper was presented in part at the 2012 International Zurich Seminar on Communications \cite{CB2} and the 2012 Allerton conference \cite{CB3}.}
Some parts of this paper were presented at the 2012 International Zurich Seminar on Communications \cite{CB2} and the 2012 Allerton conference \cite{CB3}.}
}
\date{}

\maketitle
\thispagestyle{plain}
%%%%%%%%%%%%%%%%%%%%%%%%%%%%%%%%%%%%%%%%%%%%%%%%%%%%

\begin{abstract}
Trellises are crucial graphical representations of codes. While conventional trellises are well understood, the general theory of (tail-biting) trellises is still under development. Iterative decoding concretely motivates such theory. 
In this paper we first develop a new algebraic framework for a systematic analysis of linear trellises which enables us to address open foundational questions. In particular, we present a useful and powerful characterization of linear trellis isomorphy. We also obtain a new proof of the Factorization Theorem of Koetter/Vardy  
and point out unnoticed problems for the group case. 

Next, we apply our work to: describe all the elementary trellis factorizations of linear trellises and consequently to determine all the minimal linear trellises for a given code; prove that nonmergeable one-to-one linear trellises are strikingly determined by the edge-label sequences of certain closed paths; prove self-duality theorems for minimal linear trellises; analyze quasi-cyclic linear trellises and consequently extend results on reduced linear trellises to nonreduced ones.  To achieve this, we also provide new insight into mergeability and path connectivity properties of linear trellises.

Our classification results are important for iterative decoding as we show that minimal linear trellises can yield different pseudocodewords even if they have the same graph structure.
\end{abstract}

%\begin{IEEEkeywords}
%Tail-biting trellises, linear trellises, linear block codes, factorizations of linear trellises, nonmergeable trellises, minimal trellises. 
%\end{IEEEkeywords}


\section{Introduction}
 
\textit{Trellis} representations of block codes play a prominent role in coding theory and practice as they provide combinatorial insight into algebraic codes and enable the design of efficient decoding. 

Traditionally trellises were separated in two classes, \textit{conventional trellises} (introduced in \cite{BCJR}) and \textit{tail-biting trellises} (introduced in \cite{SoTi}), however it is no harm to see conventional ones as a subclass of tail-biting ones, and we will do so. In fact for theoretical purposes it is convenient to do so (this point of view was already adopted for example in \cite{GW,GW2}). 

Also, the actual objects of study in trellis theory are \textit{linear trellises}, i.e. trellises with a linear structure (trellises without any algebraic structure are infeasible to control). The notion of linear trellis was formalized though only at a late stage, by Koetter/Vardy \cite{KV2}  (while McEliece \cite{Mc} acknowledged a bit earlier the linear structure of \textit{minimal conventional trellises}). 

The study of trellises was confined to conventional ones until late in the 90s. The interest for a general theory of (linear tail-biting) trellises surged along with the interest for suboptimal iterative decoding (sparked by Wiberg's thesis \cite{W} and the invention of Turbo Codes) as its complexity benefits from the long known fact proven in \cite{SoTi} that nonconventional representations can achieve smaller size (while optimal decoding does not, see for example \cite{RB}).  

The first works \cite{CFV,KV0,KV2,KV} towards such general theory provided a rigorous basis to the subject and had a strong influence on what came next. In particular, Koetter/Vardy \cite{KV2} considered the \textit{trellis product} operation (introduced first for conventional trellises in \cite{KsSo} and then extended to all trellises in \cite{CFV}) and proved that \underline{all} linear trellises factor as products of \textit{elementary trellises} (\textit{Factorization Theorem}), which can be more easily handled. This groundbreaking result enabled them to achieve in \cite{KV} breakthrough on the minimality problem (which is far more complicated in the general case than in the conventional case) by narrowing down the search for minimal linear trellises to computable \textit{characteristic sets} of elementary trellises, which inspired much of the subsequent research.

Steady subsequent research (e.g. \cite{B,FoGl2,GW,GW2,N,RB,SB,SB2,SKSR}) on the top of the seminal works has led to a fairly rich development of linear trellis theory. 

However, some important foundational questions have not been addressed, and as a consequence the problem of classifying minimal linear {trellis} representations has been addressed only partially. This problem is important not only for theoretical purposes but also for iterative/LP decoding (as we point out in Subsection \ref{pointout}). 

In this paper we build an algebraic framework that gives extra insight into linear trellises, answers such fundamental questions, and provides new algebraic tools that we apply to address the classification problem and more. Mathematically, we provide a thorough analysis of the monoid of linear trellises with trellis product.

%%%%%%%%%%%%%%%%%%%%%%%%%%%%%%%%%%%%%%%%%%%%%%%%%%%%%%%%%%%%%%%%%%%%%%%%%%%%%%%%%%%%%%%%%%%%%%%%%%%%%%%%%%%%%%%%%%%%%%%%%%%%%%%%%%%%%%%%%%%%%%%%%%%%%%
%%%%%%%%%%%%%%%%%%%%%%%%%%%%%%%%%%%%%%%%%%%%%%%%%%%%%%%%%%%%%%%%%%%%%%%%%%%%%%%%%%%%%%%%%%%%%%%%%%%%%%%%%%%%%%%%%%%%%%%%%%%%%%%%%%%%%%%%%%%%%%%%%%%%%%
%%%%%%%%%%%%%%%%%%%%%%%%%%%%%%%%%%%%%%%%%%%%%%%%%%%%%%%%%%%%%%%%%%%%%%%%%%%%%%%%%%%%%%%%%%%%%%%%%%%%%%%%%%%%%%%%%%%%%%%%%%%%%%%%%%%%%%%%%%%%%%%%%%%%%%

\subsection{Contents and contributions of the paper}

\begin{rmk}
{All trellises in Sections \ref{algframpap}, \ref{factpap} and Appendix \ref{graphcarpap} will be reduced}. 
\end{rmk}

%%%%%%%%%%%%%%%%%%%%%%%%%%%%%%%%%%%%%%%%%%%%%%%%%%%%%%%%%%%%%%%%%%%%%%%%%%%%%%%%%%%%%%%%%%%%%%%%%%%%%%%%%%%%%%%%%%%%%%%%%%%%%%%%%%%%%%%%%%%%%%%%%%%%%%

 {\textbf{Section II} (\textit{Preliminaries and basics on trellises}):
In this section we fix the notation/terminology and go over the necessary background
with the intent to make our
 treatment as self-contained as possible.
The reader versed in trellis theory may skim through this part, except for paying some attention to our notation for spans (given in Subsection \ref{eldef}).

%%%%%%%%%%%%%%%%%%%%%%%%%%%%%%%%%%%%%%%%%%%%%%%%%%%%%%%%%%%%%%%%%%%%%%%%%%%%%%%%%%%%%%%%%%%%%%%%%%%%%%%%%%%%%%%%%%%%%%%%%%%%%%%%%%%%%%%%%%%%%%%%%%%%%%

 \textbf{Section III} (\textit{Algebraic framework for linear trellises and new foundational insights}):
 In this section we introduce the core machinery and results on the top of which the rest of the paper is built on. The structure of linear trellises is analyzed from an algebraic perspective and related to the trellis product operation. Our methodology consists in studying  the \textit{label code} $\SSS(T)$ and showing how its properties correspond to properties of $T$. The key tools that we introduce are \textit{span subcodes} of $\SSS(T)$ (Subsection \ref{spansubpap}) and \textit{product bases} (Subsection \ref{prodbaspap}), which describe the structure of $\SSS(T)$. 
 
 The foundational paper of Koetter/Vardy \cite{KV2} on linear trellises lacks a framework that goes beyond proving the existence of elementary trellis factorizations. Our algebraic framework enables us to answer the following fundamental questions:

\begin{enumerate}
\item  How is the trellis product operation encoded by the label code $\SSS(T)$? 
\item When are two linear trellises \textit{linearly isomorphic} (i.e. equivalent)?
\item Is the linear structure of a linear trellis essentially unique? 
 Equivalently, if two linear trellises are isomorphic are they linearly isomorphic too? 
\end{enumerate}
  
In answering $1)$ we provide a new proof of the Factorization Theorem, which is simpler and more explanatory. In answering $2)$ we provide a fundamental characterization of linear trellis isomorphy (Theorem \ref{thm15}) that tells us how span subcodes discriminate linear trellises. This is the first criterion of such type in the literature and it is crucial in proving the main following results of the paper, in particular for trellis classification purposes (see Section \ref{charactpap}). It also allows us to answer $3)$ positively. So (for possible implementation purposes) there is no need to search for the best linear structure of a linear trellis, and to look for linear isomorphisms of linear trellises is equivalent to look for isomorphisms (which also resolves the confusion in the literature where some authors use the first notion while others use the second one). 
   
   At the end of this section we point out some important overlooked aspects of \textit{group trellises} by showing that the notion of \textit{elementary group trellis} given in \cite{KV2} is too strict for the Factorization Theorem to hold also in that case, and discuss a possible remedy.
 
%%%%%%%%%%%%%%%%%%%%%%%%%%%%%%%%%%%%%%%%%%%%%%%%%%%%%%%%%%%%%%%%%%%%%%%%%%%%%%%%%%%%%%%%%%%%%%%%%%%%%%%%%%%%%%%%%%%%%%%%%%%%%%%%%%%%%%%%%%%%%%%%%%%%%%

\textbf{Section IV} (\textit{Elementary trellis factorizations of linear trellises}):
In this section we apply the work of the previous one to answer the following fundamental questions:

\begin{enumerate}
\item How does a linear trellis $T$ determine its elementary trellis factorizations?
\item How can we compute (all the possible) elementary trellis factorization of $T$?
\end{enumerate}

In particular, we show that the span distribution of any such factorization of $T$ is unique (a striking fact which passed unnoticed in \cite{KV2}) and is only determined by the underlying graph structure of $T$. We also determine precise conditions under which $T$ has a unique elementary trellis factorization, and give a formula for the number of such factorizations for linear trellises with no repeated spans.

Besides the theoretical value of the above questions, to be able to compute elementary trellis factorizations of $T$ is important since such a factorization yields important data that we can use to check more easily whether certain properties hold or not for $T$. Note that there are important classes of linear trellises that are not presented as elementary trellis products, e.g. \textit{BCJR trellises} (see \cite{GW,GW2,N}).
Also, our results 
yield a method to find out (efficiently) whether two products of elementary trellises are equal or not. 

The work of this section is crucial for being able to classify and determine all the minimal linear trellises for a fixe code as these are precisely constructed as elementary trellis products. 
 
%%%%%%%%%%%%%%%%%%%%%%%%%%%%%%%%%%%%%%%%%%%%%%%%%%%%%%%%%%%%%%%%%%%%%%%%%%%%%%%%%%%%%%%%%%%%%%%%%%%%%%%%%%%%%%%%%%%%%%%%%%%%%%%%%%%%%%%%%%%%%%%%%%%%%%

 \textbf{Section V} (\textit{Insights into the nonmergeable property, classification of nonmergeable linear trellises through multicyles, and complete classification/computation of minimal linear trellises}):
It is known that \textit{biproper} nonconventional linear trellises may be \textit{mergeable}, even in the \textit{one-to-one} case (see the trellis depicted after Observation \ref{obs51}), while for conventional trellises this never happens (\cite{V}). No explanation of this phenomenon has been given so far though. In this section we present first a new characterization of the nonmergeable property for one-to-one linear trellises which explains when and how such trellises fail to satisfy that property, and reconciles the conventional and nonconventional cases. This characterization amounts to a one-to-one correspondence between all long enough paths and their edge-label sequences.

We then use this result along with Theorem \ref{thm15} 
to prove the striking fact that such trellises (in particular minimal linear trellises) are completely determined (and so classifiable) by their codes of edge-label sequences of closed paths of length greater than the trellis length (\textit{multicycles}). 
In other words, a trellis $T$ of length $n$ does not only represent a single code $C(T)$, but a sequence of codes $C(T), C^{2}(T), C^{3}(T),\ldots$, of respective length $n,2n,3n,\ldots$, and for nonmergeable linear one-to-one trellises $C^{2}(T)$ completely determines $T$. 

Next, we show how to determine/compute all the minimal linear trellises with same graph structure for a given code $C$ from the knowledge of its so-called  \textit{characteristic matrix} from \cite{KV}. In particular we give a formula for their number. Combining this with the results of \cite{KV}  yields then a method to compute and so efficiently classify  all the minimal linear trellises for a fixed code. Besides its theoretical value, we then discuss how to be able to do is very important for iterative (or LP) decoding applications as we show that two minimal linear trellises with the same graph structure for the same code can still yield drastically different \textit{pseudocodewords}, a striking phenomenon which was never observed before.

We also apply the above results to deduce some very interesting results on \textit{self-duality} of  linear trellises, e.g. we show that a \textit{KV-trellis} (and so in particular a minimal linear trellis) $T$ is self-dual if and only if $C^{2}(T)$ is. 
  
%%%%%%%%%%%%%%%%%%%%%%%%%%%%%%%%%%%%%%%%%%%%%%%%%%%%%%%%%%%%%%%%%%%%%%%%%%%%%%%%%%%%%%%%%%%%%%%%%%%%%%%%%%%%%%%%%%%%%%%%%%%%%%%%%%%%%%%%%%%%%%%%%%%%%

 \textbf{Section VI} (\textit{Factorizations and isomorphisms of quasi-cyclic linear trellises, and extension of results on reduced linear trellises to nonreduced ones}):  
In this section we further demonstrate the power of the framework developed in Section \ref{algframpap} by giving some other interesting applications. We first define \textit{quasi-cyclic trellises} and prove that in the linear case their elementary trellis factorizations and isomorphisms allow a quasi-cyclic version.  We use then this result to show that our Theorems \ref{isolin}, \ref{linstr}, and \ref{thm53} extend to nonreduced linear trellises as well. In order to achieve that we  also prove some independent interesting results on the structure of trellises, in particular, we show how the nonreduced case and the reduced case can be linked to each other by means of \textit{trellis covers} (which we define in the previous section in order to deal with duality questions), for which we provide some basic results. 
This link is very useful for extending results on reduced linear trellises to nonreduced ones.
Note that nonreduced trellises naturally arise by taking duals of reduced  trellises or wrapped fragments of quasi-cyclic trellises and it is thus worth to have results concerning them.

%%%%%%%%%%%%%%%%%%%%%%%%%%%%%%%%%%%%%%%%%%%%%%%%%%%%%%%%%%%%%%%%%%%%%%%%%%%%%%%%%%%%%%%%%%%%%%%%%%%%%%%%%%%%%%%%%%%%%%%%%%%%%%%%%%%%%%%%%%%%%%%%%%%%%

 \textbf{Appendix A} (\textit{Connectivity and graphical properties of linear trellises}):
In this appendix we prove some important facts on connectivity and path properties of linear trellises that have been overlooked and have not appeared in the literature, and which we make use of  in  the paper. In particular, we prove that a linear trellis is connected as a directed graph if and only if it is as an undirected graph, and a characterization of  the ``reduced'' property for connected linear trellises is given. 

%%%%%%%%%%%%%%%%%%%%%%%%%%%%%%%%%%%%%%%%%%%%%%%%%%%%%%%%%%%%%%%%%%%%%%%%%%%%%%%%%%%%%%%%%%%%%%%%%%%%%%%%%%%%%%%%%%%%%%%%%%%%%%%%%%%%%%%%%%%%%%%%%%%%%

 \textbf{Appendix B} (\textit{Graphical characterization of span distributions of linear trellises}):
In this appendix we present an alternative proof of the uniqueness of the span distribution of elementary factorizations of linear trellises and an alternative method to compute such distribution that are based on a graphical approach given by considering intersections of paths.

%%%%%%%%%%%%%%%%%%%%%%%%%%%%%%%%%%%%%%%%%%%%%%%%%%%%

\section{Preliminaries}
{\color{red}
 }

\subsection{Some general notation}\label{prelimi}
The ring of integers modulo $n$ is denoted by  $\ZZ_{n}$. Finite fields are denoted by 
 $\FF$. The cardinality of a set $S$ is denoted by $|S|$.  
 The \textit{support} of $\bm{v}=(v_{i})_{i\in\III}\in\prod_{i\in\III} V_{i}$ (where the $V_{i}$'s are vector spaces) is   $\su(\bm{v}):=\{i\in\III|v_{i}\neq0\}$. 
The coordinate indices of products $V_{0}\times\ldots\times V_{n-1}$ (e.g. $\FF^{n}$) will be seen as lying inside $\ZZ_{n}$ in order to perform modular operations on them. 
For $\bm{v}\in V_{0}\times\ldots\times V_{n-1}$ we denote by $\sigma$ the left cyclic shift given by 
$$\sigma(v_{0}v_{1}\ldots v_{n-1}):= v_{1}\ldots v_{n-1}v_{0}$$
We use angle brackets $\lb\rb$ to indicate the subspace generated by the elements of a vector space $V$ given within the same brackets. 
A sum $\sum_{i\in\mathcal{I}}V_{i}$ of a family of subspaces $\{V_{i}\}_{i\in\mathcal{I}}$ of a vector space $V$ is a \textit{direct sum} if given $\vv^{i}\in V_{i}$ such that $\sum_{i\in\mathcal{I}}\vv^{i}=0$ then $\vv^{i}=0$ for all $i$. In that case we write it also as  $\oplus_{i\in\mathcal{I}}V_{i}$.
\subsection{Basics on trellises}
A  (\textit{tail-biting})\textit{ trellis}\index{trellis}   of \textit{length $n$} over  $\FF$ is a directed graph $T=(\mathcal{V},\mathcal{E})$\nomenclature[t]{$T$}{trellis}  with $\FF$-labeled edges and a partition 
into \textit{vertex sets}
$\mathcal{V}=\sqcup_{i\in\ZZ_{n}}V_{i}(T)$ such that  any edge starting in $V_{i}(T)$\nomenclature[vti]{$V_{i}(T)$}{vertex set of $T$ at time $i$} 
must end in $V_{i+1}(T)$.
We assume that \textit{parallel edges}\index{parallel edges} (i.e. edges starting and ending at same vertices) must have different labels. 
 The set of edges of $T$ that start in  $V_{i}(T)$ and end in $ V_{i+1}(T)$ is thus a subset 
 $$E_{i}(T)\subseteq V_{i}(T)\times\FF\times V_{i+1}(T)$$
  We call it an \textit{edge set} of $T$. 
A \textit{linear trellis}\index{trellis!linear} over $\FF$ is a trellis $T$ over $\FF$ with an $\FF$-vector space structure on each $V_{i}(T)$ such that each $E_{i}(T)$ is a vector subspace of $V_{i}(T)\times\FF \times V_{i+1}(T)$.   
We assume that trellises are \textit{trim}\index{trellis!trim}, i.e. 
each vertex has an outgoing and an incoming edge.
If the edge-labels of $T$ are all equal then they are essentially irrelevant. In that case we call $T$ also an {\textit{unlabeled trellis}}\index{trellis!unlabeled} and simply assume that all edge-labels are equal to $0$.   If   $|V_{0}(T)|=1$ then $T$ is said to be \textit{conventional}\index{trellis!conventional}.
 
 A  \textit{subtrellis} of $T$ is a trim subgraph \index{subtrellis}  $T'\subseteq T$. If in addition $T$ is linear  and $V_{i}(T')$, $E_{i}(T')$ are vector subspaces respectively of $V_{i}(T)$, $E_{i}(T)$ for all $i$, then $T'$ is a \textit{linear subtrellis}\index{subtrellis!linear}, and we write $T'\leq T$. 
 The \textit{cyclic}  \textit{shift}\index{shift of trellis} of $T$ by $j\in\ZZ$ positions is 
the trellis $\sigma^{j}(T)$\nomenclature[sigt]{$\sigma^{j}(T)$}{shift of $T$} defined by  
\begin{align*}
V&_{i}(\sigma^{j}(T)):=V_{i+j}(T)\\
E&_{i}(\sigma^{j}(T)):=E_{i+j}(T)
\end{align*}
We write $\sigma(T)$ for $\sigma^{1}(T)$.

Trellises are usually visualized by diagrams like the one below  for a linear trellis over the binary field $\FF_{2}$. 
 Vertex sets are plotted vertically.  The leftmost vertices are identified with the rightmost ones, and make up $V_{0}(T)$.  
 The edge-labels are represented by full lines for $1$ and dashed lines for $0$. 
  The vector space structure of the $V_{i}(T)$'s is given by vertex-labels.
 All the examples in this paper will be for trellises over $\FF_{2}$, except in Subsection \ref{grouppap}. 
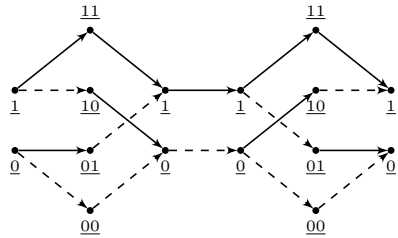
\begin{center}
 \begin{tikzpicture}[yscale=.8,>=latex',shorten >=.9pt, shorten <=1.4pt, line width=.6pt]
  \tikzstyle{every node}=[draw,circle,fill=black,minimum size=2pt,
                        inner sep=0pt]                    
\foreach \x in {0,1,2,3,4,5}{
\node at (\x,2) {};};
\foreach \x in {0,1,2,3,4,5}{
\node at (\x,1) {};};
\foreach \x in {1,4}{
\node at (\x,3) {};};
\foreach \x in {1,4}{
\node at (\x,-0) {};};
\tikzstyle{every node}=[]
                               %      \draw [] --  node [above] {${01}0{10}|{(3,3)}\otimes{10101}|{(0,4)}:$} (-6,2);

  %\draw ['] --  node [above] {${01010|{(3,1]}\otimes11111|{(0,4]}:}$} (-6,2.5);
\draw [->] (0,1) -- node[above] {} (1,1);
\draw [->, dashed] (1,1) -- node[pos=.05, above] {} (2,2);
\draw [->] (2,2) -- node[above] {} (3,2);
\draw [->] (3,2) -- node[above] {} (4,3);
\draw [->, dashed] (3,2) -- node[pos=.05, below] {} (4,1);
\draw [->] (4,3) -- node[above] {} (5,2);
\draw [->] (4,1) -- node[above] {} (5,1);
\draw [->] (0,2) -- node[above] {} (1,3);
\draw [->] (1,3) -- node[above] {} (2,2);
\draw [->, dashed] (0,2) -- node[above] {} (1,2);
\draw [->] (1,2) -- node[pos=.05, below] {} (2,1);
\draw [->, dashed] (0,1) -- node[above] {} (1,0);
\draw [->, dashed] (1,0) -- node[above] {} (2,1);
\draw [->, dashed] (2,1) -- node[above] {} (3,1);
\draw [->, dashed] (3,1) -- node[above] {} (4,0);
\draw [->, dashed] (4,0) -- node[above] {} (5,1);
\draw [->] (3,1) -- node[pos=.05, above] {} (4,2);
\draw [->, dashed] (4,2) -- node[above] {} (5,2);

\draw (0,2) node [below] {\tiny$\uu{1}$};
\draw (1,2) node [below] {\tiny$\uu{10}$};
\draw (2,2) node [below] {\tiny$\uu{1}$};
\draw (3,2) node [below] {\tiny$\uu{1}$};
\draw (4,2) node [below] {\tiny$\uu{10}$};
\draw (5,2) node [below] {\tiny$\uu{1}$};
\draw (0,1) node [below] {\tiny$\uu{0}$};
\draw (1,1) node [below] {\tiny$\uu{01}$};
\draw (2,1) node [below] {\tiny$\uu{0}$};
\draw (3,1) node [below] {\tiny$\uu{0}$};
\draw (4,1) node [below] {\tiny$\uu{01}$};
\draw (5,1) node [below] {\tiny$\uu{0}$};

\draw (1,0) node [below] {\tiny$\uu{00}$};
\draw (4,0) node [below] {\tiny$\uu{00}$};
\draw (1,3) node [above] {\tiny$\uu{11}$};
\draw (4,3) node [above] {\tiny$\uu{11}$};
%\draw (2,2) node [above] {\tiny\uu{1}};
   \end{tikzpicture}
      \end{center}   

\subsubsection{\textbf{Morphisms}}
A \textit{morphism of trellises}\index{trellis!morphism} $f:T \rightarrow T'$ is a collection of maps $f_{i}:V_{i}(T)\rightarrow V_{i}(T')$, $i\in\ZZ_{n}$, such that 
$$v\alpha w\in E_{i}(T)\Longrightarrow f_{i}(v)\alpha f_{i+1}(w)\in E_{i}(T')$$%. 
for all $i$. If the $f_{i}$'s are bijective and 
$$v\alpha w\in E_{i}(T)\Longleftrightarrow f_{i}(v)\alpha f_{i+1}(w)\in E_{i}(T')$$ 
for all $i$, we say that $f$ is an \textit{isomorphism} and that  $T$ and $T'$ are \textit{isomorphic}. In that case we  write $T\sim T'$\nomenclature[tziso]{$T\sim T'$}{isomorphic trellises}. \underline{Isomorphic trellises} \underline{must be regarded as equal} (since by renaming their vertices they are exactly the same trellis).
If $T$, $T'$, and the $f_{i}$'s are all linear we say that $f$ is \textit{linear}. If $f$ is a linear isomorphism then we say that $T$ and $T'$ are \textit{linearly isomorphic}, and write $T\simeq T'$. 
Two trellises $T,  T'$ are said to be \textit{structurally isomorphic}\index{structurally isomorphic trellises} if there exist bijective maps $f_{i}:V_{i}(T)\rightarrow V_{i}(T')$, $i\in\ZZ_{n}$, such that 
 for all $v\in V_{i}(T)$, $w\in V_{i+1}(T)$, 
 there are as many edges from $v$ to $w$ as from $f_{i}(v)$ to $f_{i+1}(w)$,
 that is,
 by forgetting all edge-labels $T$ and $T'$  have the same (ordered) graph structure.

\subsubsection{\textbf{Paths}}\label{special}
A (\textit{directed}) \textit{path}\index{path} $\bm{p}$ of length $m$ of a trellis $T$ is an ordered sequence
 $$v_{0}\alpha_{0}v_{1}\alpha_{1}\ldots v_{m-1}\alpha_{m-1}v_{m}$$
  such that $v_{j}\alpha_{j}v_{j+1}$ is an edge of $T$ for all $j=0,\ldots,m-1$. 
  The path is \textit{closed} if $v_{0}=v_{m}$. 
 We put 
\begin{align*}
L(\bm{p})&:=\alpha_{0}\ldots\alpha_{m-1}\\
\nu_{j}(\bm{p})&:=v_{j}
\end{align*} 
   By deinterleaving vertices and edge-labels, paths of length $m$ starting in $V_{i}(T)$ can be seen as vectors of $
\prod_{j=i}^{i+m}V_{j}(T)  
  \times \FF^{m}$.  
  If $v,w\in V_{i}(T)$ 
we denote by $\mathbb{P}(v,w)$\nomenclature[pvw]{$\mathbb{P}(v,w)$}{set of paths in $T$ from $v$ to $w$ of length $n$} the set of paths in $T$ from $v$ to $w$ of same length as $T$.
If $T$ is linear and all the $v_{j}$'s and $\alpha_{j}$'s are zero then $\bm{p}$ is a \textit{zero path}. The zero element of $V_{i}(T)$ will be  denoted also by $0_{i}$.

\subsubsection{\textbf{Cycles and associated codes}}
Let $n$ be the length of $T$.   
A \textit{cycle}\index{cycle} $\la$ of $T$ is a closed path of $T$ of length $n$  
starting in $V_{0}(T)$.  Cycles can be written as sequences $v_{0}\alpha_{0}\ldots v_{n-1}\alpha_{n-1}$ or (by deinterleaving) as pairs $(\vv,\al)\in\prod_{i\in\ZZ_{n}}V_{i}(T)\times \FF^{n}$. 
The \textit{label code}\index{label code} of $T$ is 
\begin{equation*}
\SSS(T):=\{\la|\la \textnormal{ is a cycle of T}\}
\end{equation*} 
 The \textit{code represented}\index{code represented by $T$} by $T$ is 
 \begin{equation*}
 C(T):=L(\SSS(T))\subseteq \FF^{n}
 \end{equation*}
  We say that $T$ is a \textit{trellis for} $C(T)$. 
If $T$ is linear so are $\SSS(T)$ and $C(T)$.
If each vertex of $T$ belongs to some cycle we say that $T$ is \textit{almost reduced}\index{trellis!almost reduced}.  If also each edge belongs to some cycle of $T$ we say that $T$ is \textit{reduced}\index{trellis!reduced}. 
A reduced trellis is identified by its label code. 
The \textit{trellis $T(S)$ spanned}\index{trellis spanned (or generated) by $S$}  by a subset $S\subseteq \SSS(T)$  
is the (largest) subtrellis of $T$ covered by the cycles in $S$.
The \textit{shift map} $$\sigma^{j}:\SSS(T)\rightarrow \SSS(\sigma^{j}(T))$$ is defined by $\sigma^{j}((\vv,\al)):=(\sigma^{j}(\vv),\sigma^{j}(\al))$.

\begin{rmk}
{While much} of our terminology goes back to \cite{KV2,KV}, within behavioral system theory  trellises can be seen also as dynamical systems (cf. \cite{FoGl, FoGl2}). %NOEXTRAREFERENCES
So it makes sense to think of: the indexing set $\ZZ_{n}$ as a (circular) time axis; indices as \textit{time indices}\index{time index}; vertices as  \textit{states}. Also, in such terminology $\SSS(T)$ is the \textit{``behavior''} of $T$, $E_{i}(T)$ is a \textit{``constraint code''}, and ``almost reduced'' becomes  \textit{``state-trim''}.

\end{rmk}

%%%%%%%%%%%%%%%%%%%%%%%%%%%%%%%%%%%%%%%%%%%%%%%%%%%%%%%%%%%%%%%%%%%%%%%%%%%%%%%%%%%%%%%%%%%%%%%%%%%%%%%%%%%%%%%%%%%%%%%%%%%%%%%%%%%%%%%%%%%%%%%%%%%%%%%%%%%%%%%%%%%%
%%%%%%%%%%%%%%%%%%%%%%%%%%%%%%%%%%%%%%%%%%%%%%%%%%%%%%%%

\subsection{Classes of trellises}
Let $T$ be a trellis. Then $T$ is:
\begin{itemize}
\item  \textit{connected}\index{trellis!connected} if for any vertices $v\neq w$  there exists a path from $v$ to $w$ (in Appendix \ref{trellisconnect} we show that \textbf{a linear trellis is connected if and only if it is connected as an undirected graph})
\item \textit{one-to-one}\index{trellis!one-to-one} if $L:\SSS(T)\rightarrow C(T)$ is injective
\item \textit{biproper}\index{trellis!biproper} if different edges with same label never start or end in the same vertex
\item  \textit{mergeable}\index{trellis!mergeable} if there exist vertices $v\neq w$ in $V_{i}(T)$ for some $i$ such that the trellis resulting from merging $v$ with $w$ represents the same code as $T$, otherwise \textit{nonmergeable}\index{trellis!nonmergeable} 
\end{itemize}
Clearly, a nonmergeable trellis is connected, and a biproper conventional trellis is one-to-one (the ``conventional'' hypothesis is necessary here).
\subsubsection{\textbf{Minimal trellises}}Given trellises $T,T'$  of same length, we say that $T'$ is \textit{smaller} than $T$ if $|V_{i}(T')| \leq|V_{i}(T)|$ for all $i$, with at least one strict inequality. If there exists no $T'$ smaller than $T$ such that $C(T')=C(T)$ then we say that $T$ is a \textit{minimal trellis}\index{trellis!minimal} (for $C(T)$). While this is the principal notion of trellis size/minimality, other notions   can be given, especially by refining this one  (see \cite{KV}). 
 However all notions of minimality coincide for conventional linear trellises, as the following holds:

\begin{thm}[Minimal Conventional Trellis \cite{M,KsV,V}]\label{minconpap}
A linear code has a unique minimal conventional trellis representation (up to isomorphism). Such trellis is linear and minimizes all $|V_{i}(T)|$ and $|E_{i}(T)|$ simultaneously.
\end{thm}

We denote the minimal conventional trellis for a linear code $C$  by $T^{*}(C)$. For nonconventional trellises the situation is not as easy, and we will go back to that in Section \ref{charactpap}.
For much on  minimal conventional trellises see the comprehensive survey \cite{V}, while for the nonconventional case see \cite{BJ,CFV,KV,RB,SB,SKSR}.

\begin{rmk}\label{nolinminpap}
If a linear trellis $T$ is minimal amongst all linear trellises is it so also amongst \underline{all trellises}?  By Theorem \ref{minconpap} if $T$ is conventional then the answer is yes. But in general the answer is not clear (a nonconventional minimal trellis for a linear code may be not linear, see \cite{KV}). This question is overlooked in the literature as only  ``linear minimality'' is really dealt with. 
Indeed there is no systematic way to construct/control nonlinear trellis representations, while linearity makes things feasible. 
We will stay in this better world, so 
 \textbf{by a minimal linear trellis  
for $C$ we will always mean a linear trellis which is minimal amongst all linear trellises for $C$}. 
Note also that the usual assumption in the literature that minimal linear trellises  be reduced    is redundant (see Theorem \ref{minred}). 
\end{rmk}

%%%%%%%%%%%%%%%%%%%%%%%%%%%%%%%%%%%%%%%%%%%%%%%%%%%%%%%%
%%%%%%%%%%%%%%%%%%%%%%%%%%%%%%%%%%%%%%%%%%%%%%%%%%%%%%%%
%%%%%%%%%%%%%%%%%%%%%%%%%%%%%%%%%%%%%%%%%%%%%%%%%%%%%%%%
%%%%%%%%%%%%%%%%%%%%%%%%%%%%%%%%%%%%%%%%%%%%%%%%%%%%%%%%

\subsubsection{\textbf{Self-dual trellises}}
We say that a linear trellis $T$ is \textit{self-dual} if $T\sim T^{\perp}$. Here $T^{\perp}$ is the dual trellis of $T$ as defined in \cite{Fonorm}, which satisfies $C(T^{\perp})=(C(T))^{\perp}$. For $\FF=\FF_{2}$, $T^{\perp}$ is given by  
\begin{align*}
V_{i}(T^{\perp})&:= V_{i}(T)\equiv \FF_{2}^{r_{i}}\\
E_{i}(T^{\perp})&:=(E_{i}(T))^{\perp}
\end{align*}
 where the dual of $E_{i}(T)\leq \FF_{2}^{r_{i}+1+r_{i+1}}$ is with respect to the standard scalar product. The definitions of dual trellis given in \cite{N} and \cite{KV} coincide with the above one when $T$ is minimal. Note that $T^{\perp}$ may be not trim, even if $T$ is reduced. However, $T^{\perp}$ is minimal  (and so reduced) if and only and if $T$ is. See also \cite{GW2} for more on trellis dualization.

%%%%%%%%%%%%%%%%%%%%%%%%%%%%%%%%%%%%%%%%%%%%%%%%%%%%%%%
%%%%%%%%%%%%%%%%%%%%%%%%%%%%%%%%%%%%%%%%%%%%%%%%%%%%%%%%%%%%%%%%%%%%%%%%%%%%%%%%%%%%%%%%%%%%%%%%%%%%%%%%%%%%%%
%%%%%%%%%%%%%%%%%%%%%%%%%%%%%%%%%%%%%%%%%%%%%%%%%%%%%%
%%%%%%%%%%%%%%%%%%%%%%%%%%%%%%%%%%%%%%%%%%%%%%%%%%%%%%
%%%%%%%%%%%%%%%%%%%%%%%%%%%%%%%%%%%%%%%%%%%%%%%%%%%%%%
%%%%%%%%%%%%%%%%%%%%%%%%%%%%%%%%%%%%%%%%%%%%%%%%%%%%%%
%%%%%
\subsection{Spans}\label{eldef}

 Let $n\geq0$ be fixed.
 For $a,b\in\ZZ_{n}$ we put 
 \begin{equation*}
 [a,b]:=\{a,a+1,\ldots,a+{l}\}\subseteq\ZZ_{n}
 \end{equation*}
where $l=\min\{l\geq0|a+{l}\equiv b\mod n\}$. This is a (\textit{circular}) \textit{interval}\index{interval of $\ZZ_{n}$} of $\ZZ_{n}$. We also put $(a,b]:=[a,b]\setminus\{a\}$. 
Let now  $\prod_{i\in\ZZ_{n}}V_{i}$ be a product  of vector spaces which we think of as alphabets. 
Given $a\in\ZZ_{n}$, $0\leq l\leq n-1$, and $\bm{v}\in  \prod_{i\in\ZZ_{n}}V_{i}$,  we say that the pair $(a,l)$\nomenclature{$(a,l)$}{span, page (see also page )} is  a \textit{span} of $\bm{v}$ of \textit{length} $l$ if $\textnormal{supp}(\bm{v})\subseteq [a,a+{l}]$.
The \textit{starting} and \textit{ending point} of $(a,l)$ are respectively $a$ and $a+{l}$. 
We also say that $\emptyset$ is a span of length $-1$ of $\bm{0}\in V$ and dually that $\ZZ_{n}$ is a span of length $n$ of all $\bm{v}\in V$. No 
starting and ending point are associated to $\emptyset$ and $\ZZ_{n}$, so we also say that these spans are \textit{degenerate}.
Nevertheless we will use respectively the notation $(a,-1)$, $(a,n)$ for the spans $\emptyset$, $\ZZ_{n}$ too.
A generic span will be also simply denoted by the letter $\sss$ when we do not need to specify its starting point and length.
 If no alphabets are specified when talking of spans of vectors in $\FF^{n}$ we tacitly assume  
that they are all one-dimensional.

\begin{rmk}
In the classical terminology one calls $[a,b]$  a span of $\vv$ if $\textnormal{supp}(\bm{v})\subseteq[a,b]$. However in tail-biting trellis theory we want to distinguish between spans with different starting and ending points, and so for $c\neq d$ we want $[c,c-1]$ and $[d,d-1]$ to be different spans, which clashes with $[c,c-1]=[d,d-1]$. Similarly, in \cite{KV} and subsequent works $(a,b]$ is called a span of $\vv$ if $\textnormal{supp}(\bm{v})\subseteq[a,b]$,
leading for example to $(0,0]$ being a span of $100$ and $(1,1]$ being not, while $(1,1]=(0,0]$.
Our terminology avoids such formal abuses while remaining compatible with the literature. 
In fact a span should still be thought of as an interval. The notation $(a,l)$ rather serves  to parametrize spans and to discriminate between them. It is also convenient for proofs by induction on span length. 
\end{rmk}
 Spans have a natural partial order: we put 
 $$(a_{1},l_{1})\leq(a_{2},l_{2})$$  if and only if $(l_{1}\leq l_{2}< n-1$ and $[a_{1},a_{1}+l_{1}]\subseteq[a_{2},a_{2}+l_{2}])$, or $(l_{2}=n-1$ and  $(a_{1},a_{1}+l_{1}]\subseteq(a_{2},a_{2}+l_{2}])$, or $l_{1}=-1$, or $l_{2}=n$.
 This yields the following Hasse diagram.
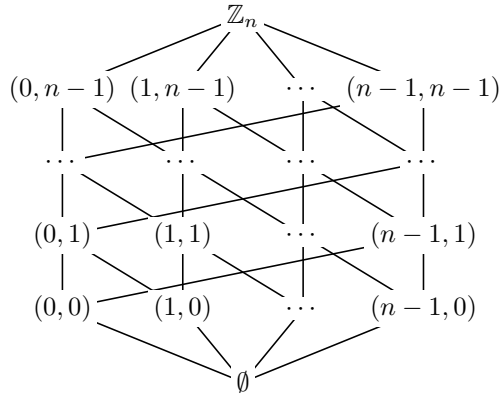
\begin{center}
 \begin{tikzpicture}[yscale=.65,xscale=.8, line width=.6pt]
                                      \tikzstyle{every node}=[]
                                     \foreach \a in {0,1,2,3}{
                        \foreach \l in {0,1,2}{
                       \draw  (2*\a,1.5*\l) -- (2*\a,1.5*\l+1.5);};};
                        \foreach \a in {1,2,3}{
                        \foreach \l in {0,1,2}{
                       \draw  (2*\a,1.5*\l) -- (2*\a-2,1.5*\l+1.5);};};
                        \foreach \l in {0,1,2}{
                       \draw  (0,1.5*\l) -- (6,1.5*\l+1.5);};
                                \foreach \a in {0,1,2,3}{
                       \draw  (2*\a,0) -- (3,-1.5);
                       \draw  (2*\a,4.5) -- (3,6);};
            
             \tikzstyle{every node}=[fill=white,minimum size=0pt,
                        inner sep=2pt]\small
          
                        \foreach \a in {0,1}{
                        \foreach \l in {0,1}{
    \draw  (2*\a,1.5*\l) node {$(\a,\l)$};};};
         
                        \foreach \l in {0,1}{
    %\draw  (6,1.5*\l) node {$(n-2,\l)$};
        \draw  (6,1.5*\l) node {$(n-1,\l)$};
        };
    
                     \foreach \a in {0,1}{
  %  \draw  (2*\a,4.5) node {$(\a,n-2)$};
     \draw  (2*\a,4.5) node {$(\a,n-1)$};};
     
  %   \draw (6,4.5) node {$(n-2,n-2)$};
 %  \draw (8,4.5) node {$(n-1,n-2)$};
        \draw (6,4.5) node {$(n-1,n-1)$};
      %              \draw (6,6) node {$(n-2,n-1)$};
 \draw (0,3) node {$\cdots$};
 \draw (2,3) node {$\cdots$};   
  \draw (6,3) node {$\cdots$};   
 
                           \foreach \a in {2}{
                        \foreach \l in {0,1,2,3}{
                       \draw  (2*\a,1.5*\l) node {$\cdots$};};};

    \draw (3,-1.5) node {$\emptyset$};
     \draw (3,6) node {$\ZZ_{n}$};
    
%    
%                  \tikzstyle{every node}=[fill=white,minimum size=1pt,
%                        inner sep=2pt]
%                          \draw  (0,0) -- (0,1);
%                        \foreach \a in {0,1}{
%                        \foreach \l in {0,1,2}{
%    \draw  (2*\a,1.5*\l) node {$\SSS_{(\a,\l)}(T)$};};};
%    \draw (4,-1.5) node {$\SSS_{\emptyset}(T)$};
%     \draw (4,4.5) node {$\SSS_{\ZZ_{n}}(T)$};
%     
%          \foreach \a in {0,1,3,4}{
%                        \foreach \l in {0,1,2}{
%    \draw  (6,1.5*\l) node {$\SSS_{(n-2,\l)}(T)$};};};
%    \draw (4,-1.5) node {$\SSS_{\emptyset}(T)$};
%     \draw (4,4.5) node {$\SSS_{\ZZ_{n}}(T)$};
%    

%\draw (1.1,3) node [below] {\tiny ${\uu{10}}$};
%\draw (1.9,3) node [below] {\tiny ${\uu{10}}$};
%\draw (1,4) node [above] {\tiny ${\uu{11}}$};
%\draw (2,4) node [above] {\tiny ${\uu{11}}$};
      \end{tikzpicture}    
    \end{center}
  If $(a_{1},l_{1})\leq(a_{2},l_{2})$ we also say that $(a_{1},l_{1})$ is \textit{contained} in $(a_{2},l_{2})$.
   A span $(a,l)$ is said to be \textit{conventional}\index{span!linear} if $(a,l)\leq(0,n-1)$.

\subsection{Elementary trellises}\label{elemtrelliscor}
Let $(a,l)$ be a span of $\al\in\FF^{n}$, with $0\leq l \leq n$. We denote by $\al|(a,l)$ the \textit{elementary trellis for $\al$ of span $(a,l)$}\index{trellis!elementary}\nomenclature[alp]{$\al|(a,l)$}{elementary trellis}. This is defined as follows:  
\begin{itemize}
\item $V_{i}(\al|(a,l)):=0$ for all $i\in\ZZ_{n}\setminus (a,a+l]$
\item $V_{i}(\al|(a,l)):=\FF$ for all $i\in (a,a+l]$
\item $E_{i}(\al|(a,l)):=\lb(v_{i},\alpha_{i},v_{i+1})\rb$ for all $i$, where $v_{i}:=0$ if $i\in\ZZ_{n}\setminus (a,a+l]$ and $v_{i}:=1$ otherwise
\end{itemize}
Clearly $\al|(a,l)$ is: reduced and linear; minimal if and only if $(a,l)$ is a minimal span of $\al$; one-to-one if and only if $\al\neq0$. 
Also, a
linear trellis $T$ is (isomorphic to) an elementary trellis if and only if $\dim\SSS(T)=1$. This justifies the adjective ``elementary''. 

Note that elementary trellis graph structures correspond to spans, thanks to our terminology. In fact the graph structure of $\al|(a,l)$ does not depend on $\al$.

\begin{ex} All the possible elementary trellises for $10$ are given by
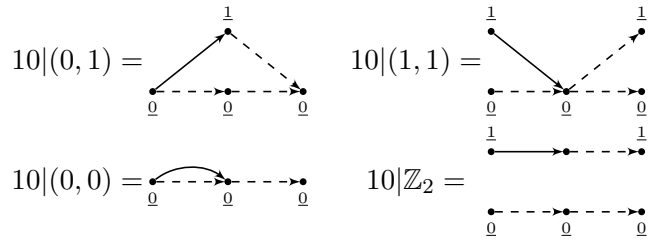
\begin{center}%yscale=.64,xscale=.8
 \begin{tikzpicture}[yscale=.8,>=latex',shorten >=.9pt, shorten <=1.4pt, line width=.6pt]
  \tikzstyle{every node}=[draw,circle,fill=black,minimum size=2pt,
                        inner sep=0pt]                    
\foreach \x in {0,1,2}{
\node at (\x,0) {};};
\foreach \x in {1}{
\node at (\x,1) {};};

\tikzstyle{every node}=[]
\draw [->,dashed] (0,0) -- (1,0);
\draw [->,dashed] (1,0) -- (2,0);

\draw [->] (0,0) -- (1,1);
\draw [->,dashed] (1,1) -- (2,0);

\draw (0,0) node [below] {\tiny ${\uu{0}}$};
\draw (1,0) node [below] {\tiny ${\uu{0}}$};
\draw (2,0) node [below] {\tiny ${\uu{0}}$};

\draw (1,1) node [above] {\tiny ${\uu{1}}$};
\draw (-1,.5) node {$10|(0,1)=$};
\begin{scope}
[xshift=4.5cm]
  \tikzstyle{every node}=[draw,circle,fill=black,minimum size=2pt,
                        inner sep=0pt]                    
\foreach \x in {0,1,2}{
\node at (\x,0) {};};
\foreach \x in {0,2}{
\node at (\x,1) {};};

\tikzstyle{every node}=[]
\draw [->,dashed] (0,0) -- (1,0);
\draw [->,dashed] (1,0) -- (2,0);

\draw [->] (0,1) -- (1,0);
\draw [->,dashed] (1,0) -- (2,1);

\draw (0,0) node [below] {\tiny ${\uu{0}}$};
\draw (1,0) node [below] {\tiny ${\uu{0}}$};
\draw (2,0) node [below] {\tiny ${\uu{0}}$};

\draw (0,1) node [above] {\tiny ${\uu{1}}$};
\draw (2,1) node [above] {\tiny ${\uu{1}}$};
\draw (-1,.5) node {$10|(1,1)=$};
\end{scope}
\begin{scope}
[yshift=-1.5cm]
  \tikzstyle{every node}=[draw,circle,fill=black,minimum size=2pt,
                        inner sep=0pt]                    
\foreach \x in {0,1,2}{
\node at (\x,0) {};};

\tikzstyle{every node}=[]
\draw [->,dashed] (0,0) -- (1,0);
\draw [->,dashed] (1,0) -- (2,0);

\draw [->] (0,0) .. controls (.3,.3) and (.7,.3) .. (1,0);

\draw (0,0) node [below] {\tiny ${\uu{0}}$};
\draw (1,0) node [below] {\tiny ${\uu{0}}$};
\draw (2,0) node [below] {\tiny ${\uu{0}}$};

\draw (-1,0) node {$10|(0,0)=$};
\end{scope}
\begin{scope}
[xshift=4.5cm,yshift=-2cm]
  \tikzstyle{every node}=[draw,circle,fill=black,minimum size=2pt,
                        inner sep=0pt]                    
\foreach \x in {0,1,2}{
\node at (\x,0) {};};
\foreach \x in {0,1,2}{
\node at (\x,1) {};};

\tikzstyle{every node}=[]
\draw [->,dashed] (0,0) -- (1,0);
\draw [->,dashed] (1,0) -- (2,0);

\draw [->] (0,1) -- (1,1);
\draw [->,dashed] (1,1) -- (2,1);

\draw (0,0) node [below] {\tiny ${\uu{0}}$};
\draw (1,0) node [below] {\tiny ${\uu{0}}$};
\draw (2,0) node [below] {\tiny ${\uu{0}}$};

\draw (0,1) node [above] {\tiny ${\uu{1}}$};
\draw (2,1) node [above] {\tiny ${\uu{1}}$};
\draw (-1,.5) node {$10|\ZZ_{2}=$};
\end{scope}
      \end{tikzpicture}
    \end{center}
Two more elementary trellis graph structures of length 2 are missing, namely those corresponding to the spans $(1,0)$ and $\emptyset$ (which are not spans of $10$). These can be given by
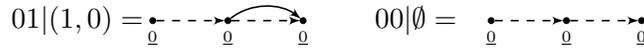
\begin{center}%yscale=.64,xscale=.8
 \begin{tikzpicture}[yscale=.8,>=latex',shorten >=.9pt, shorten <=1.4pt, line width=.6pt]
  \tikzstyle{every node}=[draw,circle,fill=black,minimum size=2pt,
                        inner sep=0pt]                    
  \tikzstyle{every node}=[draw,circle,fill=black,minimum size=2pt,
                        inner sep=0pt]                    
\foreach \x in {0,1,2}{
\node at (\x,0) {};};

\tikzstyle{every node}=[]
\draw [->,dashed] (0,0) -- (1,0);
\draw [->,dashed] (1,0) -- (2,0);

\draw [->] (1,0) .. controls (1.3,.3) and (1.7,.3) .. (2,0);

\draw (0,0) node [below] {\tiny ${\uu{0}}$};
\draw (1,0) node [below] {\tiny ${\uu{0}}$};
\draw (2,0) node [below] {\tiny ${\uu{0}}$};

\draw (-1,0) node {$01|(1,0)=$};
\begin{scope}
[xshift=4.5cm]
  \tikzstyle{every node}=[draw,circle,fill=black,minimum size=2pt,
                        inner sep=0pt]                    
\foreach \x in {0,1,2}{
\node at (\x,0) {};};

\tikzstyle{every node}=[]
\draw [->,dashed] (0,0) -- (1,0);
\draw [->,dashed] (1,0) -- (2,0);

\draw (0,0) node [below] {\tiny ${\uu{0}}$};
\draw (1,0) node [below] {\tiny ${\uu{0}}$};
\draw (2,0) node [below] {\tiny ${\uu{0}}$};

\draw (-1,0) node {$00|\emptyset=$};
\end{scope}
      \end{tikzpicture}
    \end{center}
\end{ex}

\subsection{Trellis product}
Let $T$ and $T'$ be trellises of same length. The \textit{trellis product}\index{trellis!product} $T\otimes T'$\nomenclature[tzzpr]{$T\otimes T'$}{trellis product of $T$ and $T'$} is
given by
\begin{gather*}
V_{i}(T\otimes T'):=V_{i}(T)\times V_{i}(T')\\
E_{i}(T\otimes T'):=\\
\{vv'(\alpha+\alpha')ww'| (v\alpha w, v'\alpha' w')\in E_{i}(T)\times E_{i}(T')\}
\end{gather*} for all $i$.
Below is an easy example:
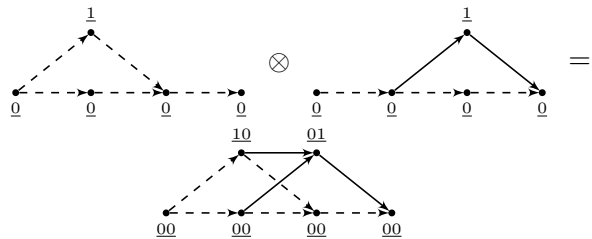
\begin{center}
  \begin{tikzpicture}[yscale=.8,>=latex',shorten >=.9pt, shorten <=1.4pt, line width=.6pt]
  \begin{scope}
 [yshift=-2cm,xshift=2cm]

  \tikzstyle{every node}=[draw,circle,fill=black,minimum size=2pt,
                        inner sep=0pt]                    
\foreach \x in {1,2}{
\node at (\x,2) {};};
\foreach \x in {0,1,2,3}{
\node at (\x,1) {};};
\tikzstyle{every node}=[]
  \draw [->,dashed] (0,1) -- (1,2);
\draw [->,dashed] (1,1) -- (2,1);
\draw [->,dashed] (2,1) -- (3,1);
  \draw [->,dashed] (0,1) -- (1,1);
\draw [->] (1,2) -- (2,2);
\draw [->] (2,2) -- (3,1);
\draw [->] (1,1) -- (2,2);
\draw [->,dashed] (1,2) -- (2,1);
\draw (0,1) node [below] {\tiny ${\uu{00}}$};
\draw (1,1) node [below] {\tiny ${\uu{00}}$};
\draw (2,1) node [below] {\tiny ${\uu{00}}$};
\draw (3,1) node [below] {\tiny ${\uu{00}}$};
\draw (1,2) node [above] {\tiny ${\uu{10}}$};
\draw (2,2) node [above] {\tiny ${\uu{01}}$};
\end{scope}
\begin{scope}
[xshift=0cm]
  
  \tikzstyle{every node}=[draw,circle,fill=black,minimum size=2pt,
                        inner sep=0pt]                    
\foreach \x in {1}{
\node at (\x,2) {};};
\foreach \x in {0,1,2,3}{
\node at (\x,1) {};};
\tikzstyle{every node}=[]
  \draw [->,dashed] (0,1) -- (1,2);
\draw [->,dashed] (1,1) -- (2,1);
\draw [->,dashed] (2,1) -- (3,1);
  \draw [->,dashed] (0,1) -- (1,1);
\draw [->,dashed] (1,2) -- (2,1);
\draw (0,1) node [below] {\tiny ${\uu{0}}$};
\draw (1,1) node [below] {\tiny ${\uu{0}}$};
\draw (2,1) node [below] {\tiny ${\uu{0}}$};
\draw (3,1) node [below] {\tiny ${\uu{0}}$};
\draw (1,2) node [above] {\tiny ${\uu{1}}$};
\end{scope}

\begin{scope}
[xshift=4cm]
    \tikzstyle{every node}=[draw,circle,fill=black,minimum size=2pt,
                        inner sep=0pt]                    
\foreach \x in {2}{
\node at (\x,2) {};};
\foreach \x in {0,1,2,3}{
\node at (\x,1) {};};
\tikzstyle{every node}=[]

\draw [->,dashed] (1,1) -- (2,1);
\draw [->,dashed] (2,1) -- (3,1);
  \draw [->,dashed] (0,1) -- (1,1);

\draw [->] (2,2) -- (3,1);
\draw [->] (1,1) -- (2,2);

\draw (0,1) node [below] {\tiny ${\uu{0}}$};
\draw (1,1) node [below] {\tiny ${\uu{0}}$};
\draw (2,1) node [below] {\tiny ${\uu{0}}$};
\draw (3,1) node [below] {\tiny ${\uu{0}}$};
\draw (2,2) node [above] {\tiny ${\uu{1}}$};
\end{scope}

\draw (3.5,1.5) node {$\otimes$};
\draw (7.5,1.5) node {$=$};
  \end{tikzpicture}
          \end{center}

The trellis product has the following basic properties: 
\begin{enumerate}
\item $T\otimes T'\sim T'\otimes T$ (commutativity)
\item $T\otimes (T'\otimes T'')\sim(T\otimes T')\otimes T''$ (associativity)
\item $0\otimes T\sim T$, where $0$ is the \textit{zero trellis} given by $V_{i}(T):=0$ and $E_{i}(T):=0$ for all $i$ (identity element)
\item $T'\sim T''\implies T\otimes T'\sim T\otimes T''$
\item if $T$ and $T'$ are linear/reduced/one-to-one so is $T\otimes T'$
\item we can replace $\sim$ with $\simeq$ in 1), 2), 3), and 4) when the trellises involved are linear
\end{enumerate}
 It also satisfies the crucial identity 
 $$C(T\otimes T')=C(T)+C(T')$$ 
 So one can construct trellises for linear codes by taking products of elementary trellises for their generators. Remarkably, the \textit{Factorization Theorem} from \cite{KV2} (which we reprove as Corollary \ref{factthm}) says that this is exhaustive, i.e. any reduced linear trellis factors into elementary trellises. 
 Finally, we call   $$\otimes_{i=1}^{r} \al^{i}|(a_{i},l_{i})$$ an \textit{elementary trellis factorization}\index{elementary (trellis) factorization} of $T$ if $T \sim\otimes_{i=1}^{r} \al^{i}|(a_{i},l_{i})$.
 
\begin{rmk}\label{assumpt1pap} \textbf{We will always assume that in a product of elementary trellises  at most one elementary trellis with span $(a,0)$ appears for each $a\in\ZZ_{n}$}, in order to avoid the degeneracy $\al|(a,0)=\al|(a,0)\otimes\al|(a,0)$ which would require cumbersome distinctions in the statements of our theorems.
\end{rmk}

%%%%%%%%%%%%%%%%%%%%%%%%%%%%%%%%%%%%%%%%%%%%%%%%%%%%

\section{An algebraic framework for linear trellises}\label{algframpap}

\subsection{Spans of cycles and span subcodes}\label{spansubpap}

We introduce here our main tool: span subcodes of $\SSS(T)$.
Let $T$ be a linear trellis of length $n$. 
Given $a\in\ZZ_{n}$, $0\leq l\leq n-1$, and $\la=(\vv,\al)\in\SSS(T)$, we say that 
$(a,l)$ 
is a \textit{span starting at $a$ of length $l$}\index{span!of cycle $\la\in\SSS(T)$}  of $\la$
if $(a+1,l-1)$ is a span of $\bm{v}\in\prod_{i\in\ZZ_{n}}V_{i}(T)$ (with respect to the  alphabets $V_{i}(T)$)
and $(a,l)$ is a span of $\bm{\alpha}\in\FF^{n}$,
i.e. if  $\su(\vv)\subseteq (a,a+l]$ and $\su(\al)\subseteq[a,a+l]$. Like in \ref{eldef}, we say as well that $\emptyset$ is a span of length $-1$ of $\bm{0}\in\SSS(T)$ and $\ZZ_{n}$ is a span of length $n$ of all $\la\in\SSS(T)$, and write also $(a,-1)$ for $\emptyset$ and $(a,n)$ for $\ZZ_{n}$.
The linear subcode of $\SSS(T)$ defined as
\begin{equation*}
\sal(T):=\{\la\in \SSS(T)| (a,l) \textnormal{ is a span of }\la\}
\end{equation*}
 is then called the $(a,l)$-\textit{span subcode}\index{span!subcode of $\SSS(T)$} of $\SSS(T)$ (or simply of $T$).  
 
Considering the partial order for spans defined in Subsection \ref{eldef}, we have that 
$$(a_{1},l_{1})\leq (a_{2},l_{2})\implies\SSS_{(a_{1},l_{1})}(T)\leq \SSS_{(a_{2},l_{2})}(T)$$ and so the Hasse diagram 
of spans
yields a diagram giving containments amongst the span subcodes $\sal(T)$. In particular $\SSS_{\emptyset}(T)=0\leq\sal(T)\leq\SSS(T)=\SSS_{\ZZ_{n}}(T)$ for all $(a,l)$. 
For elementary trellises the picture of span subcodes is simple:  
\begin{equation*}
\SSS_{(a',l')}(\al|(a,l))=
\begin{cases}
\SSS(\al|(a,l))& \textnormal{ if } (a,l)\leq(a',l')\\
0 & \textnormal{ otherwise }
\end{cases}
\end{equation*}
Here is a less trivial example:

\begin{ex}\label{ex1pap}
Consider the below linear trellis with label code $\SSS(T)=\langle(\uu{0}\ms\uu{0}\ms\uu{0},100),(\uu{0}\ms\uu{1}\ms\uu{0},100), (\uu{0}\ms\uu{0}\ms\uu{1},011) \rangle$. 
\begin{center}
 \begin{tikzpicture}[yscale=.8,>=latex',shorten >=.9pt, shorten <=1.4pt, line width=.6pt]
  \tikzstyle{every node}=[draw,circle,fill=black,minimum size=2pt,
                        inner sep=0pt]                    
\foreach \x in {1,2}{
\node at (\x,2) {};};
\foreach \x in {0,1,2,3}{
\node at (\x,1) {};};
\tikzstyle{every node}=[]
  \draw [->,dashed] (0,1) -- (1,2);
    \draw [->] (0,1) .. controls (0.1,1.8) and (0.8,2.1) .. (1,2);
\draw [->,dashed] (1,1) -- (2,1);
\draw [->,dashed] (2,1) -- (3,1);
\draw [->,dashed] (0,1) .. controls (0.3,0.7) and (0.7,0.7) ..  (1,1);
  \draw [->] (0,1) -- (1,1);
\draw [->] (1,2) -- (2,2);
\draw [->] (2,2) -- (3,1);
\draw [->] (1,1) -- (2,2);
\draw [->,dashed] (1,2) -- (2,1);
\draw (0,1) node [below] {\tiny ${\uu{0}}$};
\draw (1,1) node [below] {\tiny ${\uu{0}}$};
\draw (2,1) node [below] {\tiny ${\uu{0}}$};
\draw (3,1) node [below] {\tiny ${\uu{0}}$};
\draw (1,2) node [above] {\tiny ${\uu{1}}$};
\draw (2,2) node [above] {\tiny ${\uu{1}}$};
      \end{tikzpicture}
    \end{center}
This yields the diagram:
\begin{center}
 \begin{tikzpicture}[xscale=.7,yscale=.85, line width=.6pt]
                                      \tikzstyle{every node}=[]
                                      
             \foreach \a in {0,1,2}{
                        \foreach \l in {0,1}{
                       \draw  (4*\a,1.5*\l) -- (4*\a,1.5*\l+1.5);};};
                       
                        \foreach \a in {1,2}{
                        \foreach \l in {0,1}{
                       \draw  (4*\a,1.5*\l) -- (4*\a-4,1.5*\l+1.5);};};
                       
                        \foreach \l in {0,1}{
                       \draw  (0,1.5*\l) -- (8,1.5*\l+1.5);};
                       
                                \foreach \a in {0,1,2}{
                       \draw  (4*\a,0) -- (4,-1.5);
                       \draw  (4*\a,3) -- (4,4.5);};
            
             \tikzstyle{every node}=[fill=white,minimum size=0pt,
                        inner sep=2pt]\small
                        
\draw (0,0) node {$\langle(\uu{0}\ms\uu{0}\ms\uu{0},100)\rangle$};
\draw (0.5,1.5) node {$\langle(\uu{0}\ms\uu{0}\ms\uu{0},100),(\uu{0}\ms\uu{1}\ms\uu{0},000)\rangle$};
\draw (0,3) node {$\SSS(T)$};

\draw (4,0) node {$0$};
\draw (4.5,1.5) node {$\langle(\uu{0}\ms\uu{0}\ms\uu{1},011)\rangle$};
\draw (3.2,3) node {$\langle(\uu{0}\ms\uu{0}\ms\uu{1},011),(\uu{0}\ms\uu{0}\ms\uu{0},100)\rangle$};

\draw (8,0) node {$0$};
\draw (8,1.5) node {$\langle(\uu{0}\ms\uu{0}\ms\uu{0},100)\rangle$};
\draw (8.2,3) node {$\langle(\uu{0}\ms\uu{0}\ms\uu{0},100),(\uu{0}\ms\uu{1}\ms\uu{0},000)\rangle$};

    \draw (4,-1.5) node {$0$};
     \draw (4,4.5) node {$\SSS(T)$};
    
      \end{tikzpicture}
\end{center}
\end{ex}

Note that $\SSS_{(a,0)}(T)$ has either dimension $0$ or $1$, since between adjacent vertices 
there exists either $1$ edge (labelled with $0$) or $|\FF|$ parallel edges (whose labels span $\FF$).

 We call the minimum length of spans of $\la$ the  \textit{span length}\index{span length of cycles} of $\la$ and denote it by  
 $\ell(\la)$.
The subcode generated by all cycles of span length less or equal than a given $l$ is denoted  by 
\begin{equation*}\SSS_{l}(T):=\sum_{a\in\ZZ_{n}, l'\leq l}\SSS_{(a,l')}(T)\end{equation*}
We also put 
\begin{equation*}
\SSS_{<(a,l)}(T):=\sum_{(a',l')\lneq(a,l)}\SSS_{(a',l')}(T)
\end{equation*}
 
  If $\la$ has a minimum span
   we will denote it by $[\bm{\lambda}]$\nomenclature[lasp]{$[\la]$}{minimum span of cycle $\la$} and refer to it simply as \textit{the span} of $\bm{\lambda}$.  
   If $\la$ has not a minimum span then it is easy to see that $\la\in\SSS_{\ell(\la)-1}(T)$.
   
  \begin{ex}
  Let $T$ be the unlabeled linear trellis
  
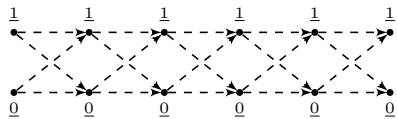
\begin{center}
 \begin{tikzpicture}[yscale=.8,>=latex',shorten >=.9pt, shorten <=1.4pt, line width=.6pt]
  \tikzstyle{every node}=[draw,circle,fill=black,minimum size=2pt,
                        inner sep=0pt]                    
\foreach \x in {0,1,2,3,4,5}{
\node at (\x,2) {};};
\foreach \x in {0,1,2,3,4,5}{
\node at (\x,1) {};};
\tikzstyle{every node}=[]
  \draw [->,dashed] (0,1) -- (1,1);
  \draw [->,dashed] (0,1) -- (1,2);
\draw [->,dashed] (0,2) -- (1,1);
\draw [->,dashed] (0,2) -- (1,2);
\begin{scope}
[xshift=1cm]
  \draw [->,dashed] (0,1) -- (1,1);
  \draw [->,dashed] (0,1) -- (1,2);
\draw [->,dashed] (0,2) -- (1,1);
\draw [->,dashed] (0,2) -- (1,2);
\draw (0,1) node [below] {\tiny ${\uu{0}}$};
\draw (0,2) node [above] {\tiny ${\uu{1}}$};
\end{scope}
\begin{scope}
[xshift=2cm]
  \draw [->,dashed] (0,1) -- (1,1);
  \draw [->,dashed] (0,1) -- (1,2);
\draw [->,dashed] (0,2) -- (1,1);
\draw [->,dashed] (0,2) -- (1,2);
\draw (0,1) node [below] {\tiny ${\uu{0}}$};
\draw (0,2) node [above] {\tiny ${\uu{1}}$};
\end{scope}
\begin{scope}
[xshift=3cm]
  \draw [->,dashed] (0,1) -- (1,1);
  \draw [->,dashed] (0,1) -- (1,2);
\draw [->,dashed] (0,2) -- (1,1);
\draw [->,dashed] (0,2) -- (1,2);
\draw (0,1) node [below] {\tiny ${\uu{0}}$};
\draw (0,2) node [above] {\tiny ${\uu{1}}$};
\end{scope}
\begin{scope}
[xshift=4cm]
  \draw [->,dashed] (0,1) -- (1,1);
  \draw [->,dashed] (0,1) -- (1,2);
\draw [->,dashed] (0,2) -- (1,1);
\draw [->,dashed] (0,2) -- (1,2);
\draw (0,1) node [below] {\tiny ${\uu{0}}$};
\draw (0,2) node [above] {\tiny ${\uu{1}}$};
\end{scope}

\draw (0,1) node [below] {\tiny ${\uu{0}}$};
\draw (0,2) node [above] {\tiny ${\uu{1}}$};
\draw (5,1) node [below] {\tiny ${\uu{0}}$};
\draw (5,2) node [above] {\tiny ${\uu{1}}$};
      \end{tikzpicture}
    \end{center}
Then 
$[(\uu{0}\ms\uu{1}\ms\uu{1}\ms\uu{1}\ms\uu{0},\bm{0})]=(0,3)$. 
On the other hand, 
$(\uu{1}\ms\uu{0}\ms\uu{1}\ms\uu{0}\ms\uu{0},\bm{0})$ 
has two minimal incomparable spans, $(1,4)$ and $(4,3)$. 
We see that $(\uu{1}\ms\uu{0}\ms\uu{1}\ms\uu{0}\ms\uu{0},\bm{0})=(\uu{0}\ms\uu{0}\ms\uu{1}\ms\uu{0}\ms\uu{0},\bm{0})+(\uu{1}\ms\uu{0}\ms\uu{0}\ms\uu{0}\ms\uu{0},\bm{0})\in\SSS_{(1,1)}(T)+\SSS_{(4,1)}(T)$, while $\ell((\uu{1}\ms\uu{0}\ms\uu{1}\ms\uu{0}\ms\uu{0},\bm{0}))=3$.
\end{ex}

 Finally, for each $(a,l)$ we put 
\begin{align*}
&C_{(a,l)}(T):=L(\sal(T))
\end{align*}
Obviously $C_{(a,l)}(T)\subseteq\{\al\in C(T)| \textrm{ (a,l) is a span of }\al\}$, but in general the containment is strict. 

%%%%%%%%%%%%%%%%%%%%%%%%%%%%%%%%%%%%%%%%%%%%%%%%%%%%%%
%%%%%%%%%%%%%%%%%%%%%%%%%%%%%%%%%%%%%%%%%%%%%%%%%%%%%%
%%%%%%%%%%%%%%%%%%%%%%%%%%%%%%%%%%%%%%%%%%%%%%%%%%%%%%
%%%%%%%%%%%%%%%%%%%%%%%%%%%%%%%%%%%%%%%%%%%%%%%%%%%%%%
%%%%%%%%%%%%%%%%%%%%%%%%%%%%%%%%%%%%%%%%%%%%%%%%%%%%%%
%%%%%%%%%%%%%%%%%%%%%%%%%%%%%%%%%%%%%%%%%%%%%%%%%%%%%%
%%%%%%%%%%%%%%%%%%%%%%%%%%%%%%%%%%%%%%%%%%%%%%%%%%%%%%
%%%%%%%%%%%%%%%%%%%%%%%%%%%%%%%%%%%%%%%%%%%%%%%%%%%%%%

\subsection{Label code maps as trellis maps}\label{labcodemapspap}

Let $T$, $T'$ be linear trellises.
Any  linear morphism $f:T\rightarrow T'$ 
induces a  linear map $\SSS(f):\SSS(T)\rightarrow\SSS(T')$\nomenclature[sf]{$\SSS(f)$}{trellis map}  given by 
$$\SSS(f)(v_{0}\alpha_{0}\ldots v_{n-1}\alpha_{n-1})=f_{0}(v_{0})\alpha_{0}\ldots f_{n-1}(v_{n-1})\alpha_{n-1}$$ 
We say that a linear map $F:\SSS(T)\rightarrow\SSS(T')$ is a \textit{trellis map}\index{trellis map ($\SSS(f)$)} if $F=\SSS(f)$ for some linear morphism $f:T\rightarrow T'$. 
It is easily seen that $F$ 
is a trellis map if and  only if it preserves edge-labels 
and $\nu_{i}(\la)=\nu_{i}(\la')\Rightarrow\nu_{i}(F(\la))=\nu_{i}(F(\la'))$ for all $i\in\ZZ_{n}$ and all $\la,\la'\in\SSS(T)$. In fact if those conditions are satisfied then  $F=\SSS(f)$ with $f_{i}(v):=\nu_{i}(F(\la))$ for all $v\in V_{i}(T)$ and any $\la$ such that $\nu_{i}(\la)=v$. By linearity we thus have also that   $F$  
is a trellis map if and only if it preserves edge-labels and  
 $$F(\SSS_{(a,n-1)}(T))\leq\SSS_{(a,n-1)}(T')$$ for all $a\in\ZZ_{n}$ (which implies that $F(\SSS_{(a,l)}(T))\leq\SSS_{(a,l)}(T')$ for all  $(a,l)$).

Note that  $\SSS(f\circ g)=\SSS(f)\circ \SSS(g)$ (the composition of trellis morphisms being the obvious one) and $f=g$ if $\SSS(f)=\SSS(g)$. So if $F=\SSS(f)$ and $F$ has an inverse which is also a trellis map 
then 
$f$ is a linear isomorphism (and vice versa). 
We thus conclude that:
 \begin{obs}\label{trellmappap}
 $T\simeq T'$  if and only if there exists a linear isomorphism $F:\SSS(T)\rightarrow\SSS(T')$ that preserves edge-labels and such that $$F(\SSS_{(a,n-1)}(T))=\SSS_{(a,n-1)}(T')$$ 
 for all  $a\in\ZZ_{n}$ (in which case $F(\SSS_{(a,l)}(T))=\SSS_{(a,l)}(T')$ for all  $(a,l)$).  
\end{obs}

This observation allows us to prove that two linear trellises are isomorphic by focusing locally on the span subcodes. See our core Theorem \ref{thm15} for this crucial approach in its proof.

Note that the existence of a linear isomorphism between $\SSS(T)$ and $\SSS(T')$ which is a trellis map only in one direction is not sufficient for $T$ and $T'$ to be isomorphic.
For example, if $T$ and $T'$ are the two nonisomorphic trellises below and $f:T\rightarrow T'$  
 is given by $f_{0}=\textrm{Id}$, $f_{1}=\textrm{Id}$, $f_{2}(\uu{01})=f_{2}(\uu{11})=\uu{1}$, then  $\SSS(f)$ is a linear isomorphism. 
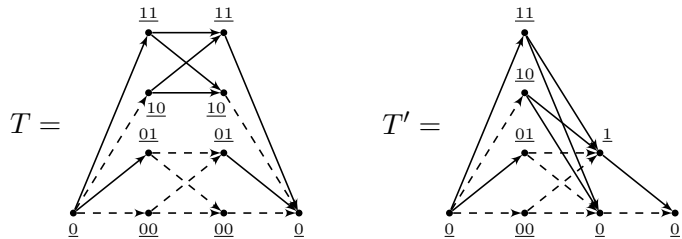
\begin{center}
 \begin{tikzpicture}[yscale=.8,>=latex',shorten >=.9pt, shorten <=1.4pt, line width=.6pt]
  \tikzstyle{every node}=[draw,circle,fill=black,minimum size=2pt,
                        inner sep=0pt]                    
\foreach \x in {1,2}{
\node at (\x,2) {};};
\foreach \x in {1,2}{
\node at (\x,3) {};};
\foreach \x in {1,2}{
\node at (\x,4) {};};
\foreach \x in {0,1,2,3}{
\node at (\x,1) {};};
\tikzstyle{every node}=[]
  \draw [->,dashed,black] (0,1) -- (1,1);
\draw [->,dashed,black] (1,1) -- (2,1);
\draw [->,dashed,black] (2,1) -- (3,1);
\draw [->,black] (0,1) -- (1,2);
\draw [->,dashed,black] (1,2) -- (2,2);
\draw [->,black] (2,2) -- (3,1);
\draw [->,dashed,black] (1,1) -- (2,2);
\draw [->,dashed,black] (1,2) -- (2,1);

  \draw [->,dashed,black] (0,1) -- (1,3);
\draw [->,black] (1,3) -- (2,3);
\draw [->,dashed,black] (2,3) -- (3,1);
\draw [->,black] (0,1) -- (1,4);
\draw [->,black] (1,4) -- (2,4);
\draw [->,black] (2,4) -- (3,1);
\draw [->,black] (1,3) -- (2,4);
\draw [->,black] (1,4) -- (2,3);

\draw (0,1) node [below] {\tiny ${\uu{0}}$};
\draw (1,1) node [below] {\tiny ${\uu{00}}$};
\draw (2,1) node [below] {\tiny ${\uu{00}}$};
\draw (3,1) node [below] {\tiny ${\uu{0}}$};
\draw (1,2) node [above] {\tiny ${\uu{01}}$};
\draw (2,2) node [above] {\tiny ${\uu{01}}$};

\draw (1.1,3) node [below] {\tiny ${\uu{10}}$};
\draw (1.9,3) node [below] {\tiny ${\uu{10}}$};
\draw (1,4) node [above] {\tiny ${\uu{11}}$};
\draw (2,4) node [above] {\tiny ${\uu{11}}$};

\draw (-.5,2.5) node {$T=$};

\begin{scope}
[xshift=5cm]
\tikzstyle{every node}=[draw,circle,fill=black,minimum size=2pt,
                        inner sep=0pt]                    
\foreach \x in {1,2}{
\node at (\x,2) {};};
\foreach \x in {1}{
\node at (\x,3) {};};
\foreach \x in {1}{
\node at (\x,4) {};};
\foreach \x in {0,1,2,3}{
\node at (\x,1) {};};
\tikzstyle{every node}=[]
  \draw [->,dashed,black] (0,1) -- (1,1);
\draw [->,dashed,black] (1,1) -- (2,1);
\draw [->,dashed,black] (2,1) -- (3,1);
\draw [->,black] (0,1) -- (1,2);
\draw [->,dashed,black] (1,2) -- (2,2);
\draw [->,black] (2,2) -- (3,1);
\draw [->,dashed,black] (1,1) -- (2,2);
\draw [->,dashed,black] (1,2) -- (2,1);

  \draw [->,dashed,black] (0,1) -- (1,3);
\draw [->,black] (1,3) -- (2,1);
\draw [->,black] (0,1) -- (1,4);
\draw [->,black] (1,4) -- (2,2);
\draw [->,black] (1,3) -- (2,2);
\draw [->,black] (1,4) -- (2,1);

\draw (0,1) node [below] {\tiny ${\uu{0}}$};
\draw (1,1) node [below] {\tiny ${\uu{00}}$};
\draw (2,1) node [below] {\tiny ${\uu{0}}$};
\draw (3,1) node [below] {\tiny ${\uu{0}}$};
\draw (1,2) node [above] {\tiny ${\uu{01}}$};
\draw (2.1,2) node [above] {\tiny ${\uu{1}}$};

\draw (1,3) node [above] {\tiny ${\uu{10}}$};
\draw (1,4) node [above] {\tiny ${\uu{11}}$};

\draw (-.5,2.5) node {$T'=$};
\end{scope}
      \end{tikzpicture}
    \end{center}

%%%%%%%%%%%%%%%%%%%%%%%%%%%%%%%%%%%%%%%%%%%%%%%%%%%%%%
%%%%%%%%%%%%%%%%%%%%%%%%%%%%%%%%%%%%%%%%%%%%%%%%%%%%%%
%%%%%%%%%%%%%%%%%%%%%%%%%%%%%%%%%%%%%%%%%%%%%%%%%%%%%%
%%%%%%%%%%%%%%%%%%%%%%%%%%%%%%%%%%%%%%%%%%%%%%%%%%%%%%
%%%%%%%%%%%%%%%%%%%%%%%%%%%%%%%%%%%%%%%%%%%%%%%%%%%%%%
%%%%%%%%%%%%%%%%%%%%%%%%%%%%%%%%%%%%%%%%%%%%%%%%%%%%%%
%%%%%%%%%%%%%%%%%%%%%%%%%%%%%%%%%%%%%%%%%%%%%%%%%%%%%%
%%%%%%%%%%%%%%%%%%%%%%%%%%%%%%%%%%%%%%%%%%%%%%%%%%%%%%
%%%%%%%%%%%%%%%%%%%%%%%%%%%%%%%%%%%%%%%%%%%%%%%%%%%%%%
 
\subsection{Algebraic structure of $\SSS(T)$: product bases}\label{prodbaspap}

We now isolate  the key properties of the label code $\SSS(T)$ (and its family of subcodes $\{\sal(T)\}$) related to how trellis factorizations are encoded in its structure.

Consider two  linear trellises $T$ and $T'$. Naturally, $T$ is identified with the linear  subtrellis $T\otimes 0\leq T\otimes T'$ via the (injective) linear morphism $T\rightarrow T\otimes T'$ given by $$V_{i}(T)\ni v\mapsto (v,0)\in V_{i}(T)\times V_{i}(T')$$
Consequently each span subcode $\SSS_{(a,l)}(T)$  is identified with a subcode of  $\SSS_{(a,l)}(T\otimes T')$ by the map $(\vv,\al)\mapsto((\bm{v},\bm{0}),\bm{\alpha})$. With those identifications in mind we have: 

\begin{obs}\label{obs10} $\SSS_{(a,l)}(T\otimes T')=\SSS_{(a,l)}(T)+\SSS_{(a,l)}(T')$ for each span $(a,l)$. Moreover, if for all $b\in[a,a+l]$ there are no parallel edges at time index $b$ in $T$ and $T'$ simultaneously 
then $$\SSS_{(a,l)}(T\otimes T')=\SSS_{(a,l)}(T)\oplus \SSS_{(a,l)}(T')$$

\end{obs}
\begin{proof}
The  equality $\SSS_{(a,l)}(T\otimes T')=\SSS_{(a,l)}(T)+\SSS_{(a,l)}(T')$  is immediate. It remains to prove the second statement, i.e. that  $\SSS_{(a,l)}(T)\cap\SSS_{(a,l)}(T')=0$ under the given hypothesis. 
Assume that we have 
$$((\bm{v},\bm{0}),\bm{\alpha})=((\bm{0},\bm{v}'),\bm{\alpha}')\in\sal(T)\cap\sal(T')$$ 
Then $\bm{v}=\bm{0}$, $\bm{v}'=\bm{0}$, and $\bm{\alpha}=\bm{\alpha}'$. This implies that for all $b\in\textnormal{supp}(\al)\subseteq[a,a+l]$ there are parallel edges at time index $b$ both in $T$ and $T'$. So $\al=\bm{0}$, and we are done.
\end{proof}

\begin{rmk}
By obvious inductive arguments the above extends to any  finite linear trellis product $\otimes_{i=1}^{r}T_{i}$. That is, identifying  
naturally $\sal(T_{i})$ as a subspace of $\sal(\otimes_{i=1}^{r}T_{i})$  we can write $\sal(\otimes_{i=1}^{r}T_{i})=\sum_{i=1}^{r}\sal(T_{i})$, where the sum is direct if $\sum_{i=1}^{r}\dim\SSS_{(b,0)}(T_{i})\leq1$ (i.e. at most one $T_{i}$ has parallel edges at $b$) for all $(b,0)\leq(a,l)$. 
\end{rmk}

We can use Observation \ref{obs10}  to analyze the structure of the label code of a product of elementary trellises.  Recall that if $T=\al|(a,l)$ 
then $\SSS_{(a',l')}(T)=\SSS(T)$ if $(a',l')\geq(a,l)$ and $\SSS_{(a',l')}(T)=0$ otherwise. It then follows that:

\begin{obs}\label{obs11}
Let $T=\otimes_{i=1}^{r}\al^{i}|(a_{i},l_{i})$. 
For each $i$ let $\la^{i}$ be a generator of $\SSS(\al^{i}|(a_{i},l_{i}))$, which we identify 
with a subspace of $\SSS(T)$ as usual. 
Then $\{\la^{i}\}_{i|(a_{i},l_{i})\leq(a,l)}$ is a basis of $\SSS_{(a,l)}(T)$. 
 \end{obs}

Thus the label code $\SSS(\otimes_{i=1}^{r}T_{i})$ of an elementary trellis product has a basis $\BB$ such that $\BB\cap\sal(\otimes_{i=1}^{r}T_{i})$ generates $\sal(\otimes_{i=1}^{r}T_{i})$, for each span $(a,l)$. We are now going to show that any linear trellis label code possesses such a basis. 
In order to do so we   
first  prove a more fundamental property of label codes.

\begin{thm}\label{thm11}  Let $T$ be a linear trellis, and let $\{(a_{i},l_{i})\}_{i\in\mathcal{I}}$ be a family of spans. If  $(a,l)$ is a span such that $(a,l)\nleq (a_{i},l_{i})$ for all $i\in\III$ then
$$\SSS_{(a,l)}(T)\cap\sum_{i\in\III}\SSS_{(a_{i},l_{i})}(T)\leq\SSS_{<(a,l)}(T)$$ 
\end{thm}
\begin{proof} Let $(\bm{v},\bm{\alpha})\in \SSS_{(a,l)}(T)\cap\sum_{i\in\III}\SSS_{(a_{i},l_{i})}(T)$.
If $l=-1$ or $l=n$, then the statement is trivially true.
If $l=0$, then it is clear that $(\bm{v},\bm{\alpha})=(\bm{0},\bm{0})$, since any cycle in $\SSS_{(a_{i},l_{i})}(T)$ yields a codeword whose support does not contain $a$. So, assume now $n>l\geq1$.
If $v_{a+1}=0$ then the statement is always true as  in that case $(\bm{v},\bm{\alpha})$ is clearly a sum of a cycle in $\SSS_{(a,0)}(T)$ and a cycle in $\SSS_{(a+1,l-1)}(T)$, namely, $(\bm{0},\bm{\alpha}')\in\SSS_{(a,0)}(T)$ where $\alpha'_{j}=\delta_{ja}\alpha_{a}$,  and  
$$(\bm{v},\bm{\alpha})-(\bm{0},\bm{\alpha}')\in\SSS_{(a+1,l-1)}(T)$$

So we can also assume  that $v_{a+1}\neq0$. Now, write $(\bm{v},\bm{\alpha})=\sum_{i\in\III}(\bm{v}^{i},\bm{\alpha}^{i})$ for some $(\bm{v}^{i},\bm{\alpha}^{i})\in\SSS_{(a_{i},l_{i})}(T)$. Let 
$$\III':=\{i\in\III| {v}^{i}_{a+1}\neq 0\}$$ 
Note that  $\III'\neq\emptyset$, as $v_{a+1}\neq0$. If $i\in\III'$ then obviously $a+1\in(a_{i},a_{i}+l_{i}]$, and so we must have $(a,a_{i}+l_{i}]\subsetneq(a,a+l]$, because otherwise $(a,a+l]\subseteq(a,a_{i}+l_{i}]\subseteq (a_{i},a_{i}+l_{i}]$, 
which would contradict the hypothesis.  
We deduce that $l\geq2$   
and so that for each $i\in\III'$ we can construct a path $\bm{p}^{i}$ in $T$ of length $l-1$ from $v^{i}_{a+1}$ to $0_{a+l}\in V_{a+l}$ given by concatenating the path 
$$v^{i}_{a+1}\alpha^{i}_{a+1}\ldots v^{i}_{a_{i}+l_{i}}\alpha^{i}_{a_{i}+l_{i}}0_{a_{i}+l_{i}+1}$$ 
with the zero path from $0_{a_{i}+l_{i}+1}$ to $0_{a+l}$. By linearity of $T$ we can add all those paths to get a path $\bm{p}=\sum_{i\in\III'}\bm{p}^{i}$ from $v_{a+1}$ to $0_{a+l}$. Concatenating  the edge $0_{a}\alpha_{a}v_{a+1}$ with $\bm{p}$, and then extending  on right and left by the zero path, we get a cycle $(\bm{v}',\bm{\alpha}')\in\SSS_{(a,l-1)}(T)$ such that $(\bm{v},\bm{\alpha})-(\bm{v}',\bm{\alpha}')\in\SSS_{(a+1,l-1)}(T)$. Hence we are done.
\end{proof}

\begin{thm}\label{thm12} Let $T$ be a linear trellis. 
For  each span $(a,l)$ lift an arbitrary basis of $$\sal(T)/\SSS_{<(a,l)}(T)$$  to a subset $\BB_{(a,l)}\subseteq\sal(T)$. Then $\BB:=\sqcup_{(a,l)}\BB_{(a,l)}$ is a basis of $\SSS(T)$  such that  $\BB\cap\SSS_{(a,l)}(T)$ generates $\SSS_{(a,l)}(T)$  for all $(a,l)$.
\end{thm}
\begin{proof} 
We will prove by induction on $ l=-1,\ldots,  n-1$, that  $\BB_{l}:=\sqcup_{(a,l')|l'\leq l}\BB_{(a,l')}$ is a basis of $\SSS_{l}(T)$  
 such that $\BB_{l}\cap \SSS_{(a,l')}(T)$ generates $\SSS_{(a,l')}(T)$ for all $(a,l')$ with $l'\leq l$. Then clearly  $\sqcup_{(a,l)}\BB_{(a,l)}$ will satisfy our statement. 
If $l=-1$ there is nothing to prove. 
Assume now $\BB_{l}$ satisfies the above property for some $-1\leq l< n-1$.  
We claim that  $\BB_{l+1}$ is a basis of $\SSS_{l+1}(T)$. So, suppose $\sum_{\la\in\BB_{l+1}}x_{\la}\bm{\lambda}=0$, for some coefficients $x_{\la}$ in $\FF$. Then for $a\in\ZZ_{n}$ we have that 
\begin{gather*}
\sum_{\la\in\BB_{(a,l+1)}}x_{\la}\bm{\lambda}\in\bigl(\SSS_{(a,l+1)}(T)\cap(\sum_{a'\neq a}\SSS_{(a',l+1)}(T)+\SSS_{l}(T))\bigr)\\
\subseteq \SSS_{<(a,l+1)}(T)
\end{gather*}
 where the containment follows from Theorem \ref{thm11}. 
 Hence $x_{\la}=0$ for all $\la\in\BB_{(a,l+1)}$  and all $a\in \ZZ_{n}$. So  $x_{\la}=0$  for all $\la\in\BB_{l}$ too, since $\BB_{l}$ is a basis of $\SSS_{l}(T)$, and our claim is proven. Now, clearly $\BB_{l+1}\cap \SSS_{(a,l')}(T)$ generates $\SSS_{(a,l')}(T)$ for all $(a,l')$ with $l'\leq l$, as $\BB_{l}\cap \SSS_{(a,l')}(T)$ does. For the same reason $\BB_{l}\cap \SSS_{<(a,l+1)}(T)$ generates $\SSS_{<(a,l+1)}(T)$. So by definition of $\BB_{(a,l+1)}$ we get that 
 $$(\BB_{l}\sqcup\BB_{(a,l+1)})\cap\SSS_{(a,l+1)}(T)$$ generates $\SSS_{(a,l+1)}(T)$.
 This concludes our proof. 
 \end{proof}

\begin{rmk}  
Given a pair $(V,\{V_{i}\}_{i\in\III})$ consisting of a vector space $V$ and a family of subspaces of $V$, in general is not true that there exists a basis $\BB$ of $V$ such that $\BB\cap V_{i}$ generates $V_{i}$ for all $i\in\III$.
A counterexample is given by $\FF^{5}$ with the family of subspaces 
$\{\langle\mathbf{e}^{0}\rangle,\langle\mathbf{e}^{1}\rangle,\langle\mathbf{e}^{2}\rangle,\langle\mathbf{e}^{3}\rangle,\langle\mathbf{e}^{0},\mathbf{e}^{1},\mathbf{e}^{4}\rangle,\langle\mathbf{e}^{1},\mathbf{e}^{2},\mathbf{e}^{0}+\mathbf{e}^{4}\rangle,\langle\mathbf{e}^{2},\mathbf{e}^{3},\mathbf{e}^{1}+\mathbf{e}^{4}\rangle,\langle\mathbf{e}^{3},\mathbf{e}^{0},\mathbf{e}^{2}+\mathbf{e}^{4}\rangle\}$ (where $\mathbf{e}^{i}$ is the $i$-th canonical basis element). Only special pairs  $(V,\{V_{i}\}_{i\in\III})$ have such bases.
\end{rmk}

A basis of $\SSS(T)$ satisfying the property of Theorem \ref{thm12} will be called a \textit{product basis}\index{product basis! of $\SSS(T)$ (or $T$)} of $\SSS(T)$ (or simply of $T$). This denomination is justified by the fact  that  product bases correspond to factorizations of $T$ in terms of elementary trellises, as we will see in the next subsection.

The following observation gives some useful properties of product bases.
\begin{obs}\label{spansub1} Let 
$\mathcal{B}$ a product basis of $T$. Then:
\begin{enumerate}
\item Every $\la\in\BB$ 
has a minimum span
\item $\sum_{i\in\mathcal{I}}\SSS_{(a_{i},l_{i})}(T)=\langle\bm{\lambda}\in\mathcal{B}|[\la]\leq(a_{i},l_{i}) \textnormal{ for some } i\in\mathcal{I}\rangle$
 \item  $\SSS_{<(a,l)}(T)=\langle\bm{\lambda}\in\mathcal{B}|[\la]<(a,l)\rangle$
 \item $\{{\la}+\SSS_{<(a,l)}(T)|\la\in\BB, [\la]=(a,l)\}$ is a basis of $\SSS_{(a,l)}(T)/\SSS_{<(a,l)}(T)$ 
 \item A cycle $\la\in\SSS(T)$ belongs to some product basis of $T$ if and only if $\la\in\SSS_{(a,l)}(T)\setminus\SSS_{<(a,l)}(T)$ for some $(a,l)$ 
\end{enumerate}
\end{obs}
\begin{proof}
Let $\bm{\lambda}\in\BB$ be a product basis element, and let $(a,l)$ be a span of minimum length $l=\ell(\la)$ of $\la$. 
Assume that $\la$ does not have a minimum span.
Then $\bm{\lambda}\in\SSS_{<(a,l)}(T)$. 
Since $\BB$ is a product basis it follows  that $$\BB\cap\left(\cup_{(a',l')<(a,l)} \SSS_{(a',l')}(T)\right)$$ 
is a basis of $\SSS_{<(a,l)}(T)$, so that $\bm{\lambda}\in\SSS_{(a',l')}(T)$ for some $(a',l')<(a,l)$, which contradicts our assumption on $l$.
As for the identities: identity ${2}$) follows easily from the definition of product basis; identity ${3}$) is an instance of ${2}$); identity ${4}$) is a plain consequence of  ${2}$) and ${3}$). 
As for $5)$ the ``if'' part follows from Theorem \ref{thm12}, while the ``only if'' part follows from $4)$.
\end{proof}

%%%%%%%%%%%%%%%%%%%%%%%%%%%%%%%%%%%%%%%%%%%%%%%%%%%%%%
%%%%%%%%%%%%%%%%%%%%%%%%%%%%%%%%%%%%%%%%%%%%%%%%%%%%%%
%%%%%%%%%%%%%%%%%%%%%%%%%%%%%%%%%%%%%%%%%%%%%%%%%%%%%%
%%%%%%%%%%%%%%%%%%%%%%%%%%%%%%%%%%%%%%%%%%%%%%%%%%%%%%
%%%%%%%%%%%%%%%%%%%%%%%%%%%%%%%%%%%%%%%%%%%%%%%%%%%%%%
%%%%%%%%%%%%%%%%%%%%%%%%%%%%%%%%%%%%%%%%%%%%%%%%%%%%%%
%%%%%%%%%%%%%%%%%%%%%%%%%%%%%%%%%%%%%%%%%%%%%%%%%%%%%%
%%%%%%%%%%%%%%%%%%%%%%%%%%%%%%%%%%%%%%%%%%%%%%%%%%%%%%
%%%%%%%%%%%%%%%%%%%%%%%%%%%%%%%%%%%%%%%%%%%%%%%%%%%%%%

\subsection{Product bases and elementary trellis factorizations}

In Observation \ref{obs11} we have observed that an elementary trellis factorization translates into a product basis. We now show that the converse is true too.

\begin{thm}\label{thm13} 
 $T$ has a product basis $\{(\bm{v}^{i},\bm{\alpha}^{i})\}_{i=1,\ldots,r}$ with $[(\bm{v}^{i},\bm{\alpha}^{i})]=(a_{i},l_{i})$ if and only if $T\simeq\otimes_{i=1}^{r}\bm{\alpha}^{i}|(a_{i},l_{i})$.
\end{thm}
\begin{proof}
Let $T'=\otimes_{i=1}^{r}\bm{\alpha}^{i}|(a_{i},l_{i})$. For each $i=1,\ldots,r$, take a generator $\la^{i}$ of $\SSS(\bm{\alpha}^{i}|(a_{i},l_{i}))$ such that $L(\la^{i})=\al^{i}$. Now, assume $\{(\bm{v}^{i},\bm{\alpha}^{i})\}_{i=1,\ldots,r}$ is a product basis of $T$ with $[(\bm{v}^{i},\bm{\alpha}^{i})]=(a_{i},l_{i})$. 
In particular, $\{(\bm{v}^{i},\bm{\alpha}^{i})\}_{i|(a_{i},l_{i})\leq(a,l)}$ is a basis of $\SSS_{(a,l)}(T)$ for all spans $(a,l)$. Similarly, by Observation \ref{obs11}, $\{\la^{i}\}_{i|(a_{i},l_{i})\leq(a,l)}$ is a basis of $\SSS_{(a,l)}(T')$, for all $(a,l)$. 
So, by Observation \ref{trellmappap} the linear isomorphism $F:\SSS(T')\rightarrow\SSS(T)$ that sends $\la^{i}$ to $(\bm{v}^{i},\al^{i})$ yields a linear isomorphism of $T'$ and $T$. 
The ``if'' part is clear, since a linear isomorphism of trellises sends a product basis to a product basis, and $\{\la^{i}\}_{i=1,\ldots,r}$ is a product basis of $T'$ by Observation \ref{obs11}. 
\end{proof}

\begin{ex}\label{ex5cr}
Let $T$ be the linear trellis with $\SSS(T)=\langle(\uu{0}\ms\uu{1}\ms\uu{0},000), (\uu{0}\ms\uu{0}\ms\uu{1},011) \rangle$. The only product basis of $T$ is $\{(\uu{0}\ms\uu{1}\ms\uu{0},000), (\uu{0}\ms\uu{0}\ms\uu{1},011)\}$. This corresponds to the only elementary trellis factorization, depicted just below.
The basis $\{(\uu{0}\ms\uu{1}\ms\uu{0},000), (\uu{0}\ms\uu{1}\ms\uu{1},011)\}$ is not a product basis as the second product below is not $T$. 
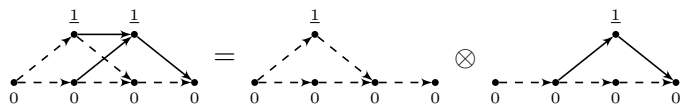
\begin{center}
 \begin{tikzpicture}[xscale=.8,yscale=.64,>=latex',shorten >=.9pt, shorten <=1.4pt, line width=.6pt]
  
  \tikzstyle{every node}=[draw,circle,fill=black,minimum size=2pt,
                        inner sep=0pt]                    
\foreach \x in {1,2}{
\node at (\x,2) {};};
\foreach \x in {0,1,2,3}{
\node at (\x,1) {};};
\tikzstyle{every node}=[]
  \draw [->,dashed] (0,1) -- (1,2);
\draw [->,dashed] (1,1) -- (2,1);
\draw [->,dashed] (2,1) -- (3,1);
  \draw [->,dashed] (0,1) -- (1,1);
\draw [->] (1,2) -- (2,2);
\draw [->] (2,2) -- (3,1);
\draw [->] (1,1) -- (2,2);
\draw [->,dashed] (1,2) -- (2,1);
\draw (0,1) node [below] {\tiny ${\uu{0}}$};
\draw (1,1) node [below] {\tiny ${\uu{0}}$};
\draw (2,1) node [below] {\tiny ${\uu{0}}$};
\draw (3,1) node [below] {\tiny ${\uu{0}}$};
\draw (1,2) node [above] {\tiny ${\uu{1}}$};
\draw (2,2) node [above] {\tiny ${\uu{1}}$};
\begin{scope}
[xshift=4cm]
  
  \tikzstyle{every node}=[draw,circle,fill=black,minimum size=2pt,
                        inner sep=0pt]                    
\foreach \x in {1}{
\node at (\x,2) {};};
\foreach \x in {0,1,2,3}{
\node at (\x,1) {};};
\tikzstyle{every node}=[]
  \draw [->,dashed] (0,1) -- (1,2);
\draw [->,dashed] (1,1) -- (2,1);
\draw [->,dashed] (2,1) -- (3,1);
  \draw [->,dashed] (0,1) -- (1,1);
\draw [->,dashed] (1,2) -- (2,1);
\draw (0,1) node [below] {\tiny ${\uu{0}}$};
\draw (1,1) node [below] {\tiny ${\uu{0}}$};
\draw (2,1) node [below] {\tiny ${\uu{0}}$};
\draw (3,1) node [below] {\tiny ${\uu{0}}$};
\draw (1,2) node [above] {\tiny ${\uu{1}}$};
\end{scope}

\begin{scope}
[xshift=8cm]
    \tikzstyle{every node}=[draw,circle,fill=black,minimum size=2pt,
                        inner sep=0pt]                    
\foreach \x in {2}{
\node at (\x,2) {};};
\foreach \x in {0,1,2,3}{
\node at (\x,1) {};};
\tikzstyle{every node}=[]

\draw [->,dashed] (1,1) -- (2,1);
\draw [->,dashed] (2,1) -- (3,1);
  \draw [->,dashed] (0,1) -- (1,1);

\draw [->] (2,2) -- (3,1);
\draw [->] (1,1) -- (2,2);

\draw (0,1) node [below] {\tiny ${\uu{0}}$};
\draw (1,1) node [below] {\tiny ${\uu{0}}$};
\draw (2,1) node [below] {\tiny ${\uu{0}}$};
\draw (3,1) node [below] {\tiny ${\uu{0}}$};
\draw (2,2) node [above] {\tiny ${\uu{1}}$};
\end{scope}

\draw (3.5,1.5) node {$=$};
\draw (7.5,1.5) node {$\otimes$};
  \end{tikzpicture}
          \end{center}
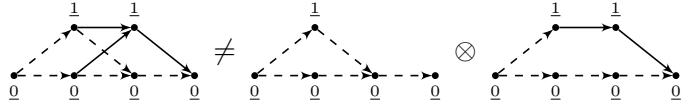
\begin{center}
 \begin{tikzpicture}[xscale=.8,yscale=.64,>=latex',shorten >=.9pt, shorten <=1.4pt, line width=.6pt]
  
  \tikzstyle{every node}=[draw,circle,fill=black,minimum size=2pt,
                        inner sep=0pt]                    
\foreach \x in {1,2}{
\node at (\x,2) {};};
\foreach \x in {0,1,2,3}{
\node at (\x,1) {};};
\tikzstyle{every node}=[]
  \draw [->,dashed] (0,1) -- (1,2);
\draw [->,dashed] (1,1) -- (2,1);
\draw [->,dashed] (2,1) -- (3,1);
  \draw [->,dashed] (0,1) -- (1,1);
\draw [->] (1,2) -- (2,2);
\draw [->] (2,2) -- (3,1);
\draw [->] (1,1) -- (2,2);
\draw [->,dashed] (1,2) -- (2,1);
\draw (0,1) node [below] {\tiny ${\uu{0}}$};
\draw (1,1) node [below] {\tiny ${\uu{0}}$};
\draw (2,1) node [below] {\tiny ${\uu{0}}$};
\draw (3,1) node [below] {\tiny ${\uu{0}}$};
\draw (1,2) node [above] {\tiny ${\uu{1}}$};
\draw (2,2) node [above] {\tiny ${\uu{1}}$};
\begin{scope}
[xshift=4cm]
  
  \tikzstyle{every node}=[draw,circle,fill=black,minimum size=2pt,
                        inner sep=0pt]                    
\foreach \x in {1}{
\node at (\x,2) {};};
\foreach \x in {0,1,2,3}{
\node at (\x,1) {};};
\tikzstyle{every node}=[]
  \draw [->,dashed] (0,1) -- (1,2);
\draw [->,dashed] (1,1) -- (2,1);
\draw [->,dashed] (2,1) -- (3,1);
  \draw [->,dashed] (0,1) -- (1,1);
\draw [->,dashed] (1,2) -- (2,1);
\draw (0,1) node [below] {\tiny ${\uu{0}}$};
\draw (1,1) node [below] {\tiny ${\uu{0}}$};
\draw (2,1) node [below] {\tiny ${\uu{0}}$};
\draw (3,1) node [below] {\tiny ${\uu{0}}$};
\draw (1,2) node [above] {\tiny ${\uu{1}}$};
\end{scope}

\begin{scope}
[xshift=8cm]
  \tikzstyle{every node}=[draw,circle,fill=black,minimum size=2pt,
                        inner sep=0pt]                    
\foreach \x in {1,2}{
\node at (\x,2) {};};
\foreach \x in {0,1,2,3}{
\node at (\x,1) {};};
\tikzstyle{every node}=[]
  \draw [->,dashed] (0,1) -- (1,2);
\draw [->,dashed] (1,1) -- (2,1);
\draw [->,dashed] (2,1) -- (3,1);
  \draw [->,dashed] (0,1) -- (1,1);
\draw [->] (1,2) -- (2,2);
\draw [->] (2,2) -- (3,1);

\draw (0,1) node [below] {\tiny ${\uu{0}}$};
\draw (1,1) node [below] {\tiny ${\uu{0}}$};
\draw (2,1) node [below] {\tiny ${\uu{0}}$};
\draw (3,1) node [below] {\tiny ${\uu{0}}$};
\draw (1,2) node [above] {\tiny ${\uu{1}}$};
\draw (2,2) node [above] {\tiny ${\uu{1}}$};
\end{scope}
\draw (3.5,1.5) node {$\neq$};
\draw (7.5,1.5) node {$\otimes$};
  \end{tikzpicture}
          \end{center}
            %We will prove this correspondence precisely in [].
\end{ex}

With such dictionary Theorems \ref{thm11} and \ref{thm12} turn out to yield a new and compact proof of the acclaimed Factorization Theorem \cite{KV2} of Koetter/Vardy.

\begin{cor}[Factorization Theorem]\label{factthm}\index{Factorization Theorem}
Any linear trellis is  linearly isomorphic to a product of elementary trellises. 
\end{cor}

\begin{rmk}
Note that the Factorization Theorem is actually stated in a weaker form in \cite{KV2}  since it ``only'' says that any linear trellis $T$ is \underline{isomorphic} to a product of elementary trellises. 
No \underline{linear isomorphism} is  mentioned in \cite{KV2}, as no importance is given to the particular linear structure of the spaces $V_{i}(T)$ (see the preamble of Section $6$ therein). However, a careful read of \cite{KV2} reveals that linear isomorphy is a plain consequence of what is proven therein. In fact, it can be easily seen that a \textit{representation matrix $G$ in product form} of $T$ (Definition $6.1$ in \cite{KV2}) yields a linear isomorphism between $T$ and the product of elementary trellises corresponding to the rows of $G$, and existence of such matrices is precisely what is proven in \cite{KV2}  (Theorem $6.2$). 
Actually we will see that  isomorphic linear trellises must be also linearly isomorphic (Theorem \ref{isolin}), but this can be proven only by using our approach. 
\end{rmk}

Despite the rather intricate proof of the Factorization Theorem given in \cite{KV2}, 
 no other proof has appeared in the literature since then
 (note though that recently, following our new proof, which we first had announced and illustrated in \cite{CB3}, Gluesing-Luerrsen came up with one further proof --- private communication). 
Our new proof is simpler (for example, we do not need elaborate transformations of vertex-labels and investigations on the connection properties of vertices as in \cite{KV2}) and shorter (the core part of it being 
contained in Theorems \ref{thm11} and \ref{thm12}). It is also more explanatory. In fact, it came out as a byproduct of our studies on the structure of linear trellises via label codes. 

Span subcodes and product bases yield an algebraic point of view and methodology which is technically powerful and gives more comprehensive understanding. Our framework allows to go beyond ``just'' proving that $T$ admits one factorization into elementary trellises. Indeed by our approach in the next sections we will be able to give a complete description of trellis factorizations (while the proof of \cite{KV2} does not shed light on what are the possible factorizations of $T$ and how these are encoded in $\SSS(T)$) and consequently to obtain also results concerning the classifications of trellises. 

%%%%%%%%%%%%%%%%%%%%%%%%%%%%%%%%%%%%%%%%%%%%%%%%%%%%%%
%%%%%%%%%%%%%%%%%%%%%%%%%%%%%%%%%%%%%%%%%%%%%%%%%%%%%%
%%%%%%%%%%%%%%%%%%%%%%%%%%%%%%%%%%%%%%%%%%%%%%%%%%%%%%
%%%%%%%%%%%%%%%%%%%%%%%%%%%%%%%%%%%%%%%%%%%%%%%%%%%%%%
%%%%%%%%%%%%%%%%%%%%%%%%%%%%%%%%%%%%%%%%%%%%%%%%%%%%%%
%%%%%%%%%%%%%%%%%%%%%%%%%%%%%%%%%%%%%%%%%%%%%%%%%%%%%%
%%%%%%%%%%%%%%%%%%%%%%%%%%%%%%%%%%%%%%%%%%%%%%%%%%%%%%
%%%%%%%%%%%%%%%%%%%%%%%%%%%%%%%%%%%%%%%%%%%%%%%%%%%%%%
%%%%%%%%%%%%%%%%%%%%%%%%%%%%%%%%%%%%%%%%%%%%%%%%%%%%%%

\subsection{Structure of $\SSS(T)$: correspondence with trellis factorizations and atomic cycles}\label{algprop}
The following theorem shows that there is a correspondence between decompositions $\SSS(T)=\oplus_{i=1}^{r} S^{i}$ satisfying $\sal(T)=\oplus_{i=1}^{r} (S^{i}\cap\sal(T))$ for all $(a,l)$ and factorizations $T\simeq \otimes_{i=1}^{r}T_{i}$, and so it
completes the picture on how the label code encodes the trellis product.
To simplify notation, for $S\leq\SSS(T)$ we will put $$S_{(a,l)}:=S\cap\sal(T)$$ 
 \begin{thm}\label{thm10}
 If  $S^{1},\ldots,S^{r}\leq\SSS(T)$ satisfy $$\sal(T)=\oplus_{i=1}^{r} S^{i}_{(a,l)}$$ for all $(a,l)$ then $S^{i}=\SSS(T(S^{i}))$ for all $i$ 
and
 there exists a linear isomorphism 
 $$f:T\rightarrow \otimes_{i=1}^{r}T(S^{i})$$ 
 such that $f(T(S^{i}))=
 0\otimes\ldots\otimes0\otimes T(S^{i})\otimes0\otimes\ldots\otimes0$ 
 for all $i$. 
 
 Vice versa, given $T_{1},\ldots,T_{r}$ with no simultaneous parallel edges (i.e. 
 $\sum_{i=1}^{r}\dim\SSS_{(a,0)}(T_{i})\leq1$ for all $a\in\ZZ_{n}$), if there exists a linear isomorphism $f:T\rightarrow \otimes_{i=1}^{r}T_{i}$ then  for all $(a,l)$ we have
 $$\sal(T)=\oplus_{i=1}^{r}\SSS_{(a,l)}(f^{-1}(0\otimes\ldots\otimes0\otimes T_{i}\otimes0\otimes\ldots\otimes0))$$ 
 \end{thm}
 \begin{proof}
The ``vice versa'' part is an immediate consequence of Observation \ref{obs10}. So let us go the first part.
  By induction it suffices to consider the case of two subcodes $S',S''\leq\SSS(T)$. 
 Let $\la\in\SSS(T(S'))$. By hypothesis $\la\in S'\oplus S''$. Let $-1\leq l\leq n$ be the least $l$ such that $\la\in S'\oplus S''_{(a,l)}$ for some $(a,l)$. If $l=-1$ then $S''_{(a,l)}=0$, and we are done. So, assume $l\geq0$. Write $\la=\la'+\la''$ for some $\la'\in S'$, $\la''\in S''_{(a,l)}$. By assumption $(a,l)$ is a minimal span of $\la''$. Since $\la\in\SSS(T(S'))$, there exists a cycle $\widetilde{\la}\in S'$ which agrees with $\la$ in the $a$-th edge.  
 Then 
 $$\la-\widetilde{\la}\in\SSS_{(a+1,n-2)}(T)=S'_{(a+1,n-2)}\oplus S''_{(a+1,n-2)}$$
 So $\la'-\widetilde{\la}+\la''=\la'''+\la''''$ for some $\la'''\in S'_{(a+1,n-2)}$, $\la''''\in S''_{(a+1,n-2)}$. But $S'\cap S''=0$ by hypothesis, hence $\la''=\la''''\in S''_{(a+1,n-2)}$, contradicting the minimality of $(a,l)$ for $\la''$. 
Hence $l=-1$ in the first place, and so $\la\in S'$.  Since $S'\subseteq \SSS(T(S'))$ is always true we have proved that $S'=\SSS(T(S'))$. 
Symmetrically,  $S''=\SSS(T(S''))$.
 
 Now, let $T':=T(S'), T'':=T(S'')$. The map $$(\la',\la'')\overset{G}{\mapsto}\la'+\la''$$ 
 is an isomorphism and satisfies $G(\sal(T')\times\sal(T''))=\sal(T)$  since  $S_{(a,l)}'\oplus S_{(a,l)}''=\sal(T)$.
On the other hand, by Observation \ref{obs10} we have a natural isomorphism  
$$F:\SSS(T')\times\SSS(T'')\rightarrow\SSS(T'\otimes T'')$$ satisfying $F(\sal(T')\times\sal(T''))=\sal(T'\times T'')$. 
So  
$F\circ G^{-1}(\sal(T))=\sal(T'\times T'')$ for all $(a,l)$. Thus $F\circ G^{-1}$ and its inverse are trellis maps. This yields a linear isomorphism from $T$ to $T'\times T''$ which by construction satisfies our statement.
 \end{proof}
 
By the defining property of product bases and the above theorem, if $\BB=\BB^{1}\sqcup\ldots\sqcup\BB^{r}$ is a partition of a product basis of $T$  then   $$T\simeq \otimes_{i=1}^{r}T(\lb\BB^{i}\rb)$$
 This is a more general formulation of the dictionary given by Theorem \ref{thm13}.
The following interesting corollary also follows.
\begin{cor}
 A product basis of $T'\leq T$  extends to a product basis of $T$ if and only if there exists $T''$
 and a linear isomorphism $f:T\rightarrow T'\times T''$ such that $f(T')=T'\times 0$  (or equivalently, if and only if there exists $S\leq \SSS(T)$ such that $\sal(T)=\sal(T')\oplus S_{(a,l)}$ for all $(a,l)$).
\end{cor}

The condition $T\simeq T'\otimes T''$ alone is not sufficient for a product basis of $T'\leq T$ to extend to one of $T$, even though that condition is equivalent to $T'\simeq T'''$ for some $T'''\leq T$ whose product basis extends. In other words, isomorphic subtrellises can contribute differently to the structure of $T$.

\begin{ex}\label{ex22}
Let $T$ 
be as below. Then 
$$T(\lb(\uu{0}\ms\uu{01}\ms\uu{01},\bm{0})\rb)\simeq T(\lb(\uu{0}\ms\uu{10}\ms\uu{10},\bm{0})\rb)$$ 
but  $\{(\uu{0}\ms\uu{01}\ms\uu{01},\bm{0})\}$ cannot be extended to a product basis of $T$, 
while $\{(\uu{0}\ms\uu{10}\ms\uu{10},\bm{0})\}$ can. 
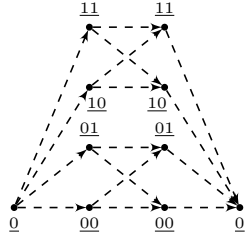
\begin{center}
 \begin{tikzpicture}[yscale=.8,>=latex',shorten >=.9pt, shorten <=1.4pt, line width=.6pt]
  \tikzstyle{every node}=[draw,circle,fill=black,minimum size=2pt,
                        inner sep=0pt]                    
\foreach \x in {1,2}{
\node at (\x,2) {};};
\foreach \x in {1,2}{
\node at (\x,3) {};};
\foreach \x in {1,2}{
\node at (\x,4) {};};
\foreach \x in {0,1,2,3}{
\node at (\x,1) {};};
\tikzstyle{every node}=[]
  \draw [->,dashed] (0,1) -- (1,1);
\draw [->,dashed] (1,1) -- (2,1);
\draw [->,dashed] (2,1) -- (3,1);
\draw [->,dashed] (0,1) -- (1,2);
\draw [->,dashed] (1,2) -- (2,2);
\draw [->,dashed] (2,2) -- (3,1);
\draw [->,dashed] (1,1) -- (2,2);
\draw [->,dashed] (1,2) -- (2,1);

  \draw [->,dashed] (0,1) -- (1,3);
\draw [->,dashed] (1,3) -- (2,3);
\draw [->,dashed] (2,3) -- (3,1);
\draw [->,dashed] (0,1) -- (1,4);
\draw [->,dashed] (1,4) -- (2,4);
\draw [->,dashed] (2,4) -- (3,1);
\draw [->,dashed] (1,3) -- (2,4);
\draw [->,dashed] (1,4) -- (2,3);

\draw (0,1) node [below] {\tiny ${\uu{0}}$};
\draw (1,1) node [below] {\tiny ${\uu{00}}$};
\draw (2,1) node [below] {\tiny ${\uu{00}}$};
\draw (3,1) node [below] {\tiny ${\uu{0}}$};
\draw (1,2) node [above] {\tiny ${\uu{01}}$};
\draw (2,2) node [above] {\tiny ${\uu{01}}$};

\draw (1.1,3) node [below] {\tiny ${\uu{10}}$};
\draw (1.9,3) node [below] {\tiny ${\uu{10}}$};
\draw (1,4) node [above] {\tiny ${\uu{11}}$};
\draw (2,4) node [above] {\tiny ${\uu{11}}$};
      \end{tikzpicture}
    \end{center}
\end{ex}
We can get more insight into this phenomenon by adopting the following point of view. 
The additive structure of a linear trellis $T$ allows for the decomposition of cycles into sums of smaller cycles, i.e. cycles with shorter span length. However, some cycles in $T$ cannot be written as sums of other smaller cycles. In other words, some cycles are \textit{atomic}. More precisely, we say that a nonzero cycle $\la\in \SSS(T)$ is \textit{atomic}\index{cycle!atomic} if 
$\la\notin\SSS_{\ell(\la)-1}(T)$. 
We now have:

\begin{obs}\label{obs22}
A cycle $\la\in\SSS(T)$ is atomic if and only if it belongs to some product basis of $T$. 
\end{obs}
\begin{proof}
Assume $\la\in\SSS_{(a,\ell(\la))}(T)$ is atomic. Then $\la\notin\SSS_{\ell(\la)-1}(T)$. In particular $\la\notin\SSS_{<(a,\ell(\la))}(T)$. By $5)$ of Observation \ref{spansub1} then $\la$ belongs to some product basis.
Vice versa, 
if $\la$ belongs to a product basis of $T$ then by  Observation \ref{spansub1} again we conclude that $\la\notin\SSS_{\ell(\la)-1}(T)$, i.e. $\la$ is atomic.
\end{proof}

In Example \ref{ex22} we see that $(\uu{0}\ms\uu{10}\ms\uu{10},\bm{0})$ 
is atomic while $(\uu{0}\ms\uu{01}\ms\uu{01},\bm{0})$ 
is  not: the different contributions of the associated (elementary) trellises to the structure of $T$ can now be intrinsically explained by the atomic property. 
Note though that while the atomic property characterizes cycles belonging to product bases, in general it is not true that a linearly independent set of atomic cycles is part of a product basis. Nevertheless, the following theorem holds:

\begin{thm}\label{thm23} A basis $\BB$ of $\SSS(T)$ is  a product basis if and only if it minimizes total span length.
\end{thm}

It is easily checked in Examples \ref{ex5cr} and \ref{ex22}    that the product bases are precisely those bases that minimize the total span length. This theorem enables to find product bases (and so elementary trellis factorizations) of $T$ via bases of minimum total span lenght. We omit its proof here since it is rather technical and it requires extra machinery which goes beyond the scope of this paper. The reader is referred to \cite{C} for the proof and more on this.

%%%%%%%%%%%%%%%%%%%%%%%%%%%%%%%%%%%%%%%%%%%%%%%%%%%%%%
%%%%%%%%%%%%%%%%%%%%%%%%%%%%%%%%%%%%%%%%%%%%%%%%%%%%%%
%%%%%%%%%%%%%%%%%%%%%%%%%%%%%%%%%%%%%%%%%%%%%%%%%%%%%%
%%%%%%%%%%%%%%%%%%%%%%%%%%%%%%%%%%%%%%%%%%%%%%%%%%%%%%
%%%%%%%%%%%%%%%%%%%%%%%%%%%%%%%%%%%%%%%%%%%%%%%%%%%%%%
%%%%%%%%%%%%%%%%%%%%%%%%%%%%%%%%%%%%%%%%%%%%%%%%%%%%%%
%%%%%%%%%%%%%%%%%%%%%%%%%%%%%%%%%%%%%%%%%%%%%%%%%%%%%%
%%%%%%%%%%%%%%%%%%%%%%%%%%%%%%%%%%%%%%%%%%%%%%%%%%%%%%
%%%%%%%%%%%%%%%%%%%%%%%%%%%%%%%%%%%%%%%%%%%%%%%%%%%%%%

\subsection{Characterization of linear trellis isomorphy} 

In the previous subsections we have seen how structural properties of the ordered family of span subcodes of a linear trellis $T$ correspond to factorizations of $T$. We now prove a central theorem of this paper which describes how isomorphy is encoded by the same family. Our theorem, which yields an effective method for checking whether or not two linear trellises are isomorphic, will be crucial in determining all factorizations of linear trellises and in our later results on the classification of nonmergeable/minimal linear trellises (Sections \ref{factpap} and \ref{charactpap}).
We actually prove our theorem for linear isomorphy, but we shall see a posteriori that the same result is then true for nonlinear isomorphy  too (see Theorem \ref{isolin}). First, an easy ``mathematical folklore'' lemma:

\begin{lem}\label{lem12}
Assume we have linear maps $V\overset{f}{\rightarrow}U$ and $W\overset{g}{\rightarrow}U$ of vector spaces $V$, $W$, $U$, such that $f(V)=g(W)$ and $\dim{} V=\dim{} W$. Then there exists an isomorphism $V\overset{h}{\rightarrow}W$ such that $g\circ h=f$, i.e. the below diagram commutes. 
$$
\begin{xy}
(0,6)*+{V}="v"; %
(0,-6)*+{W}="w"; (20,0)*+{U}="u";%
{\ar "v";"u"}?*!/_3mm/{f};
{\ar "w";"u"}?*!/^3mm/{g};
{\ar@{-->} "v";"w"};?*!/^3mm/{h};
\end{xy}
$$
\end{lem}
\begin{proof}
Take a basis $u_{1},\ldots,u_{m}$ of $f(V)=g(W)$. Take linearly independent elements $v_{1},\ldots,v_{m}\in V$ and linearly independent elements $w_{1},\ldots,w_{m}\in W$ such that $f(v_{i})=u_{i}=g(w_{i})$ for $i=1,\ldots,m$. Obviously $\langle v_{1},\ldots,v_{m}\rangle\cap\textnormal{Ker}\,f=0$, so that we can complete $v_{1},\ldots,v_{m}$ to a basis $v_{1},\ldots,v_{n}$ of $V$ such that $v_{m+1},\ldots,v_{n}\in\textnormal{Ker}\,f$. Similarly, we obtain a basis $w_{1},\ldots,w_{n}$ of $W$ such that $w_{m+1},\ldots,w_{n}\in\textnormal{Ker}\,g$. Then the linear map $h$ defined by $h(v_{i}):=w_{i}$ obviously satisfies our statement.
\end{proof}

\begin{thm}\label{thm15} Two linear trellises $T$ and $T'$ (of same length) are linearly isomorphic if and only if for all  $(a,l)$ the equalities  
\begin{align*}
\dim{} \SSS_{(a,l)}(T)&=\dim{} \SSS_{(a,l)}(T')\\
C_{(a,l)}(T)&=C_{(a,l)}(T')
\end{align*} 
hold true.
\end{thm}
\begin{proof} The ``only if'' part is trivial, so we need only to prove the ``if'' part.
We need to show that there exists a linear isomorphism $F:\SSS(T)\rightarrow\SSS(T')$ such that $F(\SSS_{(a,l)}(T))=\SSS_{(a,l)}(T')$ for all $(a,l)$, and $L(F(\la))=L(\la)$ for all cycles $\la$. For increasing $l=-1,0,\ldots,n$ we will construct for all $a\in\ZZ_{n}$ linear isomorphisms 
$$F_{(a,l)}:\SSS_{(a,l)}(T)\rightarrow\SSS_{(a,l)}(T')$$  such that 
\begin{align*}L(F_{(a,l)}(\la))&=L(\la)\\
F_{(a,l)}|_{\SSS_{(a',l')}(T)}&=F_{(a',l')}
\end{align*} 
for all $\la\in\SSS_{(a,l)}(T)$ and all $(a',l')\leq(a,l)$. 
The map $F_{\ZZ_{n}}$ will then clearly be our sought $F$. For the first step $l=-1$, $F_{\emptyset}$ can   only be the zero map, and there is nothing to prove.

Now, for $l\geq1$ assume we have constructed for each $a\in\ZZ_{n}$ and $l'\leq l-1$ linear isomorphisms $F_{(a,l')}$  satisfying the above properties, and let $\BB$ be a product basis of $T$. Fix $a\in\ZZ_{n}$. The first step towards the construction of our sought map $F_{(a,l)}$ is to extend all the isomorphisms $F_{(a',l')}$, for $(a',l')<(a,l)$, to an isomorphism $F_{<(a,l)}:\SSS_{<(a,l)}(T)\rightarrow\SSS_{<(a,l)}(T')$ which also preserves edge-labels. 

So, define $F_{<(a,l)}$ 
by  
$$F_{<(a,l)}(\la):=F_{[\la]}(\la)$$ for all $\la\in\BB\cap\SSS_{<(a,l)}(T)$. Note that if an element $\la\in\BB$ satisfies $[\la]\leq(a',l')<(a,l)$ then,  by our assumption, $F_{[\la]}(\la)=F_{(a',l')}(\la)$, so that $F_{<(a,l)}$ extends $F_{(a',l')}$ for all $(a',l')<(a,l)$. It follows that $F_{<(a,l)}$ is surjective, and hence that 
$$\dim{} \SSS_{<(a,l)}(T')\leq\dim{} \SSS_{<(a,l)}(T)$$ 
Symmetrically, $\dim{} \SSS_{<(a,l)}(T)\leq\dim{} \SSS_{<(a,l)}(T')$. Therefore $F_{<(a,l)}$ is an isomorphism. Obviously by construction we also have that $L(F_{<(a,l)}(\la))=L(\la)$ for all $\la\in\SSS_{<(a,l)}$.

We will now construct $F_{(a,l)}$ by extending $F_{<(a,l)}$. 
Let 
$$C:=L(\SSS_{<(a,l)}(T))=L(\SSS_{<(a,l)}(T'))$$
where the equality follows from the hypothesis. 
We then have two naturally induced maps 
\begin{align*}
\overline{L}&:\SSS_{(a,l)}(T)/\SSS_{<(a,l)}(T)\rightarrow\FF^{n}/C\\
 \overline{L}&:\SSS_{(a,l)}(T')/\SSS_{<(a,l)}(T')\rightarrow\FF^{n}/C
 \end{align*} 
 We have seen that $\dim \SSS_{<(a,l)}(T)= \dim \SSS_{<(a,l)}(T')$. By hypothesis %we have 
 $\dim \SSS_{(a,l)}(T) =\dim \SSS_{(a,l)}(T')$, 
and so also $\dim \SSS_{(a,l)}(T)/\SSS_{<(a,l)}(T)= \dim \SSS_{(a,l)}(T')/\SSS_{<(a,l)}(T')$. 
Moreover, since by hypothesis $C_{(a,l)}(T)=C_{(a,l)}(T')$, we have that 
$$\overline{L}(\SSS_{(a,l)}(T)/\SSS_{<(a,l)}(T))=\overline{L}(\SSS_{(a,l)}(T')/\SSS_{<(a,l)}(T'))$$
Hence by Lemma \ref{lem12} we can find an isomorphism 
$$G:\SSS_{(a,l)}(T)/\SSS_{<(a,l)}(T)\rightarrow \SSS_{(a,l)}(T')/\SSS_{<(a,l)}(T')$$ 
such that $\overline{L}\circ G=\overline{L}$, where $G=\overline{g}$ for some linear map $g:\SSS_{(a,l)}(T)\rightarrow \SSS_{(a,l)}(T')$ (it is easily seen that any linear map from a quotient of vector spaces $V_{1}/V_{2}$ to another quotient $W_{1}/W_{2}$ is the reduction of a linear map from $V_{1}$ to $W_{1}$). The situation is depicted in the following commutative diagram.
$$
\begin{xy}
(10,10)*+{\SSS_{(a,l)}(T)}="v1"; %
(10,-10)*+{\SSS_{(a,l)}(T')}="v2"; %
(40,10)*+{\SSS_{(a,l)}(T)/\SSS_{<(a,l)}(T)}="v"; %
(40,-10)*+{\SSS_{(a,l)}(T')/\SSS_{<(a,l)}(T')}="w"; 
(80,0)*+{\FF^{n}/C}="u";%
{\ar "v";"u"}?*!/_3mm/{\overline{L}};
{\ar "w";"u"}?*!/^3mm/{\overline{L}};
{\ar@{-->} "v";"w"};?*!/^3mm/{\overline{g}};
{\ar "v1";"v"};
{\ar "v2";"w"};
{\ar@{-->} "v1";"v2"};;?*!/^3mm/{{g}};
\end{xy}
$$
Now, for $\la\in \SSS_{(a,l)}(T)$ we have $\overline{L}(\overline{g}(\overline{\la}))=\overline{L}(\overline{\la})$ (where $\overline{\la}$ is the reduction of $\la$ modulo $\SSS_{<(a,l)}(T)$) and therefore $$L(g(\la))-L(\la)\in C=L(\SSS_{<(a,l)}(T'))$$
 i.e. $L(g(\la))=L(\la)+L(\widetilde{\la})$ for some $\widetilde{\la}\in \SSS_{<(a,l)}(T')$. So, putting $$F_{(a,l)}(\la):=F_{<(a,l)}(\la)$$ for each basis element $\la$ with span $[\la]<(a,l)$, and $$F_{(a,l)}(\la):=g(\la)-\widetilde{\la}$$ for each basis element $\la$ with $[\la]=(a,l)$, we get that $F_{(a,l)}$ extends $F_{(a',l')}$ whenever $(a',l')\leq(a,l)$, and that $L(F_{(a,l)}(\la))=L(\la)$ for all $\la\in\SSS_{(a,l)}(T)$. 

It remains to prove that the map $F_{(a,l)}:\SSS_{(a,l)}(T)\rightarrow\SSS_{(a,l)}(T')$ so constructed is an isomorphism. By the hypothesis it is sufficient to show that $F_{(a,l)}$ is injective. So assume that $F_{(a,l)}(\la')=0$ for $\la'\in\SSS_{(a,l)}(T)$. 
We can write $$\la'=\sum_{\la\in\BB|[\la]=(a,l)}x_{\la}\la+\bm{\mu}$$ for some $x_{\la}\in\FF$ and $\bm{\mu}\in\SSS_{<(a,l)}(T)$. Applying $F_{(a,l)}$ then
$$\sum_{\la\in\BB|[\la]=(a,l)}x_{\la}g(\la)-\sum_{\la\in\BB|[\la]=(a,l)}x_{\la}\widetilde{\la}+F_{<(a,l)}(\bm{\mu})=0$$ 
Reducing modulo $\SSS_{<(a,l)}(T')$ and $\SSS_{<(a,l)}(T)$ we thus get that 
$$\overline{g}(\sum_{\la\in\BB|[\la]=(a,l)}x_{\la}\overline{\la})={0}$$ 
and so since $\overline{g}=G$ is an isomorphism  
$$\sum_{\la\in\BB|[\la]=(a,l)}x_{\la}\overline{\la}={0}$$
 But by $4)$ of Observation \ref{spansub1} we know  that $\{\overline{\la}|\la\in\BB, [\la]=(a,l)\}$ is a basis of $\SSS_{(a,l)}(T)/\SSS_{<(a,l)}(T)$. So we conclude that $x_{\la}=0$ for all $\la$. Then $\la'=\bm{\mu}$ and $F_{<(a,l)}(\bm{\mu})=F_{(a,l)}(\la')=0$, hence $\bm{\mu}=0$ by injectivity of $F_{<(a,l)}$. Hence $F_{(a,l)}$ satisfies all the wanted properties, and therefore the proof is finished.
\end{proof}
 
For one-to-one trellises linear isomorphy reduces to equality of the represented span subcodes. In fact, $T$ is one-to-one if and only if $L:\SSS(T)\rightarrow C(T)$ is one-to-one, in which case $\dim\sal(T)=\dim C_{(a,l)}(T)$ for all $(a,l)$, so that we get the following corollary.

\begin{cor}\label{linisocor}
Two  one-to-one linear trellises $T, T'$ are linearly isomorphic if and only if $C_{(a,l)}(T)=C_{(a,l)}(T')$ for all $(a,l)$.
\end{cor}

%%%%%%%%%%%%%%%%%%%%%%%%%%%%%%%%%%%%%%%%%%%%%%%%%%%%%%
%%%%%%%%%%%%%%%%%%%%%%%%%%%%%%%%%%%%%%%%%%%%%%%%%%%%%%
%%%%%%%%%%%%%%%%%%%%%%%%%%%%%%%%%%%%%%%%%%%%%%%%%%%%%%
%%%%%%%%%%%%%%%%%%%%%%%%%%%%%%%%%%%%%%%%%%%%%%%%%%%%%%
%%%%%%%%%%%%%%%%%%%%%%%%%%%%%%%%%%%%%%%%%%%%%%%%%%%%%%
%%%%%%%%%%%%%%%%%%%%%%%%%%%%%%%%%%%%%%%%%%%%%%%%%%%%%%
%%%%%%%%%%%%%%%%%%%%%%%%%%%%%%%%%%%%%%%%%%%%%%%%%%%%%%
%%%%%%%%%%%%%%%%%%%%%%%%%%%%%%%%%%%%%%%%%%%%%%%%%%%%%%
%%%%%%%%%%%%%%%%%%%%%%%%%%%%%%%%%%%%%%%%%%%%%%%%%%%%%%

\subsection{Uniqueness of linear structure of linear trellises}

A trellis $T$ is said to be \textit{linearizable} if all its vertex sets $V_{i}(T)$ admit a vector space structure that make $T$ a linear trellis (i.e., more concretely, if we can label all vertices with words over $\FF$ so that $\SSS(T)$ is a linear code). While in  \cite{KV2} it was shown how to find out if a trellis is linearizable, the fundamental question whether a linearizable trellis admits  an essentially unique linear structure has not been addressed so far. This question can be rephrased as: are isomorphic linear trellises essentially equal? Since for linear trellises to be ``essentially equal''  means to be linearly isomorphic, the question is then whether isomorphy of linear trellises implies also linear isomorphy.

Thanks to Theorem \ref{thm15} we can now establish that for linear trellises there is no difference between being linearly isomorphic and being isomorphic. 
In particular, in Theorem \ref{thm15} we can drop the ``linearly'' adverb with no need to change anything.

\begin{thm}\label{isolin}
Let $T$, $T'$ be linear trellises. If $T$ and $T'$ are isomorphic then they are also linearly isomorphic.
\end{thm}
\begin{proof}
Assume we have a nonlinear  trellis isomorphism $f:T\rightarrow T'$ given by $f_{i}:V_{i}(T)\rightarrow V_{i}(T')$, ${i\in\ZZ_{n}}$. Let $F=\SSS(f)$.
Define a map from $\SSS(T)$ to $\SSS(T')$ by $$\la\mapsto F(\la+F^{-1}(\bm{0}))$$ 
for all cycles $\la\in\SSS(T)$. Clearly this map preserves edge labels and it is injective, as $F$ and so $F^{-1}$ do. Now, let $\la\in\sal(T)$, for some $0\leq l\leq n-1$. Then $\nu_{i}(\la)=0$ for all $i\notin (a,a+l]$. So, for  $i\notin (a,a+l]$ we have that $$\nu_{i}(\la+F^{-1}(\bm{0}))=\nu_{i}(F^{-1}(\bm{0}))=f_{i}^{-1}(0_{i})$$
 and therefore  $\nu_{i}(F(\la+F^{-1}(\bm{0})))=f_{i}(\nu_{i}(\la+F^{-1}(\bm{0})))=0_{i}$. So our map is an injective and edge-label preserving map from $\sal(T)$ to $\sal(T')$. Symmetrically, there is also such a map from $\sal(T')$ to $\sal(T)$. From this follows that the hypothesis of Theorem \ref{thm15} are satisfied, hence the two trellises are linearly isomorphic.
\end{proof}

The essential uniqueness of the linear structure of linearizable trellises now follows:

\begin{thm}\label{linstr}
Let $T=(\sqcup V_{i}(T),\sqcup E_{i}(T))$ be a trellis. Assume that each vertex set $V_{i}(T)$ has two addition operations 
$+_{1}$, $+_{2}$, and two scalar multiplication operations $\cdot_{1}$, $\cdot_{2}$, 
such that both $(+_{1},\cdot_{1})$ and  $(+_{2},\cdot_{2})$ make $T$ linear. 
Then the two resulting linear trellises $T_{1}$ and $T_{2}$ are linearly isomorphic.
\end{thm}

The reader must  be aware  that Theorem \ref{isolin} is not saying that an isomorphism of linear trellises is also a linear isomorphism. Equivalently, Theorem \ref{linstr} is not saying that the identity map of $T$  necessarily results in  a linear isomorphism between $T_{1}$ and $T_{2}$. Counterexamples can be  easily constructed indeed. 
Still,  the proofs of the above results explicitly tell us  how to  construct a linear isomorphism from the given isomorphism. 

\begin{rmk}
There is a suggestive parallel between vector spaces and linear trellises. Any $K$-vector space $V$ has a basis and thus decomposes into parts of dimension one ($V\simeq K^{n}$), and if two $K$-vector spaces are equivalent as sets (i.e. have the same cardinality) then they must be linearly isomorphic. Similarly, any linear trellis $T$ has a product basis, so that the Factorization Theorem holds ($T\simeq \otimes_{i=1}^{r}T_{i}$, with $\dim\SSS(T_{i})=1$ for all $i$), and we also have Theorem \ref{isolin}. However, while for vector spaces such fundamental properties are easy to prove, for linear trellises this is not at all the case.
\end{rmk}

When the trellises under consideration are one-to-one we get a much stronger result, which turns out to be even much simpler to prove.

\begin{thm}\label{thm16} Let $f:T\rightarrow T'$ be an isomorphism of one-to-one linear trellises. Then $f$ is linear.
\end{thm}
\begin{proof}
We need to show that $\SSS(f)$ is linear. Since  $f$ is a trellis morphism we have that $L\circ \SSS(f)=L'$, where $L$ and $L'$ are the edge-label sequence maps of cycles respectively in $T$ and $T'$. 
But since $T$ and $T'$ are one-to-one and linear, $L$ and $L'$ must be injective linear maps, and so  $\SSS(f)=L^{-1}\circ L'$ is linear (where the domain of $L^{-1}$ is $C(T)=L'(\SSS(T))$).
\end{proof}

\begin{cor}\label{cor13}
Assume we are in the same situation as in Theorem \ref{linstr}. Assume also that $T$ is one-to-one. Then the identity map of $T$ is a linear isomorphism between $T_{1}$ and $T_{2}$.
\end{cor}

The above corollary tells us that there is {literally only one} vector space structure 
on each space $V_{i}(T)$ that can possibly make a one-to-one trellis $T$ linear, so that the situation is very  rigid in that case. In other words, any two different  labelings of vertices of a one-to-one trellis $T$ that make $T$  linear  are one the linear transformation of the other.

The function-theoretical argument used in the above proof leads us also to a necessary and sufficient condition for a one-to-one trellis to admit a linear structure. 

\begin{obs}\label{obs14}
Let $T$ be a one-to-one trellis. Put $G_{i}:=\nu_{i}\circ L^{-1}$. Then $T$ is linearizable if and only if for all $\bm{v},\bm{v}',\bm{w},\bm{w}'\in C(T)$, $\alpha\in\FF$, $i\in\ZZ_{n}$ such that $G_{i}(\bm{v})=G_{i}(\bm{v}')$ and $G_{i}(\bm{w})=G_{i}(\bm{w}')$ the following equalities hold
\begin{align*}
G_{i}(\bm{v}+\bm{w})&=G_{i}(\bm{v}'+\bm{w}')\\
G_{i}(\alpha\vv)&=G_{i}(\alpha\vv')
\end{align*}
\end{obs}
\begin{proof}
If $T$ is linear then $G_{i}$ is linear, and so it obviously satisfies the stated conditions. Vice versa, if the stated conditions are satisfied then they  induce a unique vector space structure on each $V_{i}(T)$ such that $G_{i}$ is a linear map. It is then easy to check that the induced structures makes $T$ linear.
\end{proof}

Note that in \cite{KV2} no theoretical characterization is given of the linearizable property for (one-to-one) trellises. Instead, Koetter/Vardy present therein an algorithm that halts if $T$ is non-linearizable or otherwise outputs a complete set of vertex-labels which make $T$ linear.

%
%%%%%%%%%%%%%%%%%%%%%%%%
%%%%%%%%%%%%%%%%%%%%%%%%
%%%%%%%%%%%%%%%%%%%%%%%%
%%%%%%%%%%%%%%%%%%%%%%%%
%%%%%%%%%%%%%%%%%%%%%%%%
%%%%%%%%%%%%%%%%%%%%%%%%
%%%%%%%%%%%%%%%%%%%%%%%%
%%%%%%%%%%%%%%%%%%%%%%%%
%%%%%%%%%%%%%%%%%%%%%%%%
%%%%%%%%%%%%%%%%%%%%%%%%

\subsection{Group trellises: remarks on their structure and on extending results from the linear case}\label{grouppap}
Replacing $\FF$ with a group $G$ in the definition of trellis yields \textit{trellises over groups}. Definitions and statements 
%NOT Definition/statements
for trellises over fields can be thus translated (right away or after the appropriate adaptation to group theoretical language) into definitions/statements for trellises over groups.
In particular, we can talk of \textit{group trellises}, i.e. trellises with a group structure over $G$ and which represent group codes (i.e. subgroups of $G^{n}$).

Now, while in  \cite{KV2} (first Remark of Section $3$ therein) and \cite{KV} (first Remark of Section IV therein) it is stated that all the results therein ``hold essentially without change for group trellises over an abelian group'', the Factorization Theorem does not hold for (abelian) group trellises if we do not revise  the definition of elementary group trellis described therein (which  naturally generalizes the one for the linear case and on the base of which such a trellis 
must represent a group code). 
For example, 
consider the below group trellis over the cyclic group $\ZZ_{4}$,  
where edge-labels are identified by the  arrow filling pattern (dashed $\equiv0$; dash-dotted $\equiv1$; dotted $\equiv2$; full $\equiv3$).
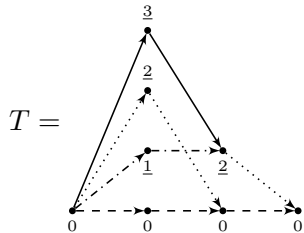
\begin{center}
 \begin{tikzpicture}[yscale=.8,>=latex',shorten >=.9pt, shorten <=1.4pt, line width=.6pt]
  \tikzstyle{every node}=[draw,circle,fill=black,minimum size=2pt,
                        inner sep=0pt]                    
\foreach \x in {0,1,2,3}{
\node at (\x,0) {};};
\foreach \x in {1,2}{
\node at (\x,1) {};};
\foreach \x in {1}{
\node at (\x,2) {};};
\foreach \x in {1}{
\node at (\x,3) {};};

\tikzstyle{every node}=[]
\draw [->,dashed] (0,0) -- (1,0);
\draw [->,dashdotted] (0,0) -- (1,1);
\draw [->,dotted] (0,0) -- (1,2);
\draw [->] (0,0) -- (1,3);

\draw [->,dashed] (1,0) -- (2,0);
\draw [->,dashdotted] (1,1) -- (2,1);

\draw [->,dotted] (1,2) -- (2,0);
\draw [->] (1,3) -- (2,1); 

\draw [->,dotted] (2,1) -- (3,0); 
\draw [->,dashed] (2,0) -- (3,0);

\draw (0,0) node [below] {\tiny ${\uu{0}}$};
\draw (1,0) node [below] {\tiny ${\uu{0}}$};
\draw (1,1) node [below] {\tiny ${\uu{1}}$};
\draw (1,2) node [above] {\tiny ${\uu{2}}$};
\draw (1,3) node [above] {\tiny ${\uu{3}}$};

\draw (2,0) node [below] {\tiny ${\uu{0}}$};
\draw (2,1) node [below] {\tiny ${\uu{2}}$};

\draw (3,0) node [below] {\tiny ${\uu{0}}$};
\draw (-.5,1.5) node  {$T=$};

      \end{tikzpicture}
    \end{center}
If  $T\sim T_{1}\otimes\ldots\otimes T_{r}$ for some group trellises $T_{1},\ldots, T_{r}$, then $C(T_{i})=\langle112\rangle$ for some $i$, since $\langle112\rangle=\sum_{i=1}^{r}C(T^{i})$. Then it's easy to see that $T_{i}\sim T$, and so also $T_{j}=0$ for all $j\neq i$.  On the other hand $T$ is not elementary (according to the mentioned definition an elementary group trellis consists of cycles that run disjointly inside its defining span while coinciding with the zero cycle outside it, like in the linear case).  In fact the following is the only conventional elementary group trellis which represents $\langle112\rangle$:
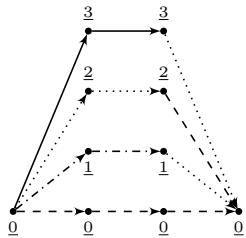
\begin{center}
 \begin{tikzpicture}[yscale=.8,>=latex',shorten >=.9pt, shorten <=1.4pt, line width=.6pt]
  \tikzstyle{every node}=[draw,circle,fill=black,minimum size=2pt,
                        inner sep=0pt]                    
\foreach \x in {0,1,2,3}{
\node at (\x,0) {};};
\foreach \x in {1,2}{
\node at (\x,1) {};};
\foreach \x in {1,2}{
\node at (\x,2) {};};
\foreach \x in {1,2}{
\node at (\x,3) {};};

\tikzstyle{every node}=[]
\draw [->,dashed] (0,0) -- (1,0);
\draw [->,dashed] (1,0) -- (2,0);
\draw [->,dashed] (2,0) -- (3,0); 

\draw [->,dashdotted] (0,0) -- (1,1);
\draw [->,dashdotted] (1,1) -- (2,1);
\draw [->,dotted] (2,1) -- (3,0); 

\draw [->,dotted] (0,0) -- (1,2);
\draw [->,dotted] (1,2) -- (2,2);
\draw [->,dashed] (2,2) -- (3,0); 

\draw [->] (0,0) -- (1,3);
\draw [->] (1,3) -- (2,3); 
\draw [->,dotted] (2,3) -- (3,0); 

\draw (0,0) node [below] {\tiny ${\uu{0}}$};
\draw (1,0) node [below] {\tiny ${\uu{0}}$};
\draw (1,1) node [below] {\tiny ${\uu{1}}$};
\draw (1,2) node [above] {\tiny ${\uu{2}}$};
\draw (1,3) node [above] {\tiny ${\uu{3}}$};

\draw (2,0) node [below] {\tiny ${\uu{0}}$};
\draw (2,1) node [below] {\tiny ${\uu{1}}$};
\draw (2,2) node [above] {\tiny ${\uu{2}}$};
\draw (2,3) node [above] {\tiny ${\uu{3}}$};

\draw (3,0) node [below] {\tiny ${\uu{0}}$};

      \end{tikzpicture}
    \end{center}
Nevertheless, we have that $T\sim T_{1}\otimes T_{2}$, where
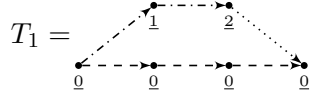
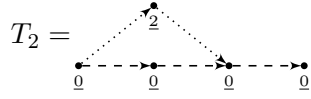
\begin{center}
 \begin{tikzpicture}[yscale=.8,>=latex',shorten >=.9pt, shorten <=1.4pt, line width=.6pt]
%  \tikzstyle{every node}=[draw,circle,fill=black,minimum size=2pt,
%                        inner sep=0pt]                    
%\foreach \x in {0,1,2,3}{
%\node at (\x,0) {};};
%\foreach \x in {1,2}{
%\node at (\x,1) {};};
%\foreach \x in {1}{
%\node at (\x,2) {};};
%\foreach \x in {1}{
%\node at (\x,3) {};};
%
%\tikzstyle{every node}=[]
%\draw [->,dashed,black] (0,0) -- (1,0);
%\draw [->,dashdotted,black] (0,0) -- (1,1);
%\draw [->,dotted,black] (0,0) -- (1,2);
%\draw [->,red] (0,0) -- (1,3);
%
%\draw [->,dashed,black] (1,0) -- (2,0);
%\draw [->,dashdotted,black] (1,1) -- (2,1);
%
%\draw [->,dotted,black] (1,2) -- (2,0);
%\draw [->,red] (1,3) -- (2,1); 
%
%\draw [->,dotted,black] (2,1) -- (3,0); 
%\draw [->,dashed,black] (2,0) -- (3,0); 
%
%
%\draw (0,0) node [below] {\tiny ${\uu{0}}$};
%\draw (1,0) node [below] {\tiny ${\uu{0}}$};
%\draw (1,1) node [below] {\tiny ${\uu{1}}$};
%\draw (1,2) node [above] {\tiny ${\uu{2}}$};
%\draw (1,3) node [above] {\tiny ${\uu{3}}$};
%
%\draw (2,0) node [below] {\tiny ${\uu{0}}$};
%\draw (2,1) node [below] {\tiny ${\uu{2}}$};
%
%
%\draw (3,0) node [below] {\tiny ${\uu{0}}$};
%%\draw (-.5,1.5) node  {$T=$};
%\draw (4,1.5) node  {$\sim$};
  \tikzstyle{every node}=[draw,circle,fill=black,minimum size=2pt,
                        inner sep=0pt]                    
\foreach \x in {0,1,2,3}{
\node at (\x,0) {};};
\foreach \x in {1,2}{
\node at (\x,1) {};};

\tikzstyle{every node}=[]
\draw [->,dashed,black] (0,0) -- (1,0);
\draw [->,dashdotted,black] (0,0) -- (1,1);
\draw [->,dashed,black] (1,0) -- (2,0);
\draw [->,dashdotted,black] (1,1) -- (2,1);

\draw [->,dotted,black] (2,1) -- (3,0); 
\draw [->,dashed,black] (2,0) -- (3,0);

\draw (0,0) node [below] {\tiny ${\uu{0}}$};
\draw (1,0) node [below] {\tiny ${\uu{0}}$};
\draw (1,1) node [below] {\tiny ${\uu{1}}$};
\draw (2,0) node [below] {\tiny ${\uu{0}}$};
\draw (2,1) node [below] {\tiny ${\uu{2}}$};

\draw (3,0) node [below] {\tiny ${\uu{0}}$};
\draw (-.5,.5) node  {$T_{1}=$};
\begin{scope}
[yshift=-2cm]
  \tikzstyle{every node}=[draw,circle,fill=black,minimum size=2pt,
                        inner sep=0pt]                    
\foreach \x in {0,1,2,3}{
\node at (\x,0) {};};
\foreach \x in {1}{
\node at (\x,1) {};};

\tikzstyle{every node}=[]
\draw [->,dashed,black] (0,0) -- (1,0);
\draw [->,dotted,black] (0,0) -- (1,1);
\draw [->,dashed,black] (1,0) -- (2,0);
\draw [->,dotted,black] (1,1) -- (2,0);

\draw [->,dashed,black] (2,0) -- (3,0);

\draw (0,0) node [below] {\tiny ${\uu{0}}$};
\draw (1,0) node [below] {\tiny ${\uu{0}}$};
\draw (1,1) node [below] {\tiny ${\uu{2}}$};
\draw (2,0) node [below] {\tiny ${\uu{0}}$};

\draw (3,0) node [below] {\tiny ${\uu{0}}$};
\draw (-.5,.5) node  {$T_{2}=$};
\end{scope}
      \end{tikzpicture}
    \end{center}
So,  
we can still say that a factorization theorem holds for $T$.  
The only catch is that $T_{1}$ is not elementary according to the previous definition, as it represents $\{000,112\}\subseteq \ZZ_{4}^{3}$, which is not a group code. Still, it makes sense to assign the ``elementary'' adjective to it too, since it is not (isomorphic to) a product of smaller trellises. Moreover we can think of $T_{1}$ as representing a group in some sense, since $\{000,112\}$ corresponds (bijectively) to the quotient group $\lb112\rb/\lb220\rb$ (it must be kept in mind though that this correspondence/information is not given by $T_{1}$ alone, i.e. it is not encoded in it, but it is only provided by the factorization $T\sim T_{1}\otimes T_{2}$).   

After having checked more examples, it seems likely that the above happens in general, i.e. any (abelian) group trellis is a product of trellises which cannot be further factored and which represent sets corresponding to quotient groups of prime order. However this still requires a proof. While the approach given in \cite{KV2} does not seem to be adaptable for yielding such a proof, our framework can be extended to group trellises and we believe that with some modifications it can yield the sought proof (we leave this for future research). 

Note though that in \cite{FoTr} related problems have been treated.

\begin{rmk} Some substantial modifications may be necessary when extending our framework to the group case: for example, the below group trellises over $\ZZ_{4}$ are isomorphic but their group structures are not, contrarily to what the straight translation of  Theorem \ref{isolin} to the group case would tell us (note though that Theorem \ref{thm16} and Corollary \ref{cor13} extend without saying to group trellises).
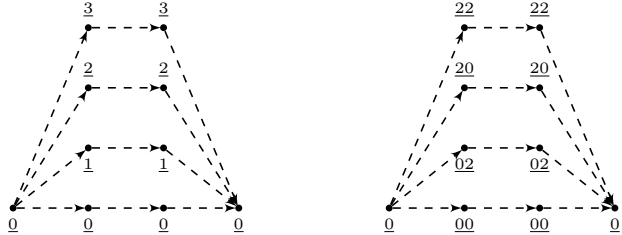
\begin{center}
 \begin{tikzpicture}[yscale=.8,>=latex',shorten >=.9pt, shorten <=1.4pt, line width=.6pt]
    \tikzstyle{every node}=[draw,circle,fill=black,minimum size=2pt,
                        inner sep=0pt]                    
\foreach \x in {0,1,2,3}{
\node at (\x,0) {};};
\foreach \x in {1,2}{
\node at (\x,1) {};};
\foreach \x in {1,2}{
\node at (\x,2) {};};
\foreach \x in {1,2}{
\node at (\x,3) {};};

\tikzstyle{every node}=[]
\draw [->,dashed,black] (0,0) -- (1,0);
\draw [->,dashed,black] (1,0) -- (2,0);
\draw [->,dashed,black] (2,0) -- (3,0); 

\draw [->,dashed,black] (0,0) -- (1,1);
\draw [->,dashed,black] (1,1) -- (2,1);
\draw [->,dashed,black] (2,1) -- (3,0); 

\draw [->,dashed,black] (0,0) -- (1,2);
\draw [->,dashed,black] (1,2) -- (2,2);
\draw [->,dashed,black] (2,2) -- (3,0); 

\draw [->,dashed,black] (0,0) -- (1,3);
\draw [->,dashed,black] (1,3) -- (2,3); 
\draw [->,dashed,black] (2,3) -- (3,0); 

\draw (0,0) node [below] {\tiny ${\uu{0}}$};
\draw (1,0) node [below] {\tiny ${\uu{0}}$};
\draw (1,1) node [below] {\tiny ${\uu{1}}$};
\draw (1,2) node [above] {\tiny ${\uu{2}}$};
\draw (1,3) node [above] {\tiny ${\uu{3}}$};

\draw (2,0) node [below] {\tiny ${\uu{0}}$};
\draw (2,1) node [below] {\tiny ${\uu{1}}$};
\draw (2,2) node [above] {\tiny ${\uu{2}}$};
\draw (2,3) node [above] {\tiny ${\uu{3}}$};

\draw (3,0) node [below] {\tiny ${\uu{0}}$};
%\draw (-.5,1.5) node  {$T=$};

%  \tikzstyle{every node}=[draw,circle,fill=black,minimum size=2pt,
%                        inner sep=0pt]                    
%\foreach \x in {0,1,2,3}{
%\node at (\x,0) {};};
%\foreach \x in {1,2}{
%\node at (\x,1) {};};
%\foreach \x in {1}{
%\node at (\x,2) {};};
%\foreach \x in {1}{
%\node at (\x,3) {};};
%
%\tikzstyle{every node}=[]
%\draw [->] (0,0) -- (1,0);
%\draw [->] (0,0) -- (1,1);
%\draw [->] (0,0) -- (1,2);
%\draw [->] (0,0) -- (1,3);
%
%\draw [->] (1,0) -- (2,0);
%\draw [->] (1,1) -- (2,1);
%
%\draw [->] (1,2) -- (2,0);
%\draw [->] (1,3) -- (2,1); 
%
%\draw [->] (2,1) -- (3,0); 
%\draw [->] (2,0) -- (3,0); 
%
%
%\draw (0,0) node [below] {\tiny ${\uu{0}}$};
%\draw (1,0) node [below] {\tiny ${\uu{0}}$};
%\draw (1,1) node [below] {\tiny ${\uu{1}}$};
%\draw (1,2) node [above] {\tiny ${\uu{2}}$};
%\draw (1,3) node [above] {\tiny ${\uu{3}}$};
%
%\draw (2,0) node [below] {\tiny ${\uu{0}}$};
%\draw (2,1) node [below] {\tiny ${\uu{2}}$};
%
%
%\draw (3,0) node [below] {\tiny ${\uu{0}}$};
%\draw (-.5,1.5) node  {$T=$};
\begin{scope}
[xshift=5cm]
\tikzstyle{every node}=[draw,circle,fill=black,minimum size=2pt,
                        inner sep=0pt]                    
\foreach \x in {0,1,2,3}{
\node at (\x,0) {};};
\foreach \x in {1,2}{
\node at (\x,1) {};};
\foreach \x in {1,2}{
\node at (\x,2) {};};
\foreach \x in {1,2}{
\node at (\x,3) {};};

\tikzstyle{every node}=[]
\draw [->,dashed,black] (0,0) -- (1,0);
\draw [->,dashed,black] (1,0) -- (2,0);
\draw [->,dashed,black] (2,0) -- (3,0); 

\draw [->,dashed,black] (0,0) -- (1,1);
\draw [->,dashed,black] (1,1) -- (2,1);
\draw [->,dashed,black] (2,1) -- (3,0); 

\draw [->,dashed,black] (0,0) -- (1,2);
\draw [->,dashed,black] (1,2) -- (2,2);
\draw [->,dashed,black] (2,2) -- (3,0); 

\draw [->,dashed,black] (0,0) -- (1,3);
\draw [->,dashed,black] (1,3) -- (2,3); 
\draw [->,dashed,black] (2,3) -- (3,0); 

\draw (0,0) node [below] {\tiny ${\uu{0}}$};
\draw (1,0) node [below] {\tiny ${\uu{00}}$};
\draw (1,1) node [below] {\tiny ${\uu{02}}$};
\draw (1,2) node [above] {\tiny ${\uu{20}}$};
\draw (1,3) node [above] {\tiny ${\uu{22}}$};

\draw (2,0) node [below] {\tiny ${\uu{00}}$};
\draw (2,1) node [below] {\tiny ${\uu{02}}$};
\draw (2,2) node [above] {\tiny ${\uu{20}}$};
\draw (2,3) node [above] {\tiny ${\uu{22}}$};

\draw (3,0) node [below] {\tiny ${\uu{0}}$};
%\draw (-.5,1.5) node  {$T'=$};
%
%  \tikzstyle{every node}=[draw,circle,fill=black,minimum size=2pt,
%                        inner sep=0pt]                    
%\foreach \x in {0,1,2,3}{
%\node at (\x,0) {};};
%\foreach \x in {1,2}{
%\node at (\x,1) {};};
%\foreach \x in {1}{
%\node at (\x,2) {};};
%\foreach \x in {1}{
%\node at (\x,3) {};};
%
%\tikzstyle{every node}=[]
%\draw [->] (0,0) -- (1,0);
%\draw [->] (0,0) -- (1,1);
%\draw [->] (0,0) -- (1,2);
%\draw [->] (0,0) -- (1,3);
%
%\draw [->] (1,0) -- (2,0);
%\draw [->] (1,1) -- (2,1);
%
%\draw [->] (1,2) -- (2,0);
%\draw [->] (1,3) -- (2,1); 
%
%\draw [->] (2,1) -- (3,0); 
%\draw [->] (2,0) -- (3,0); 
%
%
%\draw (0,0) node [below] {\tiny ${\uu{0}}$};
%\draw (1,0) node [below] {\tiny ${\uu{00}}$};
%\draw (1,1) node [below] {\tiny ${\uu{02}}$};
%\draw (1,2) node [above] {\tiny ${\uu{20}}$};
%\draw (1,3) node [above] {\tiny ${\uu{22}}$};
%
%\draw (2,0) node [below] {\tiny ${\uu{0}}$};
%\draw (2,1) node [below] {\tiny ${\uu{2}}$};
%
%
%\draw (3,0) node [below] {\tiny ${\uu{0}}$};
%\draw (-.5,1.5) node  {$T'=$};
\end{scope}
      \end{tikzpicture}
    \end{center}
\end{rmk}

The Factorization Theorem can make proving certain results for linear trellises easier, however, because of the above situation, when proving a result for such trellises it is preferable to avoid that theorem if possible 
and give instead proofs that exploit only the additive structure and can be thus immediately extended to the (abelian) group case too.

%%%%%%%%%%%%%%%%%%%%%%%%%%%%%%%%%%%%%%%%%%%%%%%%%%%%

\section{Applications: factoring linear trellises}\label{factpap}
%%%%%%%%%%%%%%%%%%%%%%%%%%%%%%%%%%%%%%%%%%%%%%%%%%%%%%

\subsection{Uniqueness of span distribution of linear trellises}

Knowing from the Factorization Theorem that every linear trellis factors into elementary trellises, a natural following question is whether such a factorization is unique. In general the answer is negative. 
For example,  
$111|(0,2)\otimes010|(1,0)$ and $101|(0,2)\otimes010|(1,0)$ yield the same trellis, but the two factorizations are different since $111|(0,2)\neq101|(0,2)$. One sees though that the list of spans of the two factorizations in this case are the same. This is no coincidence, since here we are dealing with a minimal conventional trellis and it is well known that all the elementary trellis factorizations of such a trellis  give rise to the same spans, more precisely, the \textit{atomic spans} of the represented code (see  Subsection \ref{atomspanpap}). 
However, this uniqueness of  span distribution had been proven only as a byproduct of the minimality assumption. 
More recently in \cite{GW} (Proposition III.14 therein), with arguments based again on atomic spans, this result  was extended to the class of so-called \textit{$\mathrm{KV}$-trellises}  (see Subsection \ref{atomspanpap}), which is a subclass of the class of nonmergeable, one-to-one, linear trellises containing the class of minimal linear trellises. 

We show here that this actually holds true for any linear trellis, and furthermore, that the edge-labels really play no role in determining the spans. The underlying graph structure alone determines them. 
We will prove these two claims as easy consequences of the algebraic framework developed in the previous section.

First, let us make the terminology precise: 
 by the \textit{span distribution} of the elementary trellis factorization $\otimes_{i=1}^{r} \al^{i}|(a_{i},l_{i})$ of  a linear trellis 
 we mean the multiset $$\{\{(a_{i},l_{i})|i=1,\ldots,r\}\}$$  
Now, given a product basis $\BB$ of a linear trellis $T$, recall that by $4)$ of Observation \ref{spansub1} the number of cycles $\la\in\BB$ with span $[\la]=(a,l)$ is given precisely by  $\dim\sal(T)/\SSS_{<(a,l)}(T)$, and thus it depends only on $T$,  not on the particular choice of the product basis. In particular, all product bases yield the same span distribution.
But then by Theorem \ref{isolin} and the correspondence between product bases and elementary trellis factorizations given by Theorem \ref{thm13}   we conclude that the following holds:

\begin{thm}\label{trespan}
Two elementary trellis factorizations of a linear trellis $T$ have the same span distribution. More precisely, the number of times a span $(a,l)$ appears in an elementary trellis factorization of $T$ is equal to $$\dim\sal(T)/\SSS_{<(a,l)}(T)=\dim\sal(T)-\dim\SSS_{<(a,l)}(T)$$
\end{thm}

This shows that our first claim holds, and also justifies talking about \textit{the span distribution of a linear trellis $T$}\index{span!distribution (of $T$)}, which we will henceforth denote by  $\mathcal{S}(T)$\nomenclature[sst]{$\mathcal{S}(T)$}{span distribution of $T$} (recall that this is a multiset). We also put  
\begin{align*}
\mathcal{S}_{+}(T)&:=\{\{(a,l)\in\mathcal{S}(T)| l>0\}\}\\
\mathcal{S}_{0}(T)&:=\mathcal{S}(T)\setminus\mathcal{S}_{+}(T)
\end{align*}
 (note that $\mathcal{S}_{0}(T)$ is a set, because of our assumption in Remark \ref{assumpt1pap}).
We now prove our second claim. 

\begin{thm}\label{unfact1}
Two linear trellises are structurally isomorphic if and only if they have the same span distribution. 
\end{thm}
\begin{proof}
Given a trellis $T$ we denote by $\overline{T}$ its underlying unlabeled trellis, i.e. $\overline{T}$ is defined by putting all the edge-labels of $T$ equal to $0$, so that $V_{i}(\overline{T})=V_{i}(T)$ and $E_{i}(\overline{T})$ is the image of $E_{i}(T)$ under the map $v\alpha w\mapsto v0w$, for all $i$.
In particular 
\begin{equation*}
\overline{\al|(a,l)}=
\begin{cases}
\bm{0}|(a,l)& \textnormal{ if } l>0\\
0 & \textnormal{ if } l=0
\end{cases}
\end{equation*}
Also, clearly $\overline{T}$ commutes with trellis products, i.e. 
$$\overline{T_{1}\otimes T_{2}}=\overline{T_{1}}\otimes\overline{T_{2}}$$
 
So it follows immediately that $\mathcal{S}(\overline{T})=\mathcal{S}_{+}(T)$. On the other hand,  two linear trellises $T_{1}$, $T_{2}$ are structurally isomorphic if and only if  $\mathcal{S}_{0}(T_{1})=\mathcal{S}_{0}(T_{2})$  and $\overline{T_{1}}\sim\overline{T_{2}}$. We can thus conclude by the above theorem that  structurally isomorphic linear trellises must have the same span distribution. 
The  ``if'' part is trivial, and we have included it for the sake of completeness. 
\end{proof}

 Thus the span distribution of a linear trellis does not depend on its edge-labels,  but only on its underlying graph. Note that by the arguments in the above proof it also follows that the multiplicity of a span $(a,l)\in\mathcal{S}(T)$ with positive length $l>0$ is equal to $\dim\sal(\overline{T})/\SSS_{<(a,l)}(\overline{T})$ (where $\overline{T}$ is defined in the same proof).

\begin{rmk}\label{rmk31}
The mathematically keen reader may have recognized at this point that the above theorem along with the Factorization Theorem imply that the class of unlabeled linear trellises (up to isomorphism) equipped with the trellis product operation is a unique factorization monoid.
\end{rmk}

We conclude this subsection by showing that we can compute all the span multiplicities if we know the dimension of each span subcode $\sal(T)$. The following identity makes that possible. 

\begin{obs}\label{obs30}
Let $T$ be a linear trellis. Then for each $(a,l)$ we have 
\begin{gather}\label{eq3}
\dim\sal(T)/\SSS_{<(a,l)}(T)=\\
\dim\sal(T)-\sum_{(a',l')<(a,l)}\dim\SSS_{(a',l')}(T)/\SSS_{<(a',l')}(T)\notag
\end{gather}
\end{obs}
\begin{proof}
Let $\BB$ be a product basis of $T$. Then the equality 
 follows from putting together $2)$, $3)$ and $4)$ of Observation \ref{spansub1}.
\end{proof}

Equality \eqref{eq3} can be used recursively for increasing span lengths to compute
all %NOT PRESENT
 the multiplicities $\dim\sal({T})/\SSS_{<(a,l)}({T})$ from the span subcode dimensions $\dim\sal(T)$. Indeed, starting with length 
 equal to  %NOT PRESENT
 $-1$, i.e. $\SSS_{\emptyset}(T)=0$, we get 
 
 %NOT PRESENT (ABOVE EMPTY LINE)
\begin{align*}
\dim&\SSS_{(a,-1)}(T)\equiv\dim\SSS_{\emptyset}(T)=0\\
\dim&\SSS_{(a,0)}(T)/\SSS_{<(a,0)}(T)=\dim\SSS_{(a,0)}(T)\\
\dim&\SSS_{(a,1)}(T)/\SSS_{<(a,1)}(T)=\dim\SSS_{(a,1)}(T)-\sum_{(a',0)<(a,1)}\dim\SSS_{(a',0)}(T)\\
\dim&\SSS_{(a,2)}(T)/\SSS_{<(a,2)}(T)=\dim\SSS_{(a,2)}(T)-\sum_{(a',1)<(a,2)}\dim\SSS_{(a',1)}/\SSS_{(a',1)}(T)=\\
&=\dim\SSS_{(a,2)}(T)-\sum_{(a',1)<(a,2)}\Bigl(\dim\SSS_{(a',1)}(T)-\sum_{(a'',0)<(a',1)}\dim\SSS_{(a'',0)}(T)\Bigr)
\end{align*}
and so on.

This can be practically worked out by replacing each span $(a,l)$ by $\dim\sal(T)$ in the Hasse diagram 
for spans,
and then  processing the entries in the diagram from bottom to top according to the above equations, where a single step amounts to subtracting from an entry in the diagram the sum of all the other entries below it in the diagram and update the entry. The final entries will then give the multiplicity of each span. 
For example, if we take the trellis $T$ from Example \ref{ex1pap}  we get the following sequence, where we underline the entries which have been processed so far at each stage:
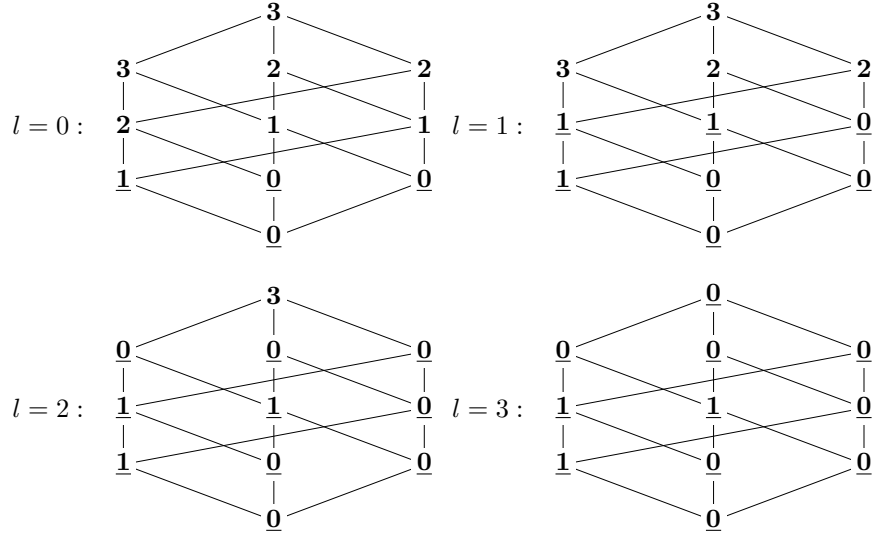
\begin{center}
 \begin{tikzpicture}[yscale=.5,xscale=.5, line width=.3pt]
                                      \tikzstyle{every node}=[]
                                      
             \foreach \a in {0,1,2}{
                        \foreach \l in {0,1}{
                       \draw  (4*\a,1.5*\l) -- (4*\a,1.5*\l+1.5);};};
                       
                        \foreach \a in {1,2}{
                        \foreach \l in {0,1}{
                       \draw  (4*\a,1.5*\l) -- (4*\a-4,1.5*\l+1.5);};};
                       
                        \foreach \l in {0,1}{
                       \draw  (0,1.5*\l) -- (8,1.5*\l+1.5);};
                       
                                \foreach \a in {0,1,2}{
                       \draw  (4*\a,0) -- (4,-1.5);
                       \draw  (4*\a,3) -- (4,4.5);};
            
             \tikzstyle{every node}=[fill=white,minimum size=0pt,
                        inner sep=2pt]\small
                        
\draw (0,0) node {$\uu{\bm1}$};
\draw (0,1.5) node {$\bm2$};
\draw (0,3) node {$\bm3$};

\draw (4,0) node {$\uu{\bm0}$};
\draw (4,1.5) node {$\bm1$};
\draw (4,3) node {$\bm2$};

\draw (8,0) node {$\uu{\bm0}$};
\draw (8,1.5) node {$\bm1$};
\draw (8,3) node {$\bm2$};

    \draw (4,-1.5) node {$\uu{\bm0}$};
     \draw (4,4.5) node {$\bm3$};
     \draw (-2,1.5) node {$l=0:$};
    
    %  \end{tikzpicture}
    
%\vspace{.5cm}      
%\begin{tikzpicture}[yscale=.5,xscale=.5, line width=.3pt]
      \end{tikzpicture}
            \vspace{.3cm}
 \begin{tikzpicture}[yscale=.5,xscale=.5, line width=.3pt]
                                      \tikzstyle{every node}=[]
                                      
             \foreach \a in {0,1,2}{
                        \foreach \l in {0,1}{
                       \draw  (4*\a,1.5*\l) -- (4*\a,1.5*\l+1.5);};};
                       
                        \foreach \a in {1,2}{
                        \foreach \l in {0,1}{
                       \draw  (4*\a,1.5*\l) -- (4*\a-4,1.5*\l+1.5);};};
                       
                        \foreach \l in {0,1}{
                       \draw  (0,1.5*\l) -- (8,1.5*\l+1.5);};
                       
                                \foreach \a in {0,1,2}{
                       \draw  (4*\a,0) -- (4,-1.5);
                       \draw  (4*\a,3) -- (4,4.5);};
            
             \tikzstyle{every node}=[fill=white,minimum size=0pt,
                        inner sep=2pt]\small
                        
\draw (0,0) node {$\uu{\bm1}$};
\draw (0,1.5) node {$\uu{\bm1}$};
\draw (0,3) node {$\bm3$};

\draw (4,0) node {$\uu{\bm0}$};
\draw (4,1.5) node {$\uu{\bm1}$};
\draw (4,3) node {$\bm2$};

\draw (8,0) node {$\uu{\bm0}$};
\draw (8,1.5) node {$\uu{\bm0}$};
\draw (8,3) node {$\bm2$};

    \draw (4,-1.5) node {$\uu{\bm0}$};
     \draw (4,4.5) node {$\bm3$};
       \draw (-2,1.5) node {$l=1:$};
      \end{tikzpicture}
            \vspace{.3cm}
\begin{tikzpicture}[yscale=.5,xscale=.5, line width=.3pt]
                                      \tikzstyle{every node}=[]
                                      
             \foreach \a in {0,1,2}{
                        \foreach \l in {0,1}{
                       \draw  (4*\a,1.5*\l) -- (4*\a,1.5*\l+1.5);};};
                       
                        \foreach \a in {1,2}{
                        \foreach \l in {0,1}{
                       \draw  (4*\a,1.5*\l) -- (4*\a-4,1.5*\l+1.5);};};
                       
                        \foreach \l in {0,1}{
                       \draw  (0,1.5*\l) -- (8,1.5*\l+1.5);};
                       
                                \foreach \a in {0,1,2}{
                       \draw  (4*\a,0) -- (4,-1.5);
                       \draw  (4*\a,3) -- (4,4.5);};
            
             \tikzstyle{every node}=[fill=white,minimum size=0pt,
                        inner sep=2pt]\small
                        
\draw (0,0) node {$\uu{\bm1}$};
\draw (0,1.5) node {$\uu{\bm1}$};
\draw (0,3) node {$\uu{\bm0}$};

\draw (4,0) node {$\uu{\bm0}$};
\draw (4,1.5) node {$\uu{\bm1}$};
\draw (4,3) node {$\uu{\bm0}$};

\draw (8,0) node {$\uu{\bm0}$};
\draw (8,1.5) node {$\uu{\bm0}$};
\draw (8,3) node {$\uu{\bm0}$};

    \draw (4,-1.5) node {$\uu{\bm0}$};
     \draw (4,4.5) node {$\bm3$};
       \draw (-2,1.5) node {$l=2:$};
%
%      \end{tikzpicture}
%      
%      \vspace{.5cm}      
%\begin{tikzpicture}[yscale=.5,xscale=.5, line width=.3pt]
      \end{tikzpicture}
      \vspace{.3cm}
\begin{tikzpicture}[yscale=.5,xscale=.5, line width=.3pt]
                                      \tikzstyle{every node}=[]
                                      
             \foreach \a in {0,1,2}{
                        \foreach \l in {0,1}{
                       \draw  (4*\a,1.5*\l) -- (4*\a,1.5*\l+1.5);};};
                       
                        \foreach \a in {1,2}{
                        \foreach \l in {0,1}{
                       \draw  (4*\a,1.5*\l) -- (4*\a-4,1.5*\l+1.5);};};
                       
                        \foreach \l in {0,1}{
                       \draw  (0,1.5*\l) -- (8,1.5*\l+1.5);};
                       
                                \foreach \a in {0,1,2}{
                       \draw  (4*\a,0) -- (4,-1.5);
                       \draw  (4*\a,3) -- (4,4.5);};
            
             \tikzstyle{every node}=[fill=white,minimum size=0pt,
                        inner sep=2pt]\small
                        
\draw (0,0) node {$\uu{\bm1}$};
\draw (0,1.5) node {$\uu{\bm1}$};
\draw (0,3) node {$\uu{\bm0}$};

\draw (4,0) node {$\uu{\bm0}$};
\draw (4,1.5) node {$\uu{\bm1}$};
\draw (4,3) node {$\uu{\bm0}$};

\draw (8,0) node {$\uu{\bm0}$};
\draw (8,1.5) node {$\uu{\bm0}$};
\draw (8,3) node {$\uu{\bm0}$};

    \draw (4,-1.5) node {$\uu{\bm0}$};
     \draw (4,4.5) node {$\uu{\bm0}$};
       \draw (-2,1.5) node {$l=3:$};
      \end{tikzpicture}
\end{center}
The final diagram tells us that the span distribution of $T$ is $\{\{(0,0),(0,1),(1,1)\}\}$, which is indeed the case.

\begin{rmk}
We had announced Theorem \ref{unfact1} first in \cite{CB2}. Therein indeed we had illustrated how the span distribution $\mathcal{S}(T)$ of a linear trellis $T$ can be recovered from simple graphical characteristics of $T$. In Appendix \ref{graphcarpap} we provide the complete details of that worthy alternative perspective. 
\end{rmk}

%%%%%%%%%%%%%%%%%%%%%%%%%%%%%%%%%%%%%%%%%%%%%%%%%%%%%%
%%%%%%%%%%%%%%%%%%%%%%%%%%%%%%%%%%%%%%%%%%%%%%%%%%%%%%
%%%%%%%%%%%%%%%%%%%%%%%%%%%%%%%%%%%%%%%%%%%%%%%%%%%%%%
%%%%%%%%%%%%%%%%%%%%%%%%%%%%%%%%%%%%%%%%%%%%%%%%%%%%%%
%%%%%%%%%%%%%%%%%%%%%%%%%%%%%%%%%%%%%%%%%%%%%%%%%%%%%%
%%%%%%%%%%%%%%%%%%%%%%%%%%%%%%%%%%%%%%%%%%%%%%%%%%%%%%
%%%%%%%%%%%%%%%%%%%%%%%%%%%%%%%%%%%%%%%%%%%%%%%%%%%%%%
%%%%%%%%%%%%%%%%%%%%%%%%%%%%%%%%%%%%%%%%%%%%%%%%%%%%%%
%%%%%%%%%%%%%%%%%%%%%%%%%%%%%%%%%%%%%%%%%%%%%%%%%%%%%%

\subsection{Edge-labels of elementary trellis factorizations}

After showing that a linear trellis determines uniquely the span distribution of its elementary trellis factorizations and how this can be determined, the naturally following problem is to find out what are the possible edge-labels of the elementary factors. We have already pointed out before that unique factorization does not hold, which means that edge-labels are not necessarily unique. Nevertheless we give here a method for finding all the legitimate edge-labelings. We will also deduce necessary and sufficient conditions for a linear trellis $T$ to have a unique elementary trellis factorization. Again, our results will be easy consequences of what proven in Section \ref{algframpap}.

\begin{obs}\label{obs31}
Let $T$, $T'$ be linear trellises such that $\mathcal{S}(T)=\mathcal{S}(T')$. Then $T$ and $T'$ are isomorphic if and only if $$C_{\sss}(T)=C_{\sss}(T')$$
 for all spans $\sss$.
\end{obs}
\begin{proof}
By our Isomorphy Theorem \ref{thm15} we need to prove that $\dim\SSS_{\sss}(T)=\dim\SSS_{\sss}(T')$ for all spans $\sss$. As the two span distributions are equal, we have that $\dim\SSS_{\sss}(T)/\SSS_{<\sss}(T)=\dim\SSS_{\sss}(T')/\SSS_{<\sss}(T')$ for all $\sss$. The sought equalities follow then from  equation \ref{eq3}. 
\end{proof}

Now, assume that we know the span distribution $\mathcal{S}(T)$ of a linear trellis $T$. For a generic span 
$\sss$ 
let  $$m(\sss,T)$$ 
denote the multiplicity of $\sss$ in $\mathcal{S}(T)$. In particular, $m(\sss,T)>0$ if $\sss\in\mathcal{S}(T)$, and $m(\sss,T)=0$ otherwise.  Any elementary trellis factorization of $T$ can be  then written as  $$\otimes_{\sss\in\mathcal{S}(T)}\otimes_{i=1}^{m(\sss,T)}\al^{\sss,i}|\sss$$ for some $\al^{\sss,i}\in\FF^{n}$. The question then is: what are the possible $\al^{\sss,i}$?  Here is the complete answer:

\begin{thm}\label{trelfac}
Let $T$ be a linear trellis.  Then 
\begin{equation}\label{eqpapfac}
\otimes_{\sss\in\mathcal{S}(T)}\otimes_{i=1}^{m(\sss,T)}\al^{\sss,i}|\sss
\end{equation}  
is an elementary trellis factorization of $T$ 
if and only if for each span $\sss\in\mathcal{S}(T)$ we have 
\begin{equation}\label{eq4}
C_{\sss}(T)=\langle\al^{\sss,1}, \ldots,\al^{\sss,m(\sss,T)}\rangle+\sum_{\sss'<\sss}C_{\sss'}(T)
\end{equation}
\end{thm}
\begin{proof}
We know by Observation \ref{obs11} 
that 
\begin{gather*}
C_{\sss}(\otimes_{\sss\in\mathcal{S}(T)}\otimes_{i=1}^{m(\sss,T)}\al^{\sss,i}|\sss)=\\
=\langle\al^{\sss',i}|\sss'\in\mathcal{S}(T), \sss'\leq\sss, i=1,\ldots,m(\sss',T)\rangle
\end{gather*}
for all $\sss$, and in particular for all $\sss\in\mathcal{S}(T)$. 
Thus, if  \eqref{eqpapfac} is an elementary trellis factorization of $T$ equation \eqref{eq4} follows. 
Vice versa, assume \eqref{eq4} holds for all $\sss\in\mathcal{S}(T)$. First, note that if $\sss\notin\mathcal{S}(T)$ then $C_{\sss}(T)=\sum_{\sss'<\sss}C_{\sss'}(T)$.  Now, by iterative substitutions in these equations and in \eqref{eq4} for decreasing span lengths we deduce that 
\begin{gather*}
C_{\sss}(T)=\\
=\langle\al^{\sss',i}|\sss'\in\mathcal{S}(T),\sss'\leq\sss, i=1,\ldots,m(\sss',T)\rangle=\\
=C_{\sss}(T')
\end{gather*}
  for all 
$\sss$. 
The conclusion follows now from Observation \ref{obs31}.
\end{proof}
The above theorem can be used to compute all the possible labelings of the elementary factors of $T$, and so, preceded by the computation of $\mathcal{S}(T)$ via Theorem \ref{trespan} (or via the graphical intersection data as described in Appendix \ref{graphcarpap}), all the possible elementary factors of $T$.

 Note that for a given span $\sss\in\mathcal{S}(T)$ only the labels for elementary factors with that same span appear in equation \eqref{eq4}. Therefore, when computing all the possible labelings of the factors of $T$ this can be worked out independently for each span $\sss\in\mathcal{S}(T)$.

\begin{ex}\label{ex31} Let $T$ be 
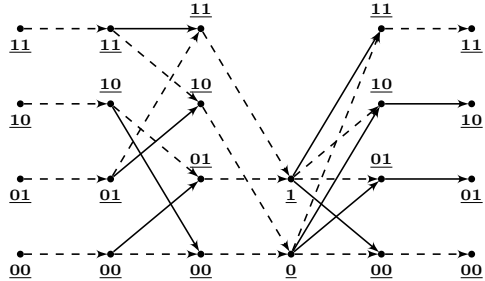
\begin{center}
 \begin{tikzpicture}[xscale=1.2,>=latex',shorten >=.9pt, shorten <=1.4pt, line width=.6pt]
  \tikzstyle{every node}=[draw,circle,fill=black,minimum size=2pt,
                        inner sep=0pt]                    
\foreach \x in {0,1,2,4,5}{
\node at (\x,0) {};};
\foreach \x in {0,1,2,4,5}{
\node at (\x,1) {};};
\foreach \x in {0,1,2,4,5}{
\node at (\x,2) {};};
\foreach \x in {0,1,2,4,5}{
\node at (\x,3) {};};
\node at (3,0) {};
\node at (3,1) {};
\tikzstyle{every node}=[]
\draw (0,0) node [below] {\tiny $\mathbf{\underline{00}}$};
\draw (0,1) node [below] {\tiny $\mathbf{\underline{01}}$};
\draw (0,2) node [below] {\tiny $\mathbf{\underline{10}}$};
\draw (0,3) node [below] {\tiny $\mathbf{\underline{11}}$};

\draw (1,0) node [below] {\tiny $\mathbf{\underline{00}}$};
\draw (1,1) node [below] {\tiny $\mathbf{\underline{01}}$};
\draw (1,2) node [above] {\tiny $\mathbf{\underline{10}}$};
\draw (1,3) node [below] {\tiny $\mathbf{\underline{11}}$};

\draw (2,0) node [below] {\tiny $\mathbf{\underline{00}}$};
\draw (2,1) node [above] {\tiny $\mathbf{\underline{01}}$};
\draw (2,2) node [above] {\tiny $\mathbf{\underline{10}}$};
\draw (2,3) node [above] {\tiny $\mathbf{\underline{11}}$};

\draw (3,0) node [below] {\tiny $\mathbf{\underline{0}}$};
\draw (3,1) node [below] {\tiny $\mathbf{\underline{1}}$};

\draw (4,0) node [below] {\tiny $\mathbf{\underline{00}}$};
\draw (4,1) node [above] {\tiny $\mathbf{\underline{01}}$};
\draw (4,2) node [above] {\tiny $\mathbf{\underline{10}}$};
\draw (4,3) node [above] {\tiny $\mathbf{\underline{11}}$};

\draw (5,0) node [below] {\tiny $\mathbf{\underline{00}}$};
\draw (5,1) node [below] {\tiny $\mathbf{\underline{01}}$};
\draw (5,2) node [below] {\tiny $\mathbf{\underline{10}}$};
\draw (5,3) node [below] {\tiny $\mathbf{\underline{11}}$};

\draw [->,dashed,black] (0,0) --  node[below] {}  (1,0);
\draw [->,dashed,black] (0,1) --   (1,1);
\draw [->,dashed,black] (0,2) --   (1,2);
\draw [->,dashed,black] (0,3) --   (1,3);

\draw [->,dashed,black] (1,0) -- node[below] {}   (2,0);
\draw [->,black] (1,0) --   (2,1);

\draw [->,black] (1,1) --   (2,2);
\draw [->,dashed,black] (1,1) --   (2,3);

\draw [->,black] (1,2) --   (2,0);
\draw [->,dashed,black] (1,2) --   (2,1);

\draw [->,dashed,black] (1,3) --   (2,2);
\draw [->,black] (1,3) --   (2,3);

\draw [->,dashed,black] (2,0) --  node[below] {}  (3,0);
\draw [->,dashed,black] (2,1) --   (3,1);

\draw [->, dashed,black] (2,2) --   (3,0);
\draw [->, dashed,black] (2,3) --   (3,1);

\draw [->,dashed,black] (3,0) --  node[below] {}  (4,0);
\draw [->,black] (3,0) --   (4,1);
\draw [->,black] (3,0) --   (4,2);
\draw [->,dashed,black] (3,0) --   (4,3);

\draw [->,black] (3,1) --   (4,0);
\draw [->,dashed,black] (3,1) --   (4,1);
\draw [->,dashed,black] (3,1) --   (4,2);
\draw [->,black] (3,1) --   (4,3);

\draw [->,dashed,black] (4,0) --  node[below] {}  (5,0);
\draw [->,black] (4,1) --   (5,1);
\draw [->,black] (4,2) --   (5,2);
\draw [->,dashed,black] (4,3) --   (5,3);

    \end{tikzpicture}
    \end{center}
One can check that $\mathcal{S}(T)=\{\{(1,2),(3,3),(3,4)\}\}$ (see also Example \ref{ex33}). Then one gets:
\begin{itemize}
\item $C_{(1,2)}(T)=\langle01010\rangle$
\item $C_{(3,3)}(T)=\langle01011\rangle$
\item $C_{(3,4)}(T)=\langle01011\rangle$
\end{itemize}
Therefore by the above theorem there are two and only two distinct factorizations of $T$, namely
\begin{gather*}
01010|(1,2)\otimes01011|(3,3)\otimes01011|(3,4)\\
01010|(1,2)\otimes01011|(3,3)\otimes00000|(3,4)
\end{gather*}
According to the same theorem 
$$01010|(1,2)\otimes00000|(3,3)\otimes01011|(3,4)$$ is not an elementary trellis factorization of $T$  (because $C_{(3,3)}(T)\neq 0$), i.e. it yields a different (i.e. nonisomorphic) trellis $T'$, which we have depicted   below. From the diagram one can indeed check that $T'\neq T$ as it is possible to go out from $\underline{0}\in V_{3}(T')$ along two paths with all edge-labels equal to $0$ that meet again at $V_{2}(T')$, while this is not possible in $T$.
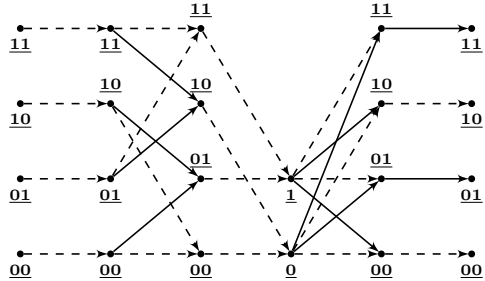
\begin{center}
 \begin{tikzpicture}[xscale=1.2,>=latex',shorten >=.9pt, shorten <=1.4pt, line width=.6pt]
  \tikzstyle{every node}=[draw,circle,fill=black,minimum size=2pt,
                        inner sep=0pt]                    
\foreach \x in {0,1,2,4,5}{
\node at (\x,0) {};};
\foreach \x in {0,1,2,4,5}{
\node at (\x,1) {};};
\foreach \x in {0,1,2,4,5}{
\node at (\x,2) {};};
\foreach \x in {0,1,2,4,5}{
\node at (\x,3) {};};
\node at (3,0) {};
\node at (3,1) {};
\tikzstyle{every node}=[]
\draw (0,0) node [below] {\tiny $\mathbf{\underline{00}}$};
\draw (0,1) node [below] {\tiny $\mathbf{\underline{01}}$};
\draw (0,2) node [below] {\tiny $\mathbf{\underline{10}}$};
\draw (0,3) node [below] {\tiny $\mathbf{\underline{11}}$};

\draw (1,0) node [below] {\tiny $\mathbf{\underline{00}}$};
\draw (1,1) node [below] {\tiny $\mathbf{\underline{01}}$};
\draw (1,2) node [above] {\tiny $\mathbf{\underline{10}}$};
\draw (1,3) node [below] {\tiny $\mathbf{\underline{11}}$};

\draw (2,0) node [below] {\tiny $\mathbf{\underline{00}}$};
\draw (2,1) node [above] {\tiny $\mathbf{\underline{01}}$};
\draw (2,2) node [above] {\tiny $\mathbf{\underline{10}}$};
\draw (2,3) node [above] {\tiny $\mathbf{\underline{11}}$};

\draw (3,0) node [below] {\tiny $\mathbf{\underline{0}}$};
\draw (3,1) node [below] {\tiny $\mathbf{\underline{1}}$};

\draw (4,0) node [below] {\tiny $\mathbf{\underline{00}}$};
\draw (4,1) node [above] {\tiny $\mathbf{\underline{01}}$};
\draw (4,2) node [above] {\tiny $\mathbf{\underline{10}}$};
\draw (4,3) node [above] {\tiny $\mathbf{\underline{11}}$};

\draw (5,0) node [below] {\tiny $\mathbf{\underline{00}}$};
\draw (5,1) node [below] {\tiny $\mathbf{\underline{01}}$};
\draw (5,2) node [below] {\tiny $\mathbf{\underline{10}}$};
\draw (5,3) node [below] {\tiny $\mathbf{\underline{11}}$};

\draw [->,dashed,black] (0,0) --  node[below] {}  (1,0);
\draw [->,dashed,black] (0,1) --   (1,1);
\draw [->,dashed,black] (0,2) --   (1,2);
\draw [->,dashed,black] (0,3) --   (1,3);

\draw [->,dashed,black] (1,0) -- node[below] {}   (2,0);
\draw [->,black] (1,0) --   (2,1);

\draw [->,black] (1,1) --   (2,2);
\draw [->,dashed,black] (1,1) --   (2,3);

\draw [->, dashed,black] (1,2) --   (2,0);
\draw [->,black] (1,2) --   (2,1);

\draw [->,black] (1,3) --   (2,2);
\draw [->,dashed,black] (1,3) --   (2,3);

\draw [->,dashed,black] (2,0) --  node[below] {}  (3,0);
\draw [->,dashed,black] (2,1) --   (3,1);

\draw [->, dashed,black] (2,2) --   (3,0);
\draw [->, dashed,black] (2,3) --   (3,1);

\draw [->,dashed,black] (3,0) --  node[below] {}  (4,0);
\draw [->,black] (3,0) --   (4,1);
\draw [->, dashed,black] (3,0) --   (4,2);
\draw [->,black] (3,0) --   (4,3);

\draw [->,black] (3,1) --   (4,0);
\draw [->,dashed,black] (3,1) --   (4,1);
\draw [->,black] (3,1) --   (4,2);
\draw [->, dashed,black] (3,1) --   (4,3);

\draw [->,dashed,black] (4,0) --  node[below] {}  (5,0);
\draw [->,black] (4,1) --   (5,1);
\draw [->, dashed,black] (4,2) --   (5,2);
\draw [->,black] (4,3) --   (5,3);

    \end{tikzpicture}
    \end{center}
\end{ex}

Since by the Factorization Theorem all linear trellises can be written as elementary trellis products the above theorem is crucial for their classification. We will use it indeed in the next section for classifying minimal linear trellises.

We want now to count the number of  distinct elementary trellis factorizations of a linear trellis $T$ 
such that $m(\sss,T)=1$ for all $\sss\in\sd(T)$. First, let 
$\otimes_{\sss\in\sd(T)}\al^{\sss}|\sss$ be an elementary trellis factorization of $T$ (which we know to exist by the Factorization Theorem).
Now, by \eqref{eq4}
$\al'$ is a valid labeling of  the elementary factor with span $\sss$ if and only if 
$$\langle\al^{\sss}\rangle+\sum_{\sss'<\sss}C_{\sss'}(T)=\langle\al'\rangle+\sum_{\sss'<\sss}C_{\sss'}(T)$$
As $\al'|\sss=\al''|\sss$ if and only if $\al'=y\al''$ for some $y\in\FF^{*}$, the elementary factors with span $\sss$ are precisely given by $(\al^{\sss}+\ww)|\sss$ for arbitrary $\ww\in\sum_{\sss'<\sss}C_{\sss'}(T)$. Thus, putting $$k_{\sss}=\dim(\sum_{\sss'<\sss}C_{\sss'}(T))$$ we have that:
\begin{itemize}
\item  if $\al^{\sss}\in\sum_{\sss'<\sss}C_{\sss'}(T)$ then 
there are precisely  $$1+\frac{q^{k_{\sss}}-1}{q-1}$$ valid elementary factors with span $\sss$ 
\item if $\al^{\sss}\notin\sum_{\sss'<\sss}C_{\sss'}(T)$ then 
$(\al^{\sss}+\ww)|\sss\neq(\al^{\sss}+\ww')|\sss$ whenever $\ww\neq\ww'$, and so there are precisely  $$q^{k_{\sss}}$$ valid elementary factors with span $\sss$
\end{itemize}
Note that if $T$ is also one-to-one then $\al^{\sss}\notin\sum_{\sss'<\sss}C_{\sss'}(T)$ and $$\dim(\sum_{\sss'<\sss}C_{\sss'}(T))=|\{\sss'\in\sd(T)|\sss'<\sss\}|$$

Resuming:
\begin{cor}\label{cor31}
Assume that $m(\sss,T)=1$ for all $\sss\in\sd(T)$. Put 
\begin{align*}
S'&:=\{\sss\in\sd(T)|C_{\sss}(T)=\sum_{\sss'<\sss}C_{\sss'}(T)\}\\
S''&:=\sd(T)\setminus S'
\end{align*}
Then $T$ has precisely 
$$\prod_{\sss\in S'}\frac{q^{k_{\sss}}+q-2}{q-1}\prod_{\sss\in S''}q^{k_{\sss}}$$ distinct elementary trellis factorizations.
If $T$ is also  one-to-one then it has precisely 
$$\exp_{q}({\sum_{\sss\in\sd(T)}|\{\sss'\in\sd(T)|\sss'<\sss\}|})$$
distinct elementary trellis factorizations.
\end{cor}

For example, the trellis $T$ of the above example is one-to-one with span distribution 
given by %NOT PRESENT
$\mathcal{S}(T)=\{\{(1,2),(3,3),(3,4)\}\}$, and so it has precisely two distinct factorizations.
We also get the following corollary. 
\begin{cor}\label{cor32}
Let $T$ be a one-to-one linear trellis. Then $T$ has a unique elementary trellis factorization if and only if all spans in $\mathcal{S}(T)$ have multiplicity equal to $1$ and are incomparable.
\end{cor} 
This extends the well known result that if there are no containments between different atomic spans of a linear code then the associated minimal conventional trellis factors uniquely, since (up to scalar multiplication) all the atomic generators are uniquely determined.

%%%%%%%%%%%%%%%%%%%%%%%%%%%%%%%%%%%%%%%%%%%%%%%%%%%%

 \section{Applications: classifying  nonmergeable and minimal linear  trellises}\label{charactpap}

  \subsection{Preliminaries}\label{fundmin}
We review here some known fundamental results on the top of which we will build our new results. Some new notation will be also introduce in Subsection \ref{atomspanpap}.  We actually start by discussing and proving some important facts which have been overlooked in the literature.

\subsubsection{\textbf{Overlooked facts on nonmergeability and minimality}}
 While it is obvious that a minimal trellis is nonmergeable the same property for minimal linear trellises (i.e. linear trellises that are minimal only amongst linear trellises --- see Remark \ref{nolinminpap}) requires a proof. The proof essentially amounts to showing that a mergeable linear trellis can be merged to a smaller linear trellis for the same code. Such fact was actually observed in \cite{Ko} (see Lemma $4$ therein), but the proof there is cloudy. In fact it turns out that the ``almost reduced'' hypothesis (i.e. each vertex belongs to some cycle) is necessary, which was overlooked in \cite{Ko}. Below we prove the correct result and then give an example of the necessity of the extra hypothesis. We will also prove that a minimal linear trellis is reduced, which is tied up with the same result.

\begin{thm}\label{merglin} Let $T$ be an almost reduced linear trellis for $C$. If $T$ is mergeable then it can be merged to a linear trellis   
for $C$ smaller than $T$. 
\end{thm}
\begin{proof} 
 Suppose $T$ is mergeable at $v\neq w\in V_{i}(T)$. By shifting $T$ we can assume $i=0$. 
 Let $u\in V_{0}(T)$ and $x\in\FF^{*}$. By hypothesis there exist paths $\bm{p}\in\mathbb{P}(v,v),\bm{p'}\in\mathbb{P}(u,u)$. By  linearity $$\mathbb{P}(u,u+x(v-w))=\bm{p}'+x\bm{p}-x\mathbb{P}(v,w)$$ so that  $L(\mathbb{P}(u,u+x(v-w)))\subseteq C$. Similarly  $L(\mathbb{P}(u+x(v-w),u))\subseteq C$. Thus the vertices in each coset $u+\FF(v-w)$ can be merged together without affecting the code. The resulting trellis is clearly linear and smaller than $T$, proving the statement. 
\end{proof}

\begin{ex} The below linear trellis for $\langle111\rangle$ is mergeable,  
and it is so only at the pair of vertices $\underline{00}$, $\underline{11}\in V_{0}(T)$.  Once those two vertices are merged no other merging is possible, so that there is no way we can obtain a smaller linear trellis by merging vertices. The trellis is clearly not almost reduced.
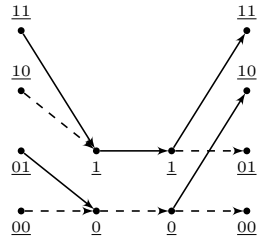
\begin{center}
 \begin{tikzpicture}[yscale=.8,>=latex',shorten >=.9pt, shorten <=1.4pt, line width=.6pt]

\begin{scope}
[xshift=0cm]
  \tikzstyle{every node}=[draw,circle,fill=black,minimum size=2pt,
                        inner sep=0pt]                    
\foreach \x in {0}{
\node at (\x,0) {};};
\foreach \x in {0}{
\node at (\x,1) {};};
\foreach \x in {0}{
\node at (\x,2) {};};
\foreach \x in {0}{
\node at (\x,3) {};};

\tikzstyle{every node}=[]

\draw [->,black,dashed] (0,0) -- (1,0);
\draw [->,black,dashed] (0,2) -- (1,1);
\draw [->,black] (0,3) -- (1,1);
\draw [->,black] (0,1) -- (1,0);

\foreach \x in {0}{
\draw (\x,0) node [below] {\tiny ${\uu{00}}$};
\draw (\x,1) node [below] {\tiny ${\uu{01}}$};
\draw (\x,2) node [above] {\tiny ${\uu{10}}$};
\draw (\x,3) node [above] {\tiny ${\uu{11}}$};
};
%\draw (-.5,1.5) node {$=$};
%\draw (3.5,1.5) node {$=T$};
\end{scope}

\begin{scope}
[xshift=1cm]
  \tikzstyle{every node}=[draw,circle,fill=black,minimum size=2pt,
                        inner sep=0pt]                    
\foreach \x in {0}{
\node at (\x,0) {};};
\foreach \x in {0}{
\node at (\x,1) {};};

\tikzstyle{every node}=[]
\draw [->,black,dashed] (0,0) -- (1,0);

\draw [->,black] (0,1) -- (1,1);

\foreach \x in {0}{
\draw (\x,0) node [below] {\tiny ${\uu{0}}$};
\draw (\x,1) node [below] {\tiny ${\uu{1}}$};
};
\end{scope}

\begin{scope}
[xshift=2cm]
  \tikzstyle{every node}=[draw,circle,fill=black,minimum size=2pt,
                        inner sep=0pt]                    
\foreach \x in {0}{
\node at (\x,0) {};};
\foreach \x in {0}{
\node at (\x,1) {};};
\node at (1,0) {};
\node at (1,1) {};
\node at (1,2) {};
\node at (1,3) {};

\tikzstyle{every node}=[]
\draw [->,black,dashed] (0,0) -- (1,0);
\draw [->,black,dashed] (0,1) -- (1,1);

\draw [->,black] (0,0) -- (1,2);
\draw [->,black] (0,1) -- (1,3);

\foreach \x in {0}{
\draw (\x,0) node [below] {\tiny ${\uu{0}}$};
\draw (\x,1) node [below] {\tiny ${\uu{1}}$};
};
\foreach \x in {1}{
\draw (\x,0) node [below] {\tiny ${\uu{00}}$};
\draw (\x,1) node [below] {\tiny ${\uu{01}}$};
\draw (\x,2) node [above] {\tiny ${\uu{10}}$};
\draw (\x,3) node [above] {\tiny ${\uu{11}}$};
};
\end{scope}

 \end{tikzpicture}
  \end{center}
  \end{ex}

 \begin{thm}\label{minred} Any minimal linear trellis is reduced.
\end{thm}
\begin{proof}
 Let $T$ be a minimal linear trellis for $C$.  
 By deleting all the edges of $T$ that do not belong to cycles we get a linear subtrellis $T'\leq T$ which is reduced and has the same vertices as $T$. Thus $T'$ must be minimal for $C$ too, and so by Theorem \ref{merglin} it must be nonmergeable, and therefore also connected. Now the statement follows  from Theorem \ref{thm45}.  
 \end{proof}

\begin{cor}\label{corpap1}
Any minimal linear trellis is nonmergeable. 
\end{cor}
\begin{proof}
Put together Theorem \ref{merglin} and Theorem \ref{minred}.
\end{proof}

 \begin{rmk}
 The definition of minimal linear trellises commonly given in the literature requires the reduced property in the hypothesis. Theorem \ref{minred} shows that by minimizing the \textit{state-complexity profile} $(|V_{0}(T)|,\ldots,|V_{n-1}(T)|)$ we get automatically rid of the unnecessary edges and so there is no such thing as a nonreduced minimal linear trellis, a fact which was not recognized before. 
\end{rmk}

\subsubsection{\textbf{Well known facts} (\cite{KV,V})}

For  a linear trellis $T$:

\begin{itemize}
\item $T$ is minimal $\implies$ $T$ is one-to-one.
\item If $T$ is reduced then: $T$ is nonmergeable $\implies$ $T$ is biproper, but the converse is not true. 
\item If $T$ is nonmergeable it is not necessarily minimal.
\item If $T$ is conventional then: $T$ is minimal $\iff$ $T$ is nonmergeable $\iff$ $T$ is biproper
\end{itemize}

%%%%%%%%%%%%%%%%%%%%%%%%%%%%%%%%%%%%%%%%%%%%%%%%%%%%%%%%%%%%%%%%%%%%%%%%%%%%%%%%%%%%%%%%%%%%%%%%%%%%%%%%%%%
\subsubsection{\textbf{Atomic bases, atomic spans, and minimal trellises}}\label{atomspanpap}
First, some notation. Henceforth by $C$ we will denote an $[n,k]$ linear code of full support (i.e. $\su(C)=\ZZ_{n}$). 
 We denote by $[\vv]^{*}$ the minimum conventional span of $\vv\in \FF^{n}$ (which clearly exists), and we call it \textit{the (conventional) span of} $\vv$. For any span $(a,l)$ we put $$\sigma((a,l)):=(a-1,l)$$
  Now, we have:

\begin{thm}[\cite{KsSo,Mc}]\label{minconv5}
There exists a basis $\{\vv^{1},\ldots,\vv^{k}\}$ of $C$ such that the spans $[\vv^{i}]^{*}$, $i=1,\ldots,k$, all  start and end at different positions, and  if  $\{\ww^{1},\ldots,\ww^{k}\}$ is another basis with the same property then $\{[\vv^{i}]^{*}\}_{i=1,\ldots,k}=\{[\ww^{i}]^{*}\}_{i=1,\ldots,k}$. 
\end{thm}  
Such bases of $C$ are  called  \textit{atomic}\index{basis!atomic}
since their elements are \textit{atomic codewords}\index{atomic codeword}, i.e. they cannot be written as sums of codewords with shorter conventional span.  
The  uniquely determined set $\{[\vv^{i}]^{*}\}_{i=1,\ldots,k}$ is called the  \textit{atomic span set}\index{atomic span set $\mathcal{S}^{*}(C)$} of $C$. We denote it by 
$$\mathcal{S}^{*}(C)$$
Atomic bases are closely tied up with minimal conventional trellises:

\begin{thm}[\cite{KsSo,KsV,V}]\label{minconv5} 
Let $\{\vv^{i}\}_{i=1,\ldots,k}$ be a basis of $C$. 
Then $\otimes_{i=1}^{k}\vv^{i}|[\vv^{i}]^{*}$ is the minimal conventional trellis $T^{*}(C)$ if and only if $\{\vv^{i}\}_{i=1,\ldots,k}$ is an atomic basis of $C$. 
\end{thm}  

This gives a method to construct $n$ (some of them possibly equal) minimal linear trellises for 
$C$: for all $i=0,\ldots,n-1$, construct the minimal conventional trellis for $\sigma^{i}(C)$ from an  
   atomic basis $\{\vv^{1},\ldots,\vv^{k}\}$  of $\sigma^{i}(C)$ with span set $\mathcal{S}^{*}(\sigma^{i}(C))$, and then shift it backwards by applying $\sigma^{-i}$. As a consequence,  the set 
   $$\mathcal{S}(C):=\cup_{i=0}^{n-1}\sigma^{-i}(\mathcal{S}^{*}(\sigma^{i}(C))$$
   contains precisely all the spans appearing in all the minimal trellises so constructed. Koetter/Vardy \cite{KV} proved the striking result that no other spans are needed to describe the graph structure of any other possible minimal linear trellis for $C$:

\begin{thm}[\cite{KV}]\label{KVmin} The following holds:
\begin{enumerate}
\item If $\otimes_{i=1}^{k}\vv^{i}|\sss^{i}$ is a minimal linear trellis for $C$ then all the spans $\sss^{i}$ are different and $\{\sss^{i}\}_{i=1,\ldots,k}\subseteq\mathcal{S}(C)$.

\item $\mathcal{S}(C)$ contains precisely $n$ spans, and they all start and end at different positions.
\item For all $i\geq0$ there are precisely $k$ conventional spans in $\sigma^{i}(\mathcal{S}(C))$, and these are the atomic spans of $\sigma^{i}(C)$.
\item If $\mathcal{S}(C)=\{(a_{1},l_{1}),\ldots,(a_{n},l_{n})\}$ then $\mathcal{S}(C^{\perp})=\{(a_{1}+l_{1},n-l_{1}),\ldots,(a_{n}+l_{n},n-l_{n})\}$.
\end{enumerate}
\end{thm}

The set $\mathcal{S}(C)$ is thus called the \textit{characteristic span set}\index{characteristic span set $\mathcal{S}(C)$} of $C$. Note that all the spans in $\mathcal{S}(C)$ are nondegenerate, corresponding to the fact that any minimal linear trellis must be connected. A set of $n$ elementary trellises $\vv^{1}|\sss^{1},\ldots,\vv^{n}|\sss^{n}$, such that $\vv^{i}\in C$, $\sss^{i}$ is a minimal span of $\vv^{i}$ for all $i$, and $\{\sss^{i}\}_{i=1,\ldots,n}=\mathcal{S}(C)$, will be called a \textit{characteristic set}\index{characteristic set (or matrix) of $C$} of $C$ (or \textit{characteristic matrix} of $C$ if the trellises are listed in  array form). 
If $\{\sss^{i}\}_{i=1,\ldots,n}$ is a characteristic set of $C$, by the above theorems it follows that 
$$\{\sigma^{j}(\vv^{i})| \ \sigma^{j}(\sss^{i}) \textrm{ is conventional}, i=1,\ldots,n\}$$
 is an atomic basis of $\sigma^{j}(C)$ for all $j=0,\ldots,n-1$.

\begin{ex}
Assume that $C$ is cyclic.  It is clear that if $\sss$ is the span of a generating codeword of $C$ then $\mathcal{S}^{*}(C)=\{\sigma^{i}(\sss)\}_{i=0,\ldots,k-1}$ is the atomic span set of $C$. But then, as $C$ is cyclic, it also follows that 
$$\mathcal{S}(C)=\{\sigma^{i}(\sss)\}_{i=0,\ldots,n-1}$$
 See \cite{RB,SKSR} for more on   
  minimal linear trellises of cyclic codes.
\end{ex}

Any one-to-one trellis for $C$ which is isomorphic to a product of elementary trellises from a characteristic set of $C$ is called a KV-\textit{trellis}\index{trellis!KV-trellis} (notice that Nori \cite{N} calls any product of elementary trellises a KV-trellis, while we are following the terminology used in \cite{GW}). Theorem \ref{KVmin} tells us that any minimal linear trellis is a $\mathrm{KV}$-trellis. The converse is known to be false.
Also, although all minimal linear trellises for $C$ are structurally isomorphic to KV-trellises produced from a single fixed characteristic set of $C$ (see Corollary \ref{cornuovo}), in general it is not true  that they are also all isomorphic to the KV-trellises produced from a single fixed characteristic set, as wrongly stated in Thm 5.5 of \cite{KV}. The simple code of Example \ref{ex51}  yields an easy counterexample. This subtlety was also independently recognized before in \cite{GW}. In general one needs several characteristic sets of $C$ to describe all minimal linear trellises as  KV-trellises.  
Nevertheless, we will show in Subsection \ref{countmin} how  to determine and count all minimal linear trellises from a single characteristic set.

%%%%%%%%%%%%%%%%%%%%%%%%%%%%%%%%%%%%%%%%%%%%%%%%%%%%%%%%%%%%%%%%%%%%%%%%%%%%%%%%%%%%%%%%%%%%%%%%%%%%%%%%%%%

 \subsection{New insight into the nonmergeable property}\label{insmerg}
We give here a new characterization of the nonmergeable property for one-to-one linear trellises.
 Let us start with a couple of definitions. We say that a trellis $T$ of length $n$ is \textit{pathwise one-to-one} if different paths of length $n$ starting at  $V_{0}(T)$  never yield the same edge-label sequence.
  We then say that a trellis $T$ of length $n$ is \textit{fragment one-to-one}\index{trellis!fragment one-to-one}
  if all its cyclic shifts $\sigma^{i}(T)$, $i=0,\ldots,n-1$, are pathwise one-to-one, that is, if any two  distinct paths in $T$ of length $n$ starting at the same time index never yield the same edge-label sequence. 
  
When $T$ is biproper any two paths of length $n$  that start at a same time index $i$ and yield the same edge-label sequence cannot intersect. Hence a biproper conventional trellis $T$ must be fragment one-to-one. Vice versa, if a conventional trellis is fragment one-to-one then  it is easy  to see that it is biproper.  So, for conventional trellises the two concepts are equivalent. On the other hand, for conventional linear trellises the biproper property  is also equivalent to the nonmergeable property.  Thus  the following holds:

 \begin{obs}\label{obs51}
 A conventional linear trellis is nonmergeable if and only if it is fragment one-to-one.
 \end{obs}

 In the nonconventional case it is known that biproper linear trellises may be mergeable. Koetter/Vardy \cite{KV} give  the following example of such a situation: 
 
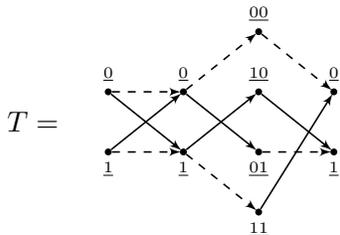
\begin{center}
 \begin{tikzpicture}[yscale=.8,>=latex',shorten >=.9pt, shorten <=1.4pt, line width=.6pt]
  \tikzstyle{every node}=[draw,circle,fill=black,minimum size=2pt,
                        inner sep=0pt]                    
\foreach \x in {1}{
\node at (\x,0) {};};
\foreach \x in {-1,0,1,2}{
\node at (\x,1) {};};
\foreach \x in {-1,0,1,2}{
\node at (\x,2) {};};
\foreach \x in {1}{
\node at (\x,3) {};};

\tikzstyle{every node}=[]
\draw [->,black] (0,1) -- (1,2);
\draw [->,black] (0,2) -- (1,1);

\draw [->,black,dashed] (-1,1) -- (0,1);
\draw [->,black,dashed] (-1,2) -- (0,2);

\draw [->,black] (-1,1) -- (0,2);
\draw [->,black] (-1,2) -- (0,1);

\draw [->,black,dashed] (0,1) -- (1,0);
\draw [->,black,dashed] (0,2) -- (1,3);

\draw [->,black,dashed] (1,1) -- (2,1);
\draw [->,black] (1,2) -- (2,1);

\draw [->,black,dashed] (1,3) -- (2,2);
\draw [->,black] (1,0) -- (2,2);

\foreach \x in {1}{
\draw (\x,0) node [below] {\tiny ${\uu{11}}$};
\draw (\x,1) node [below] {\tiny ${\uu{01}}$};
\draw (\x,2) node [above] {\tiny ${\uu{10}}$};
\draw (\x,3) node [above] {\tiny ${\uu{00}}$};
};
\foreach \x in {-1,0,2}{
\draw (\x,1) node [below] {\tiny ${\uu{1}}$};
\draw (\x,2) node [above] {\tiny ${\uu{0}}$};
}
\draw (-2,1.5) node {$T=$};
 \end{tikzpicture}
  \end{center}
Note that this trellis is one-to-one but  is not fragment one-to-one: the paths $\uu{10}\ms1\uu{1}\ms0\uu{1}\ms1\uu{10}\ms$, $ \uu{11}\ms1\uu{0}\ms0\uu{0}\ms1\uu{01}\ms$ of length $3$ starting at $V_{2}(T)$ both yield the word $101$. This is precisely the reason why $T$ is  mergeable, as we are now going to show that the ``fragment one-to-one'' property is strongly related to the nonmergeable property. Our result can be seen as a general extension of Observation \ref{obs51}.
 
 \begin{thm}\label{thm51}
Let $T$ be a linear trellis. If $T$ is  nonmergeable and one-to-one then $T$ is connected and fragment one-to-one. The converse also holds if $T$ is almost reduced.
 \end{thm}
 \begin{proof}
 Assume $T$ is nonmergable and one-to-one. First, since  $T$ is nonmergeable it must be also connected. We prove now that $T$ is fragment one-to-one.   By using cyclic shifts it is sufficient to prove that $T$ is pathwise one-to-one. So assume there is a path $\bm{p}\in\mathbb{P}(v,w)$, $v,w\in V_{0}(T)$, with edge-label sequence $L(\bm{p})=0$. Then $L(\bm{p}')\in C(T)$ for all $\bm{p}'\in\mathbb{P}(v,w)$, since $\bm{p}'-\bm{p}$ is a cycle. Similarly,  $L(\bm{p}')\in C(T)$ for all $\bm{p}'\in\mathbb{P}(w,v)$, since $\bm{p}'+\bm{p}$ is a cycle. Thus we can merge $v$ and $w$ without affecting the represented code. So it must be $v=w$. But since $T$ is one-to-one, we conclude that $\bm{p}$ is the zero cycle. Hence $T$ is pathwise one-to-one.
 
 As for the converse, assume $T$ is almost reduced, connected and fragment one-to-one. Obviously then $T$ is one-to-one. Assume that $T$ is mergeable at $v\neq w\in V_{i}(T)$. By a cyclic shift we can assume $i=0$. 
 By Corollary \ref{cor45}  we know that there exists $\bm{p}\in\mathbb{P}(v,w)$. So $L(\bm{p})\in C(T)$. But since $T$ is fragment one-to-one $\bm{p}$ must then be a cycle, i.e. $v=w$, a contradiction. So $T$ is nonmergeable.
 \end{proof}

Since minimal linear trellises are nonmergeable (Corollary \ref{corpap1}) and one-to-one, we conclude that:

\begin{cor}\label{cor51}
Any minimal linear  trellis is fragment one-to-one. 
\end{cor}
 
In \cite{GW} it was proven that any \textnormal{KV}-trellis is nonmergeable, thus we can extend the above corollary  (recall that any minimal linear trellises is a \textnormal{KV}-trellises). 
 
 \begin{cor}\label{cor52}
 Any \textnormal{KV}-trellis is fragment one-to-one.   
 \end{cor}
 
 However, Example IV.12 of \cite{GW} shows that a nonmergeable one-to-one linear trellis does not have to be a \textnormal{KV}-trellis. Thus not all fragment one-to-one linear trellises are \textnormal{KV}-trellises.
 
 Note that the proof given in \cite{GW}  of KV-trellises being nonmergeable requires  quite a long detour into BCJR-trellises (see the same reference for the definition). In fact it is proven there that  special types of BCJR-trellises are nonmergeable and that KV-trellises  are isomorphic to such special trellises. In sight of our new characterization of nonmergeable (one-to-one) trellises, it would be interesting to see if it is possible to give a direct proof that \textnormal{KV}-trellises  are fragment one-to-one. This would also make  more clear in what circumstances linear trellises fail to be KV-trellises. We leave this problem for future investigations.

 \subsection{Classifying nonmergeable trellises via multicycle codes}\label{secminimal}

In general a linear code has several  minimal linear trellis representations, and even more nonmergeable trellis representations. The problem of  how can we distinguish and thus classify these representations  has never been posed before.
Thanks to the characterization from the previous subsection we can now address this problem  for the class of  nonmergeable, one-to-one, reduced, linear trellises (and so in particular for minimal linear trellises).  Our result will tell us that it is possible to classify these trellises by means of edge-label sequences of closed paths of length longer than the length of $T$.  

To that purpose we first introduce some terminology. Let $T$ be a trellis of length $n$. A closed path in $T$ starting at $V_{0}(T)$ of length greater than $n$ will be called a \textit{multicycle}\index{multicycle} of $T$. Clearly a multicycle must have length $in$ for some $i>1$.  A multicycle of length $in$ will be also called an $i$-\textit{cycle}\index{cycle!$i$-cycle} of $T$. 
The code of edge-label 
sequences of $i$-cycles will be denote by $$C^{i}(T)$$  
and 
it will be  referred to also as the $i$-th \textit{code represented by}  $T$\index{code represented by $T$!$i$-th code represented by $T$}. So, a trellis does not represent merely a single code, but rather the sequence of codes $\{C^{i}(T)\}_{i=1}^{\infty}$.

We can now state and prove the announced result:

 \begin{thm}\label{thm53} Let $T$,  $T'$ be  fragment one-to-one, connected, and reduced linear trellises. Suppose that $$C^{i}(T)=C^{i}(T')$$ for some $i>1$. Then $T$ and $T'$ are isomorphic.
 \end{thm}
 \begin{proof} We use Corollary \ref{linisocor} of the Isomorphy Theorem \ref{thm15}. We need to show that $C_{(a,l)}(T)=C_{(a,l)}(T')$ for all nondegenerate spans $(a,l)$. By a cyclic shift, it is sufficient to show it for conventional spans. So, let $(a,l)$ be a conventional span, and let $\al\in C_{(a,l)}(T)$. Then clearly 
 $$\bm{\beta}:=\al0\ldots0\in C^{i}(T)$$
  where we are appending $(i-1)n$ zeros to the right of $\al$. Since $C^{i}(T)=C^{i}(T')$, there must exist $v_{j}\in V_{j}(T')$, $j=0,\ldots,in-1$, such that 
  $$\la:=v_{0}\beta_{0}v_{1}\beta_{1}\ldots v_{in-1}\beta_{in-1}$$
   is an $i$-cycle of $T'$. Now, since $\al\in C_{(a,l)}(T)$, the edge-label sequence of the path $$v_{a+l+1}\beta_{a+l+1}v_{a+l+2}\beta_{a+l+2}\ldots v_{in-1}\beta_{in-1}v_{0}\beta_{0}\ldots v_{a}$$ is the zero word. But since $i>1$, this path has length at least $n$, where $n$ is the length of $T'$, so all its vertices must be zero, because  $T'$ is fragment one-to-one. So we conclude that  
   $$v_{0}\beta_{0}v_{1}\beta_{1}\ldots v_{n-1}\beta_{n-1}\in\sal(T')$$
    and thus $\al\in C_{(a,l)}(T')$. So $C_{(a,l)}(T)\subseteq C_{(a,l)}(T')$. Symmetrically $C_{(a,l)}(T')\subseteq C_{(a,l)}(T)$, and  we are done.
 \end{proof}

\begin{cor}\label{cor53}
Let $T$, $T'$ be  nonmergeable, one-to-one, and reduced linear trellises (e.g. minimal linear trellises, or KV-trellises). If $C^{i}(T)=C^{i}(T')$ for some $i>1$ then $T$ and $T'$ must be isomorphic.
\end{cor} 
 
This is a striking result since it tells us that the code of $2$-cycles $C^{2}(T)$ (or $C^{i}(T)$ for any $i>1$) completely determines a trellis $T$ in the mentioned class (which is a large enough class for potential applications, especially since it contains the whole class of KV-trellises). In particular, while in general a linear code $C$ has several different minimal linear trellis representations  $T_{1},\ldots,T_{r}$ which thus satisfy $C=C(T_{1})=\ldots=C(T_{r})$, they must also satisfy $C^{2}(T_{i})\neq C^{2}(T_{j})$ for $i\neq j$. See Example 9 for instance.

 Note also that this result extends what is known for conventional trellises. In fact, it is known that a nonmergeable (equivalently, minimal) conventional trellis for a linear code $C$ is completely determined by $C$. But if $T$ and $T'$ are  conventional trellises then $C(T)=C(T')$ if and only if $C^{i}(T)=C^{i}(T')$, since $C^{i}(T)=C(T)\times\ldots\times C(T)$. Thus the above theorem includes as a special case the uniqueness  of minimal conventional trellises.  In addition, it proves that a minimal linear trellis $T$ is conventional (and so it is the unique minimal conventional trellis) if and only if $$C^{2}(T)=C(T)\times C(T)$$

Moreover this result gives a new method to determine whether or not two minimal linear trellises are equal (i.e. isomorphic). Indeed this  method is alternative to the methods given in Sections \ref{algframpap} and \ref{factpap}  for general linear trellises, which require  the knowledge respectively of all span subcodes and of an elementary trellis factorization, while this information is not required in the above result. In sight of potential applications it would be interesting to determine which method is quicker, i.e. which one yields the algorithm with lowest complexity. We leave this question for future investigations.

We can also apply the above result to deduce an interesting characterization of self-duality for KV-trellises. In order to do so, we need first a definition and a lemma.  Given a trellis $T$ of length $n$ and $i\geq1$ we define the \textit{$i$-cover of $T$} as the trellis $T^{i}$ of length $in$ given by 
\begin{align*}
V_{j}(T^{i})&:=V_{j}(T)\\
E_{j}(T^{i})&:=E_{j}(T)
\end{align*} 
for all $j\in\ZZ_{in}$. 
By definition we have $V_{j}(T^{i})=V_{j+n}(T^{i})$ and $E_{j}(T^{i})=E_{j+n}(T^{i})$ for all $j\in\ZZ_{in}$. Thus the trellis diagram of $T^{i}$ is obtained by continuing $i$ times the trellis diagram of $T$. In graph-theoretical language $T^{i}$ is  the unique (directed) graph $i$-cover of $T$ whose girth is $i$ times the girth of $T$, which justifies the terminology.  Clearly  
$$C^{i}(T)=C(T^{i})$$
 i.e. multicycles of $T$ are cycles in {trellis covers}\index{trellis!cover} of $T$.

Now, we have the following:

\begin{lem} Let $i\geq1$. If $T$ is linear then so is $T^{i}$, and  $$(T^{i})^{\perp}=(T^{\perp})^{i}$$ 
As a consequence $$(C^{i}(T))^{\perp}=C^{i}(T^{\perp})$$
 In particular if $T$ is self-dual then $C^{i}(T)$ is self-dual for all $i\geq1$.
\end{lem}
\begin{proof}
By definition of $T^{i}$ it is clear that if $T$ is linear so is $T^{i}$. On the other hand $T^{\perp}$ is defined locally (i.e. $V_{i}(T^{\perp})$, $V_{i+1}(T^{\perp})$, and $E_{i}(T^{\perp})$ depend only on $V_{i}(T^{})$, $V_{i+1}(T^{})$, and $E_{i}(T^{})$), and so, since $T^{i}$ and $T$  are locally equal, it follows that $(T^{i})^{\perp}=(T^{\perp})^{i}$.
For example, in the binare case $\FF\equiv\FF_{2}$ 
we have that 
\begin{align*}
V_{j}((T^{\perp})^{i})&=V_{j}(T^{\perp})=V_{j}(T)=V_{j}(T^{i})=V_{j}((T^{i})^{\perp})\\
E_{j}((T^{\perp})^{i})&=E_{j}(T^{\perp})=E_{j}(T)^{\perp}=E_{j}(T^{i})^{\perp}=E_{j}((T^{i})^{\perp})\end{align*}
 for all $j\in\ZZ_{in}$. Finally, we have $$(C^{i}(T))^{\perp}=(C(T^{i}))^{\perp}=C((T^{i})^{\perp})=C((T^{\perp})^{i})=C^{i}(T^{\perp})$$
\end{proof}

\begin{cor}\label{corkvdual}
Let $T$ and $T'$ be KV-trellises. Then $T^{\perp}\sim T'$ if and only if $$(C^{i}(T))^{\perp}=C^{i}(T')$$  for some $i>1$. In particular,  if $C^{i}(T)$ is self-dual for some $i>1$ then $T$ is self-dual (and so $C^{i}(T)$ is self-dual for all $i\geq1$).
\end{cor}
\begin{proof}
By Theorem IV.3 of \cite{GW2} we know that $T^{\perp}$ is also a KV-trellis, so the statement follows from Corollary \ref{cor53} and the above lemma.
\end{proof}
 
%%%%%%%%%%%%%%%%%%%%%%%%%%%%%%%%%%%%%%%%%%%%%
%%%%%%%%%%%%%%%%%%%%%%%%%%%%%%%%%%%%%%%%%%%%%
%%%%%%%%%%%%%%%%%%%%%%%%%%%%%%%%%%%%%%%%%%%%%

\subsection{Determining and counting minimal linear trellises}\label{countmin}
Let $C$ be an $[n,k]$ linear code (with full support). We want to determine and consequently count all the minimal linear trellises for $C$ with the same underlying graph structure. As we pointed out before taking products of elementary trellises from a single characteristic set of $C$ is not sufficient for that task. We will make use of what we have proved in Section \ref{factpap}. 

We start with proving the following important theorem. 

\begin{thm}\label{thm54}
Let $T=\al^{1}|\sss_{1}\otimes\ldots\otimes \al^{k}|\sss_{k}$ be a minimal linear trellis for $C$.
Let $\bm{\beta}^{1},\ldots,\bm{\beta}^{k}\in C$ such that $\sss_{i}$ is a minimal span of $\bm{\beta}^{i}$ for all $i$. Then $\bm{\beta}^{1},\ldots,\bm{\beta}^{k}$ is a basis of $C$. As a consequence, $T'=\bm{\beta}^{1}|\sss_{1}\otimes\ldots\otimes \bm{\beta}^{k}|\sss_{k}$ is also a minimal linear trellis for $C$.
\end{thm}
\begin{proof}
By changing generators one at a time it is clearly sufficient to prove the case where we change only one generator, say $\al^{1}\neq \bm{\beta}^{1}$ and $\al^{i}=\bm{\beta}^{i}$ for $i=2,\ldots,k$. Now, assume that 
$\bm{\beta}^{1},\al^{2},\ldots,\al^{k}$ are linearly dependent. Then $$\bm{\beta}^{1}=x_{2}\al^{2}+\ldots+x_{k}\al^{k}$$ for $x_{i}\in\FF$.
 Since $\sss_{1}$ is a minimal span for both $\al^{1}$ and $\bm{\beta}^{1}$, there exists $y\in\FF$ such that $\vv:=\al^{1}-y \bm{\beta}^{1}$ has a span $\sss\lneq\sss_{1}$. It follows that $\vv,\al^{2},\ldots,\al^{k}$ is a basis: if not then $\vv\in\FF \al^{2}+\ldots\FF \al^{k}$, and so also 
$$\al^{1}\in\FF \bm{\beta}^{1}+\FF \al^{2}+\ldots\FF \al^{k}\subseteq\FF \al^{2}+\ldots\FF \al^{k}$$ which is impossible by our assumption. But then we conclude that the trellis $T':=\vv|\sss\otimes \al^{2}|\sss_{2}\ldots\otimes \al^{k}|\sss_{k}$ represents $C$ and is smaller than $T$, which is a contradiction.
\end{proof}

Note that the above theorem was given as  an open problem in \cite{GW}. We have been later informed by Gluesing-Luerssen that Elizabeth  Weaver has independently proved it in the related context of counting characteristic matrices. 

 Combined with Theorem \ref{KVmin} the above theorem yields the following crucial corollary:
 
 \begin{cor}\label{cornuovo} Let $T$ be a minimal linear trellis for an $[n,k]$ linear code and let $\chi$ be a characteristic set of $C$. Then $T$ is structurally isomorphic to $\otimes_{i=1}^{k}\alpha^{i}|\sss_{i}$ for some 
 $\alpha^{1}|\sss_{1},\ldots,\alpha^{k}|\sss_{k}\in\chi$.
 \end{cor}

Now, let 
$$\chi=\{\al^{\sss}|\sss\}_{\sss\in\mathcal{S}(C)}$$
 be a fixed characteristic set of $C$, where $\mathcal{S}(C)$ is the characteristic span set of $C$. Let $S\subseteq \mathcal{S}(C)$ be a subset of $k$ spans for which there exists a minimal linear trellis for $C$ whose span set is precisely $S$. 
 By the above theorem $
\otimes_{\sss\in S}\al^{\sss}|\sss$ is  also a minimal linear trellis for $C$. Then we have the following main result:

\begin{thm}\label{thm55}
Let $T$ be a linear trellis. Then $T$ is a minimal linear trellis for $C$ which is structurally isomorphic to the minimal linear trellis $\otimes_{\sss\in S}\al^{\sss}|\sss$ for $C$ if and only if   
$$T\sim\otimes_{\sss\in S}(\al^{\sss}+\ww^{\sss})|\sss$$ for some $\ww^{\sss}\in \lb\al^{\sss'}|\sss'<\sss,\sss'\in\mathcal{S}(C)\setminus S\rb$, $\sss\in{S}$.  
Moreover, if $$\otimes_{\sss\in S}(\al^{\sss}+\ww^{\sss})|\sss\sim\otimes_{\sss\in S}(\al^{\sss}+\ww'^{\sss})|\sss$$ for some  $\ww^{\sss},\ww'^{\sss}\in \lb\al^{\sss'}|\sss'<\sss,\sss'\in\mathcal{S}(C)\setminus S\rb$, $\sss\in{S}$, then $\ww^{\sss}=\ww'^{\sss}$ for all $\sss\in S$.
\end{thm}
\begin{proof}
Throughout we will reserve the notation $\ww^{\sss}$ for elements of $\lb\al^{\sss'}|\sss'<\sss,\sss'\in\mathcal{S}(C)\setminus S\rb$.
Now, assume that $T$ is a minimal linear trellis for $C$ which is structurally isomorphic to $
\otimes_{\sss\in S}\al^{\sss}|\sss$. Then by  the Factorization Theorem and Theorem \ref{unfact1} we have that 
$$T\sim\otimes_{\sss\in S}\bm{\beta}^{\sss}|\sss$$
 for some
$\bm{\beta}^{\sss}$, $\sss\in S$. Since $T$ is minimal, $\sss$ must be a minimal span of $\bm{\beta}^{\sss}$.
By Theorem \ref{KVmin}  for each $\sss\in S$ there exists $r\geq0$ such that $\sigma^{r}(\sss)$ is an atomic span of $\sigma^{r}(C)$. Since $\chi=\{\al^{\sss}|\sss\}_{\sss\in\mathcal{S}(C)}$ is a characteristic set,  by properties of atomic bases it follows  that $\bm{\beta}^{\sss}\in\lb\al^{\sss'}|\sss'\leq\sss,\sss'\in\mathcal{S}(C)\rb$ for each $\sss\in S$. Thus for each $\sss\in S$ we have $$\bm{\beta}^{\sss}=x^{\sss}\al^{\sss}+\vv^{\sss}+\ww^{\sss}$$ for some $x^{\sss}\in\FF$, $\vv^{\sss}\in\lb\al^{\sss'}|\sss'\in S, \sss'<\sss\rb$, and $\ww^{s}$. By rescaling we can assume that $\bm{\beta}^{\sss}=\al^{\sss}+\vv^{\sss}+\ww^{\sss}$. In particular $\bm{\beta}^{\sss}=\al^{\sss}+\ww^{\sss}$ for all the minimal spans $\sss$ in $S$.
Now, by Theorem \ref{trelfac} we know that if $\otimes_{\sss\in S}\bm{c}^{\sss}|\sss$ is a factorization of $T$ and $\bm{b}^{\sss}\in\lb\bm{c}^{\sss'}|\sss'\in S, \sss'<\sss\rb$
then 
$$\otimes_{\sss\in S}\bm{c}^{\sss}+\bm{b}^{\sss}|\sss$$ 
is also a factorization of $T$. Thus  starting with 
the spans $\sss\in S$ directly above the minimal spans of $S$ and going up, we can transform each codeword $\bm{\beta}^{\sss}$ into the form $\al^{\sss}+\ww^{\sss}$ by adding an appropriate multiple of  the ones with smaller span, and so get the sought factorization of $T$.     

Vice versa, assume $T\sim\otimes_{\sss\in S}(\al^{\sss}+\ww^{\sss})|\sss$. Clearly $T$ is structurally isomorphic to $\otimes_{\sss\in S}\al^{\sss}|\sss$, 
 since the spans of the factors are the same. Also, it is clear that $\sss\in S$ is a minimal span of $\al^{\sss}+\ww^{\sss}$  since the characteristic spans all start and end at different positions. But then by Theorem \ref{thm54} we deduce that $T$ must be minimal for $C$.

Finally, assume that $$T\sim\otimes_{\sss\in S}(\al^{\sss}+\ww^{\sss})|\sss\sim\otimes_{\sss\in S}(\al^{\sss}+\ww'^{\sss})|\sss$$
By Theorem \ref{trelfac} we have that  
$$\lb\al^{\sss}+\ww^{\sss}\rb+\sum_{\sss'<\sss}C_{\sss'}(T)=\lb\al^{\sss}+\ww'^{\sss}\rb+\sum_{\sss'<\sss}C_{\sss'}(T)$$ for every $\sss\in S$. 
Since $T$ is one-to-one we must have $$\al^{\sss}+\ww^{\sss}\notin \sum_{\sss'<\sss}C_{\sss'}(T)$$
So from the above equality we deduce that $$\ww^{\sss}-\ww'^{\sss}\in\sum_{\sss'<\sss}C_{\sss'}(T)=\lb\al^{\sss'}+\ww^{\sss'}|\sss'<\sss,\sss'\in S\rb$$ On the other hand $\ww^{\sss}-\ww'^{\sss}\in\lb\al^{\sss'}|\sss'<\sss,\sss'\in\mathcal{S}(C)\setminus S\rb$. But if $A$ and $B$ are disjoint sets of characteristic spans all sitting inside the same span then $\lb\al^{\sss}|\sss\in A\rb\cap\lb\al^{\sss}|\sss\in B\rb=0$, since different characteristic spans start and end at different positions. So $\ww^{\sss}=\ww'^{\sss}$ for all $\sss\in S$.
\end{proof}

A straight application of the above theorem gives us  the possibility to count minimal linear trellises with the same underlying graph structure.
 
\begin{cor}\label{cor55}
Let $T=\otimes_{\sss\in S}\al^{\sss}|\sss$ be a minimal linear trellis for $C$. Then the number of distinct minimal linear trellises for $C$ that are structurally isomorphic to $T$ is
\begin{equation}\label{mincount}
\exp_{q}\left(\sum_{\sss\in S}\#\{\sss'<\sss,\sss'\in \mathcal{S}(C)\setminus S\}\right)
\end{equation} 
\end{cor}
From the above 
we deduce immediately the following two  corollaries for trellises and codes with special characteristic span distributions.

\begin{cor}\label{cor56}
Let $T=\otimes_{\sss\in S}\al^{\sss}|\sss$ be a minimal linear trellis for $C$ such that no span in $S$ contains a 
   characteristic span of $C$ not in $S$. Then there are no other minimal trellises for $C$ with same graph structure as that of $T$. In particular, if $C$ is self-dual and $T$ is structurally isomorphic to $T^{\perp}$ then $T$ is self-dual.
\end{cor}

\begin{cor}\label{cor57}
If there are no containments between the characteristic spans of $C$ then two distinct minimal linear trellises for $C$ are never structurally isomorphic. 
In particular, this holds true for cyclic codes.
\end{cor}

\begin{rmk}\label{rmkc7}
By similar arguments to the ones used in the proof of Theorem \ref{thm55} one can show that Corollary \ref{cor56} (and hence also Corollary \ref{cor57}) holds for KV-trellises too.
\end{rmk}

By Corollary \ref{cornuovo} and Theorem \ref{thm55} to determine and count all the minimal linear trellises for a given $[n,k]$ code $C$ we can proceed as follows:

\begin{enumerate}
\item Compute a characteristic set $\chi$ of $C$ (using for exampe ``Algorithm A'' from \cite{KV}).
\item Find all the possible KV-trellises yielded by $\chi$ (i.e. find all the possible subsets of $k$ elements of $\chi$ whose codewords are linearly independent).
\item Apply a sorting algorithm to the so found list of KV-trellises to find out which of those are minimal. 
\item Apply Theorem \ref{thm55} (and Corollary \ref{cor55}) to each minimal linear trellis so found.
\end{enumerate}

\begin{ex}\label{ex51}
Let $C=\lb01010,11111\rb$. Then a characteristic matrix for $C$ is given by
\begin{equation*}
\begin{pmatrix}
0{101}0|(1,2)\\
{01}0{10}|(3,3)\\
{10101}|(0,4)\\
{1}0{101}|(2,3)\\
{101}0{1}|(4,3)\\
\end{pmatrix}
\end{equation*}
A  trellis for $C$ from this characteristic matrix can be built only as a product of one of the first rows with one of the last three rows, which gives 6 possibilities. All these possibilities turn out to be  minimal linear trellises. Five of them are shifted conventional trellises.  
The remaining one is the product of the second and third row, and has span set $S=\{(3,3),(0,4)\}$. Note that $(3,3)$ contains no other characteristic span, while $(0,4)$ contains $(1,2)$. Thus there are precisely two nonisomorphic minimal linear trellises for $C$ with span distribution equal to $S$. These are 
depicted just below:
\begin{center}
 \begin{tikzpicture}[yscale=.8,>=latex',shorten >=.9pt, shorten <=1.4pt, line width=.6pt]
  \tikzstyle{every node}=[draw,circle,fill=black,minimum size=2pt,
                        inner sep=0pt]                    
\foreach \x in {0,1,2,3,4,5}{
\node at (\x,2) {};};
\foreach \x in {0,1,2,3,4,5}{
\node at (\x,1) {};};
\foreach \x in {1,4}{
\node at (\x,3) {};};
\foreach \x in {1,4}{
\node at (\x,-0) {};};
\tikzstyle{every node}=[]
  %\draw ['] --  node [above] {${01010|{(3,1]}\otimes11111|{(0,4]}:}$} (-6,2.5);
\draw [->,black] (0,1) -- node[above] {} (1,1);
\draw [->,black, dashed] (1,1) -- node[pos=.05, above] {} (2,2);
\draw [->,black] (2,2) -- node[above] {} (3,2);
\draw [->,black] (3,2) -- node[above] {} (4,3);
\draw [->,black, dashed] (3,2) -- node[pos=.05, below] {} (4,1);
\draw [->,black] (4,3) -- node[above] {} (5,2);
\draw [->,black] (4,1) -- node[above] {} (5,1);
\draw [->,black] (0,2) -- node[above] {} (1,3);
\draw [->,black] (1,3) -- node[above] {} (2,2);
\draw [->,black, dashed] (0,2) -- node[above] {} (1,2);
\draw [->,black] (1,2) -- node[pos=.05, below] {} (2,1);
\draw [->,black, dashed] (0,1) -- node[above] {} (1,0);
\draw [->,black, dashed] (1,0) -- node[above] {} (2,1);
\draw [->,black, dashed] (2,1) -- node[above] {} (3,1);
\draw [->,black, dashed] (3,1) -- node[above] {} (4,0);
\draw [->,black, dashed] (4,0) -- node[above] {} (5,1);
\draw [->,black] (3,1) -- node[pos=.05, above] {} (4,2);
\draw [->,black, dashed] (4,2) -- node[above] {} (5,2);
                                     \draw [] (2.5,4) --  node [] {${01}0{10}|{(3,3)}\otimes{10101}|{(0,4)}=$} (2.5,4);

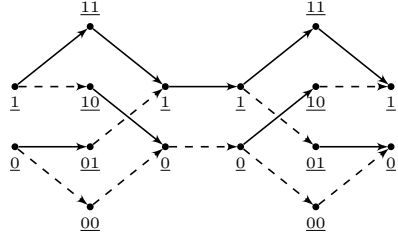
\draw (0,2) node [below] {\tiny$\uu{1}$};
\draw (1,2) node [below] {\tiny$\uu{10}$};
\draw (2,2) node [below] {\tiny$\uu{1}$};
\draw (3,2) node [below] {\tiny$\uu{1}$};
\draw (4,2) node [below] {\tiny$\uu{10}$};
\draw (5,2) node [below] {\tiny$\uu{1}$};
\draw (0,1) node [below] {\tiny$\uu{0}$};
\draw (1,1) node [below] {\tiny$\uu{01}$};
\draw (2,1) node [below] {\tiny$\uu{0}$};
\draw (3,1) node [below] {\tiny$\uu{0}$};
\draw (4,1) node [below] {\tiny$\uu{01}$};
\draw (5,1) node [below] {\tiny$\uu{0}$};

\draw (1,0) node [below] {\tiny$\uu{00}$};
\draw (4,0) node [below] {\tiny$\uu{00}$};
\draw (1,3) node [above] {\tiny$\uu{11}$};
\draw (4,3) node [above] {\tiny$\uu{11}$};
%\draw (2,2) node [above] {\tiny\uu{1}};
   \end{tikzpicture}
      \end{center}   
      \begin{center}
 \begin{tikzpicture}[yscale=.8,>=latex',shorten >=.9pt, shorten <=1.4pt, line width=.6pt]
  \tikzstyle{every node}=[draw,circle,fill=black,minimum size=2pt,
                        inner sep=0pt]                    
\foreach \x in {0,1,2,3,4,5}{
\node at (\x,2) {};};
\foreach \x in {0,1,2,3,4,5}{
\node at (\x,1) {};};
\foreach \x in {1,4}{
\node at (\x,3) {};};
\foreach \x in {1,4}{
\node at (\x,-0) {};};
\tikzstyle{every node}=[]
                                    \draw [] (2.5,4) --  node [above] {${01}0{10}|{(3,3)}\times{11111}|{(0,4)}=$} (2.5,4);

  %\draw ['] --  node [above] {${01010|{(3,1]}\otimes11111|{(0,4]}:}$} (-6,2.5);
\draw [->,black] (0,1) -- node[above] {} (1,1);
\draw [->,black] (1,1) -- node[pos=.05, above] {} (2,2);
\draw [->,black] (2,2) -- node[above] {} (3,2);
\draw [->,black, dashed] (3,2) -- node[above] {} (4,3);
\draw [->,black] (3,2) -- node[pos=.05, below] {} (4,1);
\draw [->,black] (4,3) -- node[above] {} (5,2);
\draw [->,black] (4,1) -- node[above] {} (5,1);
\draw [->,black] (0,2) -- node[above] {} (1,3);
\draw [->,black, dashed] (1,3) -- node[above] {} (2,2);
\draw [->,black, dashed] (0,2) -- node[above] {} (1,2);
\draw [->,black] (1,2) -- node[pos=.05, below] {} (2,1);
\draw [->,black, dashed] (0,1) -- node[above] {} (1,0);
\draw [->,black, dashed] (1,0) -- node[above] {} (2,1);
\draw [->,black, dashed] (2,1) -- node[above] {} (3,1);
\draw [->,black, dashed] (3,1) -- node[above] {} (4,0);
\draw [->,black, dashed] (4,0) -- node[above] {} (5,1);
\draw [->,black] (3,1) -- node[pos=.05, above] {} (4,2);
\draw [->,black, dashed] (4,2) -- node[above] {} (5,2);

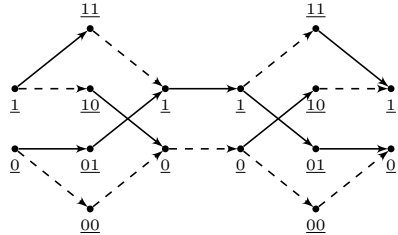
\draw (0,2) node [below] {\tiny$\uu{1}$};
\draw (1,2) node [below] {\tiny$\uu{10}$};
\draw (2,2) node [below] {\tiny$\uu{1}$};
\draw (3,2) node [below] {\tiny$\uu{1}$};
\draw (4,2) node [below] {\tiny$\uu{10}$};
\draw (5,2) node [below] {\tiny$\uu{1}$};
\draw (0,1) node [below] {\tiny$\uu{0}$};
\draw (1,1) node [below] {\tiny$\uu{01}$};
\draw (2,1) node [below] {\tiny$\uu{0}$};
\draw (3,1) node [below] {\tiny$\uu{0}$};
\draw (4,1) node [below] {\tiny$\uu{01}$};
\draw (5,1) node [below] {\tiny$\uu{0}$};

\draw (1,0) node [below] {\tiny$\uu{00}$};
\draw (4,0) node [below] {\tiny$\uu{00}$};
\draw (1,3) node [above] {\tiny$\uu{11}$};
\draw (4,3) node [above] {\tiny$\uu{11}$};
%\draw (2,2) node [above] {\tiny\uu{1}};
   \end{tikzpicture}
      \end{center}     
In particular, the minimal trellis $01010|(3,3)\otimes11111|(0,4)$ is not equal to any $\mathrm{KV}$-trellis  coming from the above characteristic set, while $01010|(3,3)\otimes10101|(0,4)$ is not equal to any $\mathrm{KV}$-trellis coming from the only other possible characteristic set of $C$.   
We conclude also that in total $C$ has 7 minimal linear trellises. 

Note that the 5 minimal linear trellises which are shifted conventional trellises give rise to the following codes of 2-cycles: $$C\times C, \sigma(C\times C), \sigma^{2}(C\times C), \sigma^{3}(C\times C), \sigma^{4}(C\times C)$$
 The other two trellises instead give rise to:
\begin{gather*}
C^{2}(01010|(3,3)\otimes10101|(0,4))=\\=\lb1010110101,0101001010,
1011101000,1011111101\rb
\\
C^{2}(01010|(3,3)\otimes11111|(0,4))=\\=\lb1010110101,0101001010,
1110101000,1110110111\rb
\end{gather*}
As predicted by Corollary \ref{cor53} one can check that these 7 codes of length 10 are all different.
\end{ex}

Below we give another interesting example of the potential application of the above results, where we deduce that the Golay trellis \cite{CFV} is self-dual.

\begin{ex}\label{golay2} 
 Consider the Golay trellis $T_{\mathcal{G}}$  presented in \cite{CFV} (see also \cite{CFV0}) as 
 $T_{\mathcal{G}}=\otimes_{i=0}^{2}\sigma^{8i}(T_{1}\otimes T_{2}\otimes T_{3}\otimes T_{4})$ 
where
\begin{align*}
T_{1}&={1101110111}00000000000000|(0,9)\\
T_{2}&=00{1111100111}000000000000|(2,9)\\
T_{3}&=0000{1101101111}0000000000|(4,9)\\
T_{4}&=000000{1101110111}00000000|(6,9)
\end{align*}
$T_{\mathcal{G}}$ is  a minimal trellis for the binary $[24,12,8]$ Golay code $\mathcal{G}$ (with the given special coordinate ordering). It is actually a very special trellis, since it achieves simultaneously the minimum value for $\prod_{i}|V_{i}(T)|$  and   $\max \{|V_{i}(T)|\}_{i}$ for any possible $T$ for $\mathcal{G}$ under any possible coordinate ordering. Now, the spans of its factors have all length $9$. Thus by Theorem \ref{KVmin}, since $\mathcal{G}$ is  self-dual, the other $12$ characteristic spans have all length $15$, and so cannot be contained in any span of  $T_{\mathcal{G}}$. Hence by Corollary \ref{cor57} there is no other minimal trellis  for $\mathcal{G}$ with the same graph structure as that of $T_{\mathcal{G}}$.
Moreover, since the dual trellis $T_{\mathcal{G}}^{\perp}$ has the same state-complexity profile as that of $T_{\mathcal{G}}$, we have that $$\prod_{i}|V_{i}(T_{\mathcal{G}}^{\perp})|=\prod_{i}|V_{i}(T_{\mathcal{G}})|$$ from which follows that $\mathcal{S}(T_{\mathcal{G}}^{\perp})$ must be also made up of the $12$ spans of length $9$ in $\mathcal{S}(\mathcal{G})$ (as $T_{\mathcal{G}}^{\perp}$ also represents and hence is minimal for $\mathcal{G}=\mathcal{G}^{\perp}$), and so  $T_{\mathcal{G}}^{\perp}$ is structurally isomorphic to     $T_{\mathcal{G}}$. 
But then 
$T_{\mathcal{G}}\sim T_{\mathcal{G}}^{\perp}$, i.e. $T_{\mathcal{G}}$ is self-dual.

\end{ex}

\subsection{Improving iterative/LP trellis decoding through the complete classification of minimal linear trellises}\label{pointout}
The complexity of trellis decoding is directly proportional to  trellis size, so, to achieve low complexity it is necessary to search for the smallest trellis representations of codes. If trellis size is all what one is interested in then obviously the way labels are arranged on a trellis does not matter (as long as the trellis represents the prescribed code), and it is thus sufficient to classify all  trellis representations up to structural isomorphism. The work done by Koetter/Vardy \cite{KV} 
goes precisely in that direction as it actually focuses on the structural classification of minimal linear trellises for linear codes.

On the other hand, the performance and behavior of iterative and LP trellis decoding is affected by so-called \textit{pseudocodewords} (see \cite{F, FKKR,FKMT,H,KVo}), which gives importance to sorting trellises also with respect to their pseudocodewords. In fact, one wants to find trellises that yield few bad pseudocodewords when performing iterative/LP decoding. 

Now, consider the two minimal linear trellises $T=01010|(3,3)\otimes10101|(0,4)$ and $T'=01010|(3,3)\otimes11111|(0,4)$ for the code $C=\lb01010,10101\rb$  depicted in Example \ref{ex51}.  
$T'$ yields the (unscaled) pseudocodeword $12101\in\RR^{5}$ (arising from the unique $2$-cycle with edge-label sequence $1110101000$), while as one can easily check all the (unscaled) pseudocodewords of $T$ are sums of the codewords of $C$ seen as a subset of $\RR^{5}$ (via the map $\FF_{2}\ni0\mapsto0\in\RR$, $\FF_{2}\ni1\mapsto1\in\RR$). In particular the convex cones generated by the pseudocodewords of $T$ and $T'$ are different (the shape of such cones strongly influences the behavior of iterative/LP decoding).

This shows that it is possible to have two different (i.e. nonisomorphic) but structurally isomorphic   (minimal) linear trellises for the same code yielding different pseudocodewords. In other words, a rearrangement of edge-labels that preserves the represented code may still change the yielded pseudocodewords. Such phenomenon (whose discovery we actually had announced first in \cite{CB}) was never observed before in the literature, 
and it implies that in order to sort all the (minimal) linear trellises for a given linear code with respect to pseudocodewords it is not sufficient to do a classification of such trellises up to structural isomorphism. 
We need instead a complete classification. For minimal linear trellises this completed classification can be feasibly carried out as described in the paragraph before Example \ref{ex51}, which is the result of joining our work with the one of Koetter/Vardy \cite{KV}. 
 By carrying out this classification we can thus  find the trellises which at the same time achieve the lowest decoding complexity and have the best behavior and performance for iterative/LP decoding.

\begin{rmk}\label{rmk32}
One can easily check that if a linear trellis  is optimized with respect to iterative/LP decoding (in the sense that it yields the fewest possible pseudocodewords) then it must be {one-to-one}. So,
for such decoding purposes, 
even if we want to explore nonminimal linear trellis representations we can still restrict ourselves to the classification of those that are one-to-one.
\end{rmk}

%%%%%%%%%%%%%%%%%%%%%%%%%%%%%%%%%%%%%%%%%%%%%%%%%%%%

\section{Further applications to quasi-cyclic and nonreduced linear trellises}

\subsection{Quasi-cyclic factorizations and isomorphisms of quasi-cyclic linear trellises}

 Let $m\geq1$ divide  the length of $T$.  
 We say that $T$ is \textit{$m$-quasi-cyclic}\index{trellis!quasi-cyclic} if it is isomorphic to a trellis $T'$ satisfying 
 \begin{align*}
 V_{i}(T')&=V_{i+m}(T')\\
 E_{i}(T')&=E_{i+m}(T')
 \end{align*} 
 for all $i\in\ZZ_{n}$. 
Note that if $T$ is linear then the isomorphic trellis $T'$ is also linear, and they are linearly isomorphic. 

Given a trellis $T$ obviously the $i$-cover $T^{i}$ is an $i$-quasi-cyclic trellis, since by its very definition $T^{i}$ satisfies the above equalities. Another natural way to construct quasi-cyclic trellises is given as follows: given a trellis $T$ of length $n$ and an $m$ dividing $n$ then it is easily proven that $$T\otimes \sigma^{m}(T)\otimes\sigma^{2m}(T)\otimes\ldots\otimes \sigma ^{(\frac{n}{m}-1)m}(T)$$ is $m$-quasi-cyclic. 
For example, if $T=101|(0,2)$ and $m=1$, by definition of trellis product we get that $T\otimes\sigma(T)\otimes\sigma^{2}(T)$ is precisely equal to
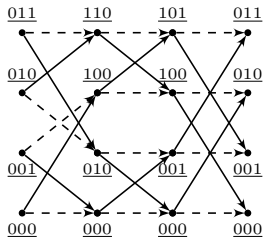
\begin{center}
 \begin{tikzpicture}[yscale=.8,>=latex',shorten >=.9pt, shorten <=1.4pt, line width=.6pt]
\begin{scope}
[xshift=0cm]
  \tikzstyle{every node}=[draw,circle,fill=black,minimum size=2pt,
                        inner sep=0pt]                    
\foreach \x in {0,1}{
\node at (\x,0) {};};
\foreach \x in {0,1}{
\node at (\x,1) {};};
\foreach \x in {0,1}{
\node at (\x,2) {};};
\foreach \x in {0,1}{
\node at (\x,3) {};};

\tikzstyle{every node}=[]

\draw [->,black,dashed] (0,0) -- (1,0);
\draw [->,black] (0,0) -- (1,2);

\draw [->,black,dashed] (0,3) -- (1,3);
\draw [->,black] (0,3) -- (1,1);

\draw [->,black] (0,1) -- (1,0);
\draw [->,black,dashed] (0,1) -- (1,2);

\draw [->,black] (0,2) -- (1,3);
\draw [->,black,dashed] (0,2) -- (1,1);

\foreach \x in {0}{
\draw (\x,0) node [below] {\tiny ${\uu{000}}$};
\draw (\x,1) node [below] {\tiny ${\uu{001}}$};
\draw (\x,2) node [above] {\tiny ${\uu{010}}$};
\draw (\x,3) node [above] {\tiny ${\uu{011}}$};
};
\foreach \x in {1}{
\draw (\x,0) node [below] {\tiny ${\uu{000}}$};
\draw (\x,1) node [below] {\tiny ${\uu{010}}$};
\draw (\x,2) node [above] {\tiny ${\uu{100}}$};
\draw (\x,3) node [above] {\tiny ${\uu{110}}$};
};
%\draw (-.5,1.5) node {$=$};
%\draw (3.5,1.5) node {$=T$};
\end{scope}
\begin{scope}
[xshift=1cm]
  \tikzstyle{every node}=[draw,circle,fill=black,minimum size=2pt,
                        inner sep=0pt]                    
\foreach \x in {0,1}{
\node at (\x,0) {};};
\foreach \x in {0,1}{
\node at (\x,1) {};};
\foreach \x in {0,1}{
\node at (\x,2) {};};
\foreach \x in {0,1}{
\node at (\x,3) {};};

\tikzstyle{every node}=[]
\draw [->,black,dashed] (0,0) -- (1,0);

\draw [->,black,dashed] (0,3) -- (1,3);

\draw [->,black,dashed] (0,1) -- (1,1);
\draw [->,black] (0,1) -- (1,0);
\draw [->,black] (0,0) -- (1,1);

\draw [->,black,dashed] (0,2) -- (1,2);
\draw [->,black] (0,2) -- (1,3);
\draw [->,black] (0,3) -- (1,2);

\foreach \x in {1}{
\draw (\x,0) node [below] {\tiny ${\uu{000}}$};
\draw (\x,1) node [below] {\tiny ${\uu{001}}$};
\draw (\x,2) node [above] {\tiny ${\uu{100}}$};
\draw (\x,3) node [above] {\tiny ${\uu{101}}$};
};
\end{scope}

\begin{scope}
[xshift=2cm]
  \tikzstyle{every node}=[draw,circle,fill=black,minimum size=2pt,
                        inner sep=0pt]                    
\foreach \x in {0,1}{
\node at (\x,0) {};};
\foreach \x in {0,1}{
\node at (\x,1) {};};
\foreach \x in {0,1}{
\node at (\x,2) {};};
\foreach \x in {0,1}{
\node at (\x,3) {};};

\tikzstyle{every node}=[]
\draw [->,black,dashed] (0,0) -- (1,0);
\draw [->,black] (0,0) -- (1,2);

\draw [->,black,dashed] (0,3) -- (1,3);
\draw [->,black] (0,3) -- (1,1);

\draw [->,black] (0,1) -- (1,3);
\draw [->,black,dashed] (0,1) -- (1,1);

\draw [->,black] (0,2) -- (1,0);
\draw [->,black,dashed] (0,2) -- (1,2);

\foreach \x in {1}{
\draw (\x,0) node [below] {\tiny ${\uu{000}}$};
\draw (\x,1) node [below] {\tiny ${\uu{001}}$};
\draw (\x,2) node [above] {\tiny ${\uu{010}}$};
\draw (\x,3) node [above] {\tiny ${\uu{011}}$};
};\end{scope}

 \end{tikzpicture}
  \end{center}
which is isomorphic to (just swap the vertices labeled $\underline{100}$ and $\underline{001}$ at time index $2$, and then relabel all the vertices)
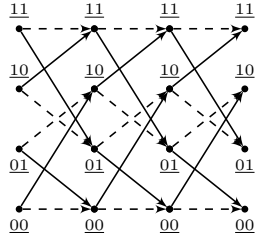
\begin{center}
 \begin{tikzpicture}[yscale=.8,>=latex',shorten >=.9pt, shorten <=1.4pt, line width=.6pt]
\begin{scope}
[xshift=0cm]
  \tikzstyle{every node}=[draw,circle,fill=black,minimum size=2pt,
                        inner sep=0pt]                    
\foreach \x in {0,1}{
\node at (\x,0) {};};
\foreach \x in {0,1}{
\node at (\x,1) {};};
\foreach \x in {0,1}{
\node at (\x,2) {};};
\foreach \x in {0,1}{
\node at (\x,3) {};};

\tikzstyle{every node}=[]

\draw [->,black,dashed] (0,0) -- (1,0);
\draw [->,black] (0,0) -- (1,2);

\draw [->,black,dashed] (0,3) -- (1,3);
\draw [->,black] (0,3) -- (1,1);

\draw [->,black] (0,1) -- (1,0);
\draw [->,black,dashed] (0,1) -- (1,2);

\draw [->,black] (0,2) -- (1,3);
\draw [->,black,dashed] (0,2) -- (1,1);

\foreach \x in {0}{
\draw (\x,0) node [below] {\tiny ${\uu{00}}$};
\draw (\x,1) node [below] {\tiny ${\uu{01}}$};
\draw (\x,2) node [above] {\tiny ${\uu{10}}$};
\draw (\x,3) node [above] {\tiny ${\uu{11}}$};
};
\foreach \x in {1}{
\draw (\x,0) node [below] {\tiny ${\uu{00}}$};
\draw (\x,1) node [below] {\tiny ${\uu{01}}$};
\draw (\x,2) node [above] {\tiny ${\uu{10}}$};
\draw (\x,3) node [above] {\tiny ${\uu{11}}$};
};
\end{scope}
\begin{scope}
[xshift=1cm]
  \tikzstyle{every node}=[draw,circle,fill=black,minimum size=2pt,
                        inner sep=0pt]                    
\foreach \x in {0,1}{
\node at (\x,0) {};};
\foreach \x in {0,1}{
\node at (\x,1) {};};
\foreach \x in {0,1}{
\node at (\x,2) {};};
\foreach \x in {0,1}{
\node at (\x,3) {};};

\tikzstyle{every node}=[]
\draw [->,black,dashed] (0,0) -- (1,0);
\draw [->,black] (0,0) -- (1,2);

\draw [->,black,dashed] (0,3) -- (1,3);
\draw [->,black] (0,3) -- (1,1);

\draw [->,black] (0,1) -- (1,0);
\draw [->,black,dashed] (0,1) -- (1,2);

\draw [->,black] (0,2) -- (1,3);
\draw [->,black,dashed] (0,2) -- (1,1);

\foreach \x in {1}{
\draw (\x,0) node [below] {\tiny ${\uu{00}}$};
\draw (\x,1) node [below] {\tiny ${\uu{01}}$};
\draw (\x,2) node [above] {\tiny ${\uu{10}}$};
\draw (\x,3) node [above] {\tiny ${\uu{11}}$};
};
\end{scope}

\begin{scope}
[xshift=2cm]
  \tikzstyle{every node}=[draw,circle,fill=black,minimum size=2pt,
                        inner sep=0pt]                    
\foreach \x in {0,1}{
\node at (\x,0) {};};
\foreach \x in {0,1}{
\node at (\x,1) {};};
\foreach \x in {0,1}{
\node at (\x,2) {};};
\foreach \x in {0,1}{
\node at (\x,3) {};};

\tikzstyle{every node}=[]
\draw [->,black,dashed] (0,0) -- (1,0);
\draw [->,black] (0,0) -- (1,2);

\draw [->,black,dashed] (0,3) -- (1,3);
\draw [->,black] (0,3) -- (1,1);

\draw [->,black] (0,1) -- (1,0);
\draw [->,black,dashed] (0,1) -- (1,2);

\draw [->,black] (0,2) -- (1,3);
\draw [->,black,dashed] (0,2) -- (1,1);

\foreach \x in {1}{
\draw (\x,0) node [below] {\tiny ${\uu{00}}$};
\draw (\x,1) node [below] {\tiny ${\uu{01}}$};
\draw (\x,2) node [above] {\tiny ${\uu{10}}$};
\draw (\x,3) node [above] {\tiny ${\uu{11}}$};
};\end{scope}

 \end{tikzpicture}
  \end{center}
 Similarly, the Golay trellis of Example \ref{golay2} (which is depicted in \cite{CFV}) is $3$-quasi-cyclic.

Now, we have the following theorem: 

\begin{thm}\label{qthm1}
Let $T$ be a connected, reduced, linear trellis of length $n$ such that $V_{i}(T)=V_{i+m}(T)$ and $E_{i}(T)=E_{i+m}(T)$ for all $i\in\ZZ_{n}$, and some $m$ dividing $n$. Consider the shift map of cycles 
$$\beta:=\sigma^{m}:\SSS(T)\rightarrow \SSS(\sigma^{m}(T))=\SSS(T)$$
Let $\BB$ be a product basis of $T$ and let $\BB_{[0,m)}$ be the subset of those $\la\in\BB$ whose span starting point $a$ satisfies $ 0\leq a <m$. Then $$\sqcup_{i=0}^{(n/m)-1}\beta^{i}(\BB_{[0,m)})$$ is a product basis of $T$.

\end{thm}
\begin{proof} 
Note that $\beta$ is a linear isomorphism of $\SSS(T)$ with itself.
We clearly have that $\beta(\SSS_{(a+m,l)}(T))=\sal(T)$ for all spans $(a,l)$. So,  a subset $S\subset\SSS_{(a+m,l)}(T)$ is a lifting of a basis of $\SSS_{(a+m,l)}(T)/\SSS_{<(a+m,l)}(T)$ if and only if $\beta(S)\subseteq\sal(T)$ is a lifting of a basis of $\sal(T)/\SSS_{<(a,l)}(T)$. Thus, by Observation \ref{spansub1} and Theorem \ref{thm12} it follows that the union $\sqcup_{i=0}^{(n/m)-1}\beta(\BB_{[0,m)})$ is a product basis of $$\SSS(T)=\SSS_{n-1}(T)=\sum_{a\in\ZZ_{n},l\leq n-1}\sal(T)$$
 where the first equality follows from  $T$ being connected, 
and so we are done.
\end{proof}

By combining the above theorem with the Factorization Theorem and Theorem \ref{thm13} we get also the following corollaries:

\begin{cor}\label{qcor1}
The span distribution $\mathcal{S}(T)$ of any connected, $m$-quasi-cyclic, reduced, linear trellis $T$ of length $n$ is decomposed into disjoint orbits of order $n/m$ under the action of $\sigma^{m}$.
\end{cor}

\begin{cor}\label{qcor2}
Let $\otimes_{i=1}^{r} T_{i}$ be a product of connected elementary trellises of length $n$. If $\otimes_{i=1}^{r} T_{i}$ is $m$-quasi-cyclic then 
 $$\otimes_{i=1}^{r} T_{i}\sim\otimes_{j=0}^{(n/m)-1}(\beta^{j}(T_{i_{1}})\otimes\ldots\otimes\beta^{j}(T_{i_{s}}))$$
  where  $T_{i_{1}},\ldots,T_{i_{s}}$ are those elementary trellises whose span starting point $a$ satisfies $0\leq a<m$ 
   and $\beta=\sigma^{m}$.  
\end{cor}

The above tells us that the above type of trellises admit a  product basis/elementary trellis factorization with a quasi-cyclic structure.

We can also deduce the following theorem which we will use in the next subsection to extend results for reduced trellises to nonreduced ones. This theorem tells us that if two aforementioned trellises  are isomorphic then they admit an isomorphism with a quasi-cyclic structure too. 

\begin{thm}\label{qthm2}
Let $T$ and $T'$ be two connected, reduced, linear trellises of same length $n$ and satisfying $V_{i}(T)=V_{i+m}(T)$, $V_{i}(T')=V_{i+m}(T')$, $E_{i}(T)=E_{i+m}(T)$,  and $E_{i}(T')=E_{i+m}(T')$ for all $i\in\ZZ_{n}$, and some $m$ dividing $n$. Assume that  $T$ and $T'$ are (linearly) isomorphic. Then there exists a linear isomorphism $f:T\rightarrow T'$ such that 
 $$\beta\circ \SSS(f)= \SSS(f)\circ\beta$$ for the  shift operator $\beta=\sigma^{m}$, i.e. $f_{i}=f_{i+m}$ for all $i$. 
\end{thm}
\begin{proof}
Let $\BB$  be a product basis  of $T$ and  $g:T\rightarrow T'$ a linear isomorphism. The induced isomorphism $\SSS(g):\SSS(T)\rightarrow\SSS(T')$ sends $\BB$ to a product basis $\BB'$ of $T'$ and satisfies $\SSS(g)(\BB_{[0,m)})=\BB'_{[0,m)}$ (we are using the notation from Theorem \ref{qthm1}). 
By  Theorem \ref{qthm1} $$\widetilde{\BB}:=\sqcup_{i=0}^{(n/m)-1}\beta^{i}(\BB_{[0,m)}), \ \widetilde{\BB'}:=\sqcup_{i=0}^{(n/m)-1}\beta^{i}(\BB'_{[0,m)})$$ are product bases respectively of $T, T'$. So we can define a linear isomorphism $F:\SSS(T)\rightarrow\SSS(T')$ by the formula $$F(\beta^{i}(\la)):=\beta^{i}(\SSS(g)(\la))$$ for all $i$ and all $\la\in\BB_{[0,m)}$. 
By construction $F(\widetilde{\BB})=\widetilde{\BB'}$ and $\beta\circ F=F\circ\beta$.

 Now, since $L\circ\SSS(g)=L$, i.e. $\SSS(g)$ preserves edge-labels, it follows  that 
 \begin{gather*}L(F(\beta^{i}(\la)))=L(\beta^{i}(\SSS(g)(\la)))=\beta^{i}(L(\SSS(g)(\la)))=\\=\beta^{i}(L(\la))=L(\beta^{i}(\la))
\end{gather*}
  for all $i$ and all $\la\in\BB_{[0,m)}$, and so $F$ preserves edge-labels too. Also, being $g$ a linear isomorphism we have  $[\SSS(g)(\la)]=[\la]$ for all $\la\in\BB$, and so 
  \begin{gather*}[F(\beta^{i}(\la))]=[\beta^{i}(\SSS(g)(\la))]=\beta^{i}([\SSS(g)(\la)])=\\=\beta^{i}([\la])=[\beta^{i}(\la)]
  \end{gather*}
   for all $i$ and all $\la\in\BB_{[0,m)}$, i.e. $[F(\widetilde{\la})]=[\widetilde{\la}]$ for all $\widetilde{\la}\in\widetilde{\BB}$, from which also follows that $$[F^{-1}(\widetilde{\la})]=[F(F^{-1}(\widetilde{\la}))]=[\widetilde{\la}]$$ for all $\widetilde{\la}\in\widetilde{\BB'}$. Thus $F(\sal(T))=\sal(T')$ for all $(a,l)$. We conclude that $F=\SSS(f)$ for some linear isomorphism $f:T\rightarrow T'$ (see Subsection \ref{labcodemapspap}). Finally, $\beta\circ F=F\circ\beta$
 means that $\beta\circ \SSS(f)= \SSS(f)\circ\beta$.
\end{proof}

\subsection{Extending results to nonreduced trellises}

We start with proving an important property which makes it possible to extend results for reduced linear trellises to nonreduced linear trellises, as we shall see.

\begin{thm}\label{covred}
Let $T$ be a linear trellis. Then there exists $i\geq1$ such that $T^{i}$ is reduced. \end{thm}
\begin{proof}
By Corollary \ref{cor44} we know that for each edge $\bm{e}$ of $T$ there exists $i\geq1$ such that $\bm{e}$ belongs to an $i$-cycle of $T$. Also, clearly, if $\bm{e}$ belongs to an $i$-cycle, then it belongs to an $ir$-cycle for all $r\geq1$. Thus, by taking a  common multiple, there exists $m\geq1$ such that each edge of $T$ belongs to some $m$-cycle. But then it is clear that $T^{m}$ is reduced (since its diagram is just $m$ concatenated copies of the diagram of $T$). 
\end{proof}

We can apply the above result to prove that connected linear trellises are determined by their covers. First an important lemma:

\begin{lem}\label{lem2pap} If $T$ is a connected linear trellis and $i\geq1$ then $T^{i}$ is also connected.
\end{lem}
\begin{proof}
Let $v\in V_{j}(T^{i})= V_{j}(T)$. Then there exists a path $\bm{p}$ in $T$ from $v$ to $0\in V_{0}(T)$. Clearly $\bm{p}$ yields a path in $T^{i}$ from $v\in V_{j}(T^{i})$ to $0\in V_{rn}(T^{i})$ for some $r\geq0$. But all the zero vertices are connected, thus $T^{i}$ is connected. 
\end{proof}

\begin{thm}\label{isthmcov}
Let $T$ and $T'$ be connected linear trellises, and let $i>1$. Then $$T\simeq T' \Longleftrightarrow T^{i}\simeq(T')^{i}$$ 
\end{thm}
\begin{proof}
Clearly $T\simeq T'\Rightarrow T^{i}\simeq(T')^{i}$.
 By the previous lemma and theorem there exist $s,s'\geq1$ such that $T^{s}$ and $(T')^{s'}$ are both reduced and connected. So, if $T^{i}\simeq(T')^{i}$  then $T^{iss'}$ and $ (T')^{iss'}$ are linearly isomorphic, connected and reduced (if a trellis $T$ is reduced then $T^{j}$ is clearly reduced too for any $j\geq1$).  Then by Theorem \ref{qthm2}  we are done.
\end{proof}

Note that linearity is necessary both for Theorem \ref{covred} and Lemma \ref{lem2pap} as one can check by easy examples. Also, the connectedness hypothesis is necessary for Theorem \ref{isthmcov}. We depict below nonisomorphic trellises $T$ and $T'$ ($T$ has $3$ connected components while $T'$ has $4$) such that $T^{2}\simeq (T')^{2}$:
     
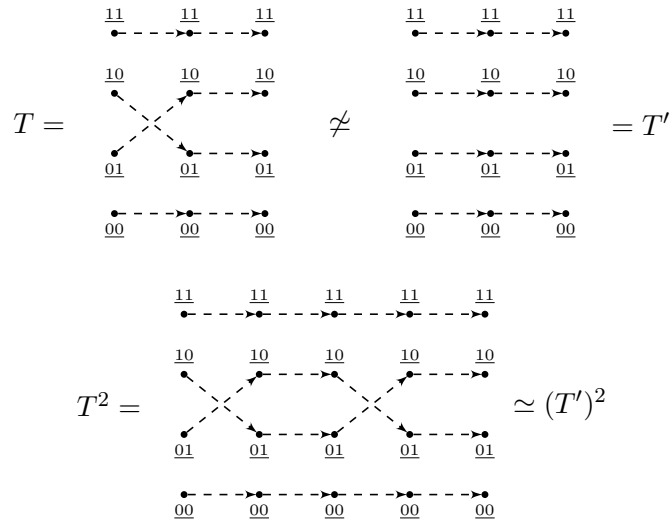
\begin{center}
 \begin{tikzpicture}[yscale=.8,>=latex',shorten >=.9pt, shorten <=1.4pt, line width=.6pt]
  \tikzstyle{every node}=[draw,circle,fill=black,minimum size=2pt,
                        inner sep=0pt]                    
\foreach \x in {0,1,2}{
\node at (\x,0) {};};
\foreach \x in {0,1,2}{
\node at (\x,1) {};};
\foreach \x in {0,1,2}{
\node at (\x,2) {};};
\foreach \x in {0,1,2}{
\node at (\x,3) {};};

\tikzstyle{every node}=[]
\draw [->,dashed] (0,0) -- (1,0);
\draw [->,dashed] (1,0) -- (2,0);

\draw [->,dashed] (0,3) -- (1,3);
\draw [->,dashed] (1,3) -- (2,3);

\draw [->,dashed] (0,1) -- (1,2);
\draw [->,dashed] (1,2) -- (2,2);

\draw [->,dashed] (0,2) -- (1,1);
\draw [->,dashed] (1,1) -- (2,1);

\foreach \x in {0,1,2}{
\draw (\x,0) node [below] {\tiny ${\uu{00}}$};
\draw (\x,1) node [below] {\tiny ${\uu{01}}$};
\draw (\x,2) node [above] {\tiny ${\uu{10}}$};
\draw (\x,3) node [above] {\tiny ${\uu{11}}$};
};

\draw (-1,1.5) node {$T=$};
\begin{scope}
[xshift=4cm]
  \tikzstyle{every node}=[draw,circle,fill=black,minimum size=2pt,
                        inner sep=0pt]                    
\foreach \x in {0,1,2}{
\node at (\x,0) {};};
\foreach \x in {0,1,2}{
\node at (\x,1) {};};
\foreach \x in {0,1,2}{
\node at (\x,2) {};};
\foreach \x in {0,1,2}{
\node at (\x,3) {};};

\tikzstyle{every node}=[]
\draw [->,dashed] (0,0) -- (1,0);
\draw [->,dashed] (1,0) -- (2,0);

\draw [->,dashed] (0,3) -- (1,3);
\draw [->,dashed] (1,3) -- (2,3);

\draw [->,dashed] (0,1) -- (1,1);
\draw [->,dashed] (1,2) -- (2,2);

\draw [->,dashed] (0,2) -- (1,2);
\draw [->,dashed] (1,1) -- (2,1);

\foreach \x in {0,1,2}{
\draw (\x,0) node [below] {\tiny ${\uu{00}}$};
\draw (\x,1) node [below] {\tiny ${\uu{01}}$};
\draw (\x,2) node [above] {\tiny ${\uu{10}}$};
\draw (\x,3) node [above] {\tiny ${\uu{11}}$};
};

\draw (-1,1.5) node {$\not\simeq $};
\draw (3,1.5) node {$=T'$};
\end{scope}
      \end{tikzpicture}
    \end{center}
\begin{center}
 \begin{tikzpicture}[yscale=.8,>=latex',shorten >=.9pt, shorten <=1.4pt, line width=.6pt]
  \tikzstyle{every node}=[draw,circle,fill=black,minimum size=2pt,
                        inner sep=0pt]                    
\foreach \x in {0,1,2,3,4}{
\node at (\x,0) {};};
\foreach \x in {0,1,2,3,4}{
\node at (\x,1) {};};
\foreach \x in {0,1,2,3,4}{
\node at (\x,2) {};};
\foreach \x in {0,1,2,3,4}{
\node at (\x,3) {};};

\tikzstyle{every node}=[]
\draw [->,dashed] (0,0) -- (1,0);
\draw [->,dashed] (1,0) -- (2,0);
\draw [->,dashed] (2,0) -- (3,0);
\draw [->,dashed] (3,0) -- (4,0);

\draw [->,dashed] (0,3) -- (1,3);
\draw [->,dashed] (1,3) -- (2,3);
\draw [->,dashed] (2,3) -- (3,3);
\draw [->,dashed] (3,3) -- (4,3);

\draw [->,dashed] (0,1) -- (1,2);
\draw [->,dashed] (1,2) -- (2,2);
\draw [->,dashed] (2,2) -- (3,1);
\draw [->,dashed] (3,1) -- (4,1);

\draw [->,dashed] (0,2) -- (1,1);
\draw [->,dashed] (1,1) -- (2,1);
\draw [->,dashed] (2,1) -- (3,2);
\draw [->,dashed] (3,2) -- (4,2);

\foreach \x in {0,1,2,3,4}{
\draw (\x,0) node [below] {\tiny ${\uu{00}}$};
\draw (\x,1) node [below] {\tiny ${\uu{01}}$};
\draw (\x,2) node [above] {\tiny ${\uu{10}}$};
\draw (\x,3) node [above] {\tiny ${\uu{11}}$};
};

\draw (-1,1.5) node {$T^{2}=$};
\draw (5,1.5) node {$\simeq (T')^{2}$};
%\begin{scope}
%[yshift=-4cm]
%  \tikzstyle{every node}=[draw,circle,fill=black,minimum size=2pt,
%                        inner sep=0pt]                    
%\foreach \x in {0,1,2,3,4}{
%\node at (\x,0) {};};
%\foreach \x in {0,1,2,3,4}{
%\node at (\x,1) {};};
%\foreach \x in {0,1,2,3,4}{
%\node at (\x,2) {};};
%\foreach \x in {0,1,2,3,4}{
%\node at (\x,3) {};};
%
%\tikzstyle{every node}=[]
%\draw [->,dashed] (0,0) -- (1,0);
%\draw [->,dashed] (1,0) -- (2,0);
%\draw [->,dashed] (2,0) -- (3,0);
%\draw [->,dashed] (3,0) -- (4,0);
%
%\draw [->,dashed] (0,3) -- (1,3);
%\draw [->,dashed] (1,3) -- (2,3);
%\draw [->,dashed] (2,3) -- (3,3);
%\draw [->,dashed] (3,3) -- (4,3);
%
%\draw [->,dashed] (0,1) -- (1,1);
%\draw [->,dashed] (1,2) -- (2,2);
%\draw [->,dashed] (2,1) -- (3,1);
%\draw [->,dashed] (3,1) -- (4,1);
%
%\draw [->,dashed] (0,2) -- (1,2);
%\draw [->,dashed] (1,1) -- (2,1);
%\draw [->,dashed] (2,2) -- (3,2);
%\draw [->,dashed] (3,2) -- (4,2);
%
%
%\foreach \x in {0,1,2,3,4}{
%\draw (\x,0) node [below] {\tiny ${\uu{00}}$};
%\draw (\x,1) node [below] {\tiny ${\uu{01}}$};
%\draw (\x,2) node [above] {\tiny ${\uu{10}}$};
%\draw (\x,3) node [above] {\tiny ${\uu{11}}$};
%};
%
%\draw (-1,1.5) node {$(T')^{2}=$};
%\end{scope}
      \end{tikzpicture}
    \end{center}

 An immediate consequence of the above two theorems  
 is that we can extend the important Theorems \ref{isolin} and \ref{linstr} to the nonreduced case too: 

\begin{cor}\label{qcor4}
Two connected, nonreduced, linear trellises are isomorphic if and only they are linearly isomorphic. 
\end{cor}
\begin{cor}\label{qcor5}
The linear structure of a connected, nonreduced, linear trellis is essentially unique (as in theorem \ref{linstr}). 
\end{cor}

We can also extend Theorem \ref{thm53} to nonreduced trellises, i.e. prove that the sequence $\{C^{i}(T)\}_{i=1}^{\infty}$ determines trellises (in the mentioned class) even in the nonreduced case. 
We will need the following theorem.

\begin{thm}\label{covred2} Let $T$ be a connected linear  trellis. If $T^{i}$  is  reduced then $T^{i+1}$ is reduced.  
\end{thm}
\begin{proof}
By Lemma \ref{lem2pap} $T^{i}$ and $T^{i+1}$ are connected. By Theorem \ref{thm45} we know that a connected linear  trellis of length $n$ is reduced if and only if each vertex can be connected in both directions to a zero vertex by paths of length $n-1$. Thus for every vertex $v\in T^{i}$  there exist paths $v\rightarrow 0$ and $0 \rightarrow v$ of length $in-1$, where $n$ is the length of $T$. But then it is clear that  the same holds for $T^{i+1}$. So $T^{i+1}$ is reduced too.
\end{proof}

The connectedness hypothesis is necessary. For example, here we have a linear disconnected trellis $T$ such that $T^{2}$ is reduced while $T^{3}$ is not:
  
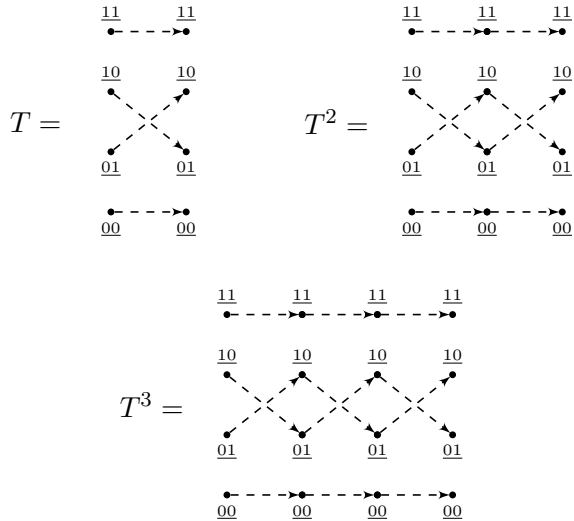
\begin{center}
 \begin{tikzpicture}[yscale=.8,>=latex',shorten >=.9pt, shorten <=1.4pt, line width=.6pt]
  \tikzstyle{every node}=[draw,circle,fill=black,minimum size=2pt,
                        inner sep=0pt]                    
\foreach \x in {0,1}{
\node at (\x,0) {};};
\foreach \x in {0,1}{
\node at (\x,1) {};};
\foreach \x in {0,1}{
\node at (\x,2) {};};
\foreach \x in {0,1}{
\node at (\x,3) {};};

\tikzstyle{every node}=[]
\draw [->,dashed] (0,0) -- (1,0);

\draw [->,dashed] (0,3) -- (1,3);

\draw [->,dashed] (0,1) -- (1,2);

\draw [->,dashed] (0,2) -- (1,1);

\foreach \x in {0,1}{
\draw (\x,0) node [below] {\tiny ${\uu{00}}$};
\draw (\x,1) node [below] {\tiny ${\uu{01}}$};
\draw (\x,2) node [above] {\tiny ${\uu{10}}$};
\draw (\x,3) node [above] {\tiny ${\uu{11}}$};
};

\draw (-1,1.5) node {$T=$};

\begin{scope}
[xshift=4cm]
  \tikzstyle{every node}=[draw,circle,fill=black,minimum size=2pt,
                        inner sep=0pt]                    
\foreach \x in {0,1}{
\node at (\x,0) {};};
\foreach \x in {0,1}{
\node at (\x,1) {};};
\foreach \x in {0,1}{
\node at (\x,2) {};};
\foreach \x in {0,1}{
\node at (\x,3) {};};

\tikzstyle{every node}=[]
\draw [->,dashed] (0,0) -- (1,0);

\draw [->,dashed] (0,3) -- (1,3);

\draw [->,dashed] (0,1) -- (1,2);

\draw [->,dashed] (0,2) -- (1,1);

\foreach \x in {0,1}{
\draw (\x,0) node [below] {\tiny ${\uu{00}}$};
\draw (\x,1) node [below] {\tiny ${\uu{01}}$};
\draw (\x,2) node [above] {\tiny ${\uu{10}}$};
\draw (\x,3) node [above] {\tiny ${\uu{11}}$};
};
\draw (-1,1.5) node {$T^{2}=$};
\end{scope}
\begin{scope}
[xshift=5cm]
  \tikzstyle{every node}=[draw,circle,fill=black,minimum size=2pt,
                        inner sep=0pt]                    
\foreach \x in {0,1}{
\node at (\x,0) {};};
\foreach \x in {0,1}{
\node at (\x,1) {};};
\foreach \x in {0,1}{
\node at (\x,2) {};};
\foreach \x in {0,1}{
\node at (\x,3) {};};

\tikzstyle{every node}=[]
\draw [->,dashed] (0,0) -- (1,0);

\draw [->,dashed] (0,3) -- (1,3);

\draw [->,dashed] (0,1) -- (1,2);

\draw [->,dashed] (0,2) -- (1,1);

\foreach \x in {1}{
\draw (\x,0) node [below] {\tiny ${\uu{00}}$};
\draw (\x,1) node [below] {\tiny ${\uu{01}}$};
\draw (\x,2) node [above] {\tiny ${\uu{10}}$};
\draw (\x,3) node [above] {\tiny ${\uu{11}}$};
};
\end{scope}
      \end{tikzpicture}
    \end{center}
  \begin{center}
 \begin{tikzpicture}[yscale=.8,>=latex',shorten >=.9pt, shorten <=1.4pt, line width=.6pt]
\begin{scope}
  \tikzstyle{every node}=[draw,circle,fill=black,minimum size=2pt,
                        inner sep=0pt]                    
\foreach \x in {0,1}{
\node at (\x,0) {};};
\foreach \x in {0,1}{
\node at (\x,1) {};};
\foreach \x in {0,1}{
\node at (\x,2) {};};
\foreach \x in {0,1}{
\node at (\x,3) {};};

\tikzstyle{every node}=[]
\draw [->,dashed] (0,0) -- (1,0);

\draw [->,dashed] (0,3) -- (1,3);

\draw [->,dashed] (0,1) -- (1,2);

\draw [->,dashed] (0,2) -- (1,1);

\foreach \x in {0,1}{
\draw (\x,0) node [below] {\tiny ${\uu{00}}$};
\draw (\x,1) node [below] {\tiny ${\uu{01}}$};
\draw (\x,2) node [above] {\tiny ${\uu{10}}$};
\draw (\x,3) node [above] {\tiny ${\uu{11}}$};
};
\draw (-1,1.5) node {$T^{3}=$};
\end{scope}
\begin{scope}
[xshift=1cm]
  \tikzstyle{every node}=[draw,circle,fill=black,minimum size=2pt,
                        inner sep=0pt]                    
\foreach \x in {0,1}{
\node at (\x,0) {};};
\foreach \x in {0,1}{
\node at (\x,1) {};};
\foreach \x in {0,1}{
\node at (\x,2) {};};
\foreach \x in {0,1}{
\node at (\x,3) {};};

\tikzstyle{every node}=[]
\draw [->,dashed] (0,0) -- (1,0);

\draw [->,dashed] (0,3) -- (1,3);

\draw [->,dashed] (0,1) -- (1,2);

\draw [->,dashed] (0,2) -- (1,1);

\foreach \x in {1}{
\draw (\x,0) node [below] {\tiny ${\uu{00}}$};
\draw (\x,1) node [below] {\tiny ${\uu{01}}$};
\draw (\x,2) node [above] {\tiny ${\uu{10}}$};
\draw (\x,3) node [above] {\tiny ${\uu{11}}$};
};
\end{scope}
\begin{scope}
[xshift=2cm]
  \tikzstyle{every node}=[draw,circle,fill=black,minimum size=2pt,
                        inner sep=0pt]                    
\foreach \x in {0,1}{
\node at (\x,0) {};};
\foreach \x in {0,1}{
\node at (\x,1) {};};
\foreach \x in {0,1}{
\node at (\x,2) {};};
\foreach \x in {0,1}{
\node at (\x,3) {};};

\tikzstyle{every node}=[]
\draw [->,dashed] (0,0) -- (1,0);

\draw [->,dashed] (0,3) -- (1,3);

\draw [->,dashed] (0,1) -- (1,2);

\draw [->,dashed] (0,2) -- (1,1);

\foreach \x in {1}{
\draw (\x,0) node [below] {\tiny ${\uu{00}}$};
\draw (\x,1) node [below] {\tiny ${\uu{01}}$};
\draw (\x,2) node [above] {\tiny ${\uu{10}}$};
\draw (\x,3) node [above] {\tiny ${\uu{11}}$};
};
\end{scope}
      \end{tikzpicture}
    \end{center}

Now we can prove our extension of Theorem \ref{thm53}.

\begin{thm}
Let $T$ and $T'$ be fragment, one-to-one, connected linear trellises. Let $i,j\geq1$ such that $T^{i}$ and $(T')^{j}$ are reduced (by Theorem \ref{covred} such $i,j$ exist). Assume that $C^{s}(T)=C^{s}(T')$ for some $s>\max\{i,j\}$. Then $T\simeq T'$.
\end{thm}
\begin{proof}
Let $h:=\max\{i,j\}$ and $s>h$ such that $C^{s}(T)=C^{s}(T')$. By Theorem \ref{covred2} we have that $T^{h}$ and $(T')^{h}$ are both reduced. But then an immediate adaption of the arguments used in the proof of Theorem \ref{thm53} yields that $T^{h}\simeq(T')^{h}$, which by Theorem \ref{isthmcov} implies that $T\simeq T'$. 
\end{proof}

We conclude this subsection by observing that a motivation for studying nonreduced linear trellises comes from the fact that these naturally arise by taking duals of reduced linear trellises or wrapped fragments of quasi-cyclic trellises (i.e. cutting $T$ at time indices $i,j$ such that $V_{i}(T)=V_{j}(T)$ and wrapping). For example the nonreduced linear trellis of length $8$ depicted in \cite{CFV} which represents the $[8,4,4]$ Hamming code is a fragment of the Golay trellis $T_{\mathcal{G}}$ (see Example \ref{golay2}) from the same paper.

%%%%%%%%%%%%%%%%%%%%%%%%%%%%%%%%%%%%%%%%%%%%%%%%%%%

%\appendices
\appendix
\label{trellisconnect}
\section{Connectivity of linear trellises}\label{trellisconnect}

We prove in this appendix some fundamental results on connectivity of linear trellises which  have not appeared before in the literature. We will also make use of them in the paper.  
We use the notation $v\rightarrow w$\nomenclature[vw]{$v\rightarrow w$}{path from $v$ to $w$} to mean a path from $v$ to $w$.

\begin{rmk}
Connectedness is closely related to the notion of controllability in systems theory. See \cite{FoGl,FoGl2} for   relations between connectedness and other trellis properties from the  ``controllability'' point of view.  
\end{rmk}

We start with proving that for linear trellises there is no distinction between being connected by directed paths and being connected by undirected paths.

\begin{thm}\label{thm43}
Let $T$ be a linear trellis. Let $\bm{e}$ be an edge in $T$ from $v$ to $w$. Then there exists a path in $T$ from $w$ to $v$.
\end{thm}
\begin{proof}
Put $v_{0}:=v$, $v_{1}:=w$.
Since all our trellises are trim the outdegree and indegree of each vertex of $T$ are positive, and so we can construct a doubly infinite sequence of vertices $$\ldots v_{-2},v_{-1}, v_{0},v_{1},v_{2},\ldots$$ such that for each $i\in\ZZ$ there exists an edge in $T$ from $v_{i}$ to $v_{i+1}$. In particular, for all $i<j$ there exists a path from $v_{i}$ to $v_{j}$ of length $j-i$.  Now, $T$ has finitely many vertices so we can find $i,j>0$ such that $v_{-i}=v_{-i-n}$ and $v_{j}=v_{j+n}$, where $n$ is the length of $T$. Thus we get closed paths $v_{-i}\overset{\bm{p}}{\rightarrow}v_{-i}$, $v_{j}\overset{\bm{q}}{\rightarrow}v_{j}$ of length $n$. Take $i\leq i'<i+n$ such that $-i'\equiv j \mod n$.  Changing the starting point, we can assume that $\bm{p}$ starts (and so ends)  at $v_{-i'}$. We  also know that there exists a path $v_{-i'}\overset{\bm{s}}{\rightarrow}v_{j}$ of length $j+i'$. 
Note that $v_{-i'}$ and $v_{j}$ belong to the same vertex set $V_{h}(T)$ for some $h\in\ZZ_{n}$, so $i'+j=rn$, for some $r\geq1$. Let $v_{-i'}\overset{\bm{p}^{r}}{\rightarrow}v_{-i'}$ be the closed path of length $rn$ given by cycling $r$ times around $\bm{p}$. Define $v_{j}\overset{\bm{q}^{r}}{\rightarrow}v_{j}$ similarly.
By linearity we get a path $$\bm{s}'=\bm{p}^{r}+\bm{q}^{r}-\bm{s}$$
 of length $rn$ from $v_{j}=v_{-i'}+v_{j}-v_{-i'}$ to $v_{-i'}=v_{-i'}+v_{j}-v_{j}$. Since  the vertices $v=v_{0}$ and $w=v_{1}$ belong by construction to the path $v_{-i'}\overset{\bm{s}}{\rightarrow}v_{j}$, we can use it in conjunction with $v_{j}\overset{\bm{s}'}{\rightarrow}v_{-i'}$ to reach $v$ from $w$, and so our proof is concluded.
\end{proof}

\begin{cor}\label{cor43} 
A linear trellis is connected  if and only if it connected as an undirected graph.
\end{cor}

\begin{cor}\label{cor44}
Each vertex and edge of a linear trellis belongs to some closed path.
\end{cor}

\begin{ex}\label{ex41}
The following is a linear, connected, and nonreduced trellis. 
  
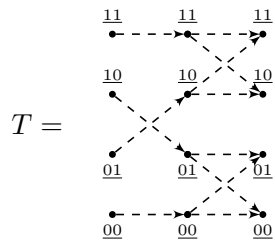
\begin{center}
 \begin{tikzpicture}[yscale=.8,>=latex',shorten >=.9pt, shorten <=1.4pt, line width=.6pt]
  \tikzstyle{every node}=[draw,circle,fill=black,minimum size=2pt,
                        inner sep=0pt]                    
\foreach \x in {0,1}{
\node at (\x,0) {};};
\foreach \x in {0,1}{
\node at (\x,1) {};};
\foreach \x in {0,1}{
\node at (\x,2) {};};
\foreach \x in {0,1}{
\node at (\x,3) {};};

\tikzstyle{every node}=[]
\draw [->,dashed] (0,0) -- (1,0);

\draw [->,dashed] (0,3) -- (1,3);

\draw [->,dashed] (0,1) -- (1,2);

\draw [->,dashed] (0,2) -- (1,1);

\foreach \x in {0,1}{
\draw (\x,0) node [below] {\tiny ${\uu{00}}$};
\draw (\x,1) node [below] {\tiny ${\uu{01}}$};
\draw (\x,2) node [above] {\tiny ${\uu{10}}$};
\draw (\x,3) node [above] {\tiny ${\uu{11}}$};
};

\draw (-1,1.5) node {$T=$};

\begin{scope}
[xshift=1cm]
  \tikzstyle{every node}=[draw,circle,fill=black,minimum size=2pt,
                        inner sep=0pt]                    
\foreach \x in {0,1}{
\node at (\x,0) {};};
\foreach \x in {0,1}{
\node at (\x,1) {};};
\foreach \x in {0,1}{
\node at (\x,2) {};};
\foreach \x in {0,1}{
\node at (\x,3) {};};

\tikzstyle{every node}=[]
\draw [->,dashed] (0,0) -- (1,0);

\draw [->,dashed] (0,3) -- (1,3);

\draw [->,dashed] (0,1) -- (1,1);
\draw [->,dashed] (0,1) -- (1,0);
\draw [->,dashed] (0,0) -- (1,1);

\draw [->,dashed] (0,2) -- (1,2);
\draw [->,dashed] (0,2) -- (1,3);
\draw [->,dashed] (0,3) -- (1,2);

\foreach \x in {1}{
\draw (\x,0) node [below] {\tiny ${\uu{00}}$};
\draw (\x,1) node [below] {\tiny ${\uu{01}}$};
\draw (\x,2) node [above] {\tiny ${\uu{10}}$};
\draw (\x,3) node [above] {\tiny ${\uu{11}}$};
};
\end{scope}

 \end{tikzpicture}
  \end{center}
Note that some edges of $T$ belong to cycles, i.e. closed paths of length $2$, while other edges belong only to closed paths of length $4$ or even length $6$. For example, this happens respectively for the edges $\uu{00}0\uu{00}\in E_{0}(T)$, $\uu{10}0\uu{01}\in E_{0}(T)$, $\uu{00}0\uu{01}\in E_{1}(T)$.
 \end{ex}
The above results do not hold for nonlinear trellises as one can easily check. 

The following theorem is another important consequence of linearity. It was also observed by Heide Gluesing-Luerssen (private communication). Note also that Lemma 6.8 of \cite{KV2} can be obtained as a special case of it. 
\begin{thm}\label{thm44}
Let $T$ be a linear trellis of length $n$. Suppose  $T$ is almost reduced. Let $v,w\in V_{i}(T)$ for some $i\in\ZZ_{n}$, and suppose there exists a path in $T$ from $v$ to $w$. Then there exists a path from $v$ to $w$ of length $n$.
\end{thm} 
 \begin{proof} Assume we have a path $v\overset{\bm{p}}{\rightarrow}w$ of length $rn$, $r>1$. Then there are vertices $v_{j}\in V_{i}(T)$ for $j=1,\ldots, r$, with $v_{0}=v$, $v_{r}=w$, and paths $v_{j}\overset{\bm{p}^{j}}{\rightarrow}v_{j+1}$ of length $n$  for $j=0,\ldots,r-1$. By hypothesis we also have closed paths $v_{j}\overset{\bm{q}^{j}}{\rightarrow}v_{j}$  of length $n$ for $j=1,\ldots,r$. By linearity we then get a path $\sum_{i=0}^{r-1}\bm{p}^{j}-\sum_{j=1}^{r}\bm{q}^{j-1}$ of length $n$ from 
 $$v=v_{0}=\sum_{j=0}^{r-1}v_{j}-\sum_{j=1}^{r-1}v_{j}$$ 
 to
 $$w=v_{r}=\sum_{j=0}^{r-1}v_{j+1}-\sum_{j=1}^{r-1}v_{j}$$.
 \end{proof}
 
 \begin{cor}\label{cor45} 
 If $T$ is a connected, almost reduced, linear trellis of length $n$ then for each pair of vertices  $v,w\in V_{0}(T)$ there exists paths $v \rightarrow w$ and $w \rightarrow v$ of length $n$.
 \end{cor}
 
 \begin{cor}\label{cor46}
 Let $T$ be a  reduced linear trellis of length $n$. Assume $v$ is connected to some (and thus each) zero vertex. Then there exist paths $v{}{\rightarrow}0$ and $0{}{\rightarrow}v$ of length $n-1$.
 \end{cor}
 \begin{proof}
 By inverting the direction of all edges it is sufficient to prove that there exists a path $v{}{\rightarrow}0$ of length $n-1$. Now, since all zero vertices are connected,  from the assumption it follows that there exists a path $v{}{\rightarrow}0$ of length $rn$ for some $r\geq1$. Thus, from Theorem \ref{thm44} there exists a path $v\overset{\bm{p}}{\rightarrow}0$ of length $n$, i.e. a path $v\overset{\bm{p}'}{\rightarrow}w$ of length $n-1$ and a path $w\overset{\bm{p}''}{\rightarrow}0$ of length $1$, for some vertex $w$. Since $T$ is reduced there must exist a path $0\overset{\bm{q}}{\rightarrow}w$ of length $n-1$. So $\bm{p}'-\bm{q}$ is a path of length $n-1$ from $v=v-0$ to $0=w-w$.
 \end{proof}
 
 As a consequence of the last corollary we can give an alternative characterization of reduced trellises in the connected case.
 \begin{thm}\label{thm45}
 Let $T$ be a    connected linear trellis of length $n$. Then $T$ is reduced if and only if for each vertex $v$ of $T$ there exist paths $v{}{\rightarrow}0$ and $0{}{\rightarrow}v$ of length $n-1$.
 \end{thm}
 \begin{proof}
 The ``only if'' part is due to Corollary \ref{cor46}. Vice versa, assume $\bm{e}=v\alpha w$ is an edge of $T$. We want to show that $\bm{e}$ belongs to a cycle of $T$, i.e. that there exists a path $w{}{\rightarrow}v$ of length $n-1$. By hypothesis we have paths $0\overset{\bm{p}}{\rightarrow}v$ and $w\overset{\bm{q}}{\rightarrow}0$ of length $n-1$. The starting time indices of $\bm{p}$ and $\bm{q}$ are equal. Thus we can add them  and get the path $\bm{p}+\bm{q}$ of length $n-1$ from $w$ to $v$.
 \end{proof}

 %%%%%%%%%%%%%%%%%%%%%%%%%%%%%%%%%%%%%%%%%%%%%%%%%%%%%
%%%%%%%%%%%%%%%%%%%%%%%%%%%%%%%%%%%%%%%%%%%%%%%%%%%%%
%%%%%%%%%%%%%%%%%%%%%%%%%%%%%%%%%%%%%%%%%%%%%%%%%%%%%
%%%%%%%%%%%%%%%%%%%%%%%%%%%%%%%%%%%%%%%%%%%%%%%%%%%%%
%%%%%%%%%%%%%%%%%%%%%%%%%%%%%%%%%%%%%%%%%%%%%%%%%%%%%
%%%%%%%%%%%%%%%%%%%%%%%%%%%%%%%%%%%%%%%%%%%%%%%%%%%%%
%%%%%%%%%%%%%%%%%%%%%%%%%%%%%%%%%%%%%%%%%%%%%%%%%%%%%
%%%%%%%%%%%%%%%%%%%%%%%%%%%%%%%%%%%%%%%%%%%%%%%%%%%%%
%%%%%%%%%%%%%%%%%%%%%%%%%%%%%%%%%%%%%%%%%%%%%%%%%%%%%
%%%%%%%%%%%%%%%%%%%%%%%%%%%%%%%%%%%%%%%%%%%%%%%%%%%%%
%%%%%%%%%%%%%%%%%%%%%%%%%%%%%%%%%%%%%%%%%%%%%%%%%%%%%

 \section{Graphical characterization of span distributions}\label{graphcarpap}

In this appendix we show how Theorem \ref{unfact1} can be proven by means of a direct graphical approach. This approach involves looking at the earliest intersections of paths starting along different edges from a fixed vertex. It turns out also that from this intersection data one can completely determine $\mathcal{S}(T)$. 
As argued in the proof of Theorem \ref{unfact1} we need to consider only unlabeled trellises, so all trellises in this appendix will be unlabeled. 

Now, let us first give some notation. Given a multiset $S$, we write $$m(x,S)$$ for the multiplicity of $x$ in $S$. If  $\bm{e}=vv'$ is an edge of a trellis $T$ then 
\begin{align*}
h(\bm{e})&:=v'\\
t(\bm{e})&:=v
\end{align*}
 are respectively the \textit{head}  and  \textit{tail} of $\bm{e}$. 
Given  two different edges $\bm{e}\neq \bm{e}'$ of $T$ such that $t(\bm{e})=t(\bm{e}')$, 
we define $$l(\bm{e},\bm{e}')$$ to be the smallest $r\geq0$ such there exist two (directed) paths $\bm{p}=v_{0}\ldots v_{r+1}$,   $\bm{p}'=v'_{0}\ldots v'_{r+1}$ in $T$ satisfying  $\bm{e}=v_{0}v_{1}$, $\bm{e}'=v'_{0}v'_{1}$ (so that $v_{0}=v'_{0}$), and $v_{r+1}=v'_{r+1}$. If there is no  path satisfying those conditions then we put $l(\bm{e},\bm{e}'):=\infty$. 
We define then the multiset $$I(\bm{e}):=\{ \{l(\bm{e},\bm{e}')|\bm{e}'\neq \bm{e}, t(\bm{e}')=t(\bm{e})\}\}$$ 

\begin{ex}\label{ex32} Consider the nonlinear trellis
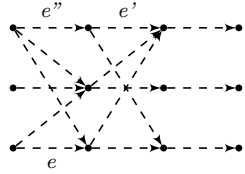
\begin{center}
 \begin{tikzpicture}[yscale=.8,>=latex',shorten >=.9pt, shorten <=1.4pt, line width=.6pt]
  \tikzstyle{every node}=[draw,circle,fill=black,minimum size=2pt,
                        inner sep=0pt]                    
\foreach \x in {0,1,2,3}{
\node at (\x,2) {};};
\foreach \x in {0,1,2,3}{
\node at (\x,1) {};};
\foreach \x in {0,1,2,3}{
\node at (\x,3) {};};
\tikzstyle{every node}=[]
\draw [->,dashed] (2,1) --   (3,1);
\draw [->,dashed] (2,2) --   (3,2);
\draw [->,dashed] (2,3) --   (3,3);

\draw [->,dashed] (0,2) --   (1,2);
\draw [->,dashed] (1,2) --  (2,3);
\draw [->,dashed] (1,1) --  (2,3);
\draw [->,dashed] (1,2) --  (2,2);
\draw [->,dashed] (0,1) --  node[below] {\scriptsize e} (1,1);
\draw [->,dashed] (1,1) --   (2,1);
\draw [->,dashed] (0,3) --  node[above] {\scriptsize e''} (1,3);
\draw [->,dashed] (1,3) --  node[above] {\scriptsize e'}(2,3);
\draw [->,dashed] (0,3) --  node[above] {} (1,1);
\draw [->,dashed] (0,3) --  node[above] {} (1,2);
\draw [->,dashed] (1,3) --  node[above] {} (2,1);
\draw [->,dashed] (0,1) --  node[above] {} (1,2);
    \end{tikzpicture}
    \end{center}
Then $I(\bm{e})=\{\{1\}\}$, $I(\bm{e}')=\{\{2\}\}$, and $I(\bm{e}'')=\{\{1,1\}\}$.
\end{ex}

For linear trellises to compute $l(\bm{e},\bm{e}')$ one can fix $\bm{p}$ and let only $\bm{p}'$  vary. In fact the following holds. 

\begin{obs}\label{obs32} Let $T$ be a linear trellis. Fix a path $\bm{p}=v_{0}\ldots v_{r+1}$ such that $\bm{e}=v_{0}v_{1}$. Then $l(\bm{e},\bm{e}')=\min\{r\geq0| \ \exists \bm{p}'=v'_{0}\ldots v'_{r+1} \textnormal{ such that } \bm{e}'=v'_{0}v'_{1}, v'_{r+1}=v_{r+1}\}$.
\end{obs}
\begin{proof}
This is an immediate consequence of linearity.
\end{proof}

The highly symmetrical graph structure of linear trellises is further reflected in the following fundamental lemma.

\begin{lem}\label{lem31} Let $T$ be a linear trellis of length $n$.
Then:
\begin{itemize}
\item $l(\bm{e},\bm{e}')\leq n-1$ 
\item $I(\bm{e})=I(\bm{e}')$ if $t(\bm{e}), t(\bm{e}')\in V_{a}(T)$ for some $a\in\ZZ_{n}$
\end{itemize}
\end{lem}
\begin{proof}
The inequality is clearly true for elementary trellises, and so it is also true for product of elementary trellises, i.e. linear trellises. Finally, the equality is an immediate consequence of the linearity of $T$. 
\end{proof}

In sight of the above, for  a linear trellis $T$ and $a\in\ZZ_{n}$ it is legitimate to define $$I_{a}(T):=I(\bm{e})$$ where $\bm{e}$ is any edge of $T$ such that $t(\bm{e})\in V_{a}$. 
The next lemma tells us that $I_{a}(T)$ is determined only by those elementary factors of $T$ whose span starts at $a$.
\begin{lem}\label{lem32}
Let $T$ be a linear trellis, and let $T' $ be an elementary trellis with span not starting at $a\in\ZZ_{n}$.  Then $I_{a}(T\otimes T')=I_{a}(T)$.
\end{lem}
\begin{proof}
The equality follows immediately from the fact that any path $\bm{p}$ in $T\otimes T'$ is given by a sequence of vertices $(v_{0},v'_{0}),\ldots,(v_{r+1},v'_{r+1})$ where $(v_{i},v'_{i})\in V_{i}(T\otimes T')=V_{i}(T)\times V_{i}(T')$ for all $i$.  
\end{proof}

The above lemma is crucial for the following theorem, which is the key result:

\begin{thm}\label{thm31} Let $T$ be an unlabeled linear trellis.
Fix $a\in\ZZ_{n}$. 
Then for $ l=0,\ldots,n-1$
\begin{equation}\label{eq5}
1+\sum_{i=0}^{l} m(i,I_{a}(T))=\exp_{q}\Bigl({\sum_{i=0}^{l}m((a,i),T)}\Bigr)
\end{equation}
\end{thm}
\begin{proof}
Since $T$ is an unlabeled linear trellis we have that $T=\bm{0}|(a_{1},l_{1})\otimes\ldots\otimes \bm{0}|(a_{s},l_{s})$ for some $a_{j}\in\ZZ_{n}$ and $l_{j}>0$.
By the above lemma we can assume without loss of generality that $a_{j}=a$ and $l_{j}<n$ 
for all $j=1,\ldots,s$. By a shift we can also clearly assume that $a=0$, so that $T$ is conventional. Since all the spans start at $a$, the outgoing degree of any $v\in V_{a}(T)$ with $a>0$ is equal to one. But then both sides  of \eqref{eq5} are equal to the number of cycles in $T$ that pass through $0\in V_{l}(T)$. 
\end{proof}

For any fixed $a\in\ZZ_{n}$, equation \ref{eq5} can be recursively solved for $m((a,i),T)$, $i=0,\ldots,n-1$, and the solution is unique. Thus the graph-theoretical data given by the multisets $I_{a}(T)$ gives us the possibility to find all nondegenerate spans along with their multiplicities. The multiplicity of the  
degenerate span $\ZZ_{n}$ is instead given by the logarithm of the number of connected components of $T$. Also, as the left-hand side of \ref{eq5} depends only on $T$ and not on its factorization, this gives another proof of Theorem \ref{unfact1}.

\begin{ex}\label{ex33} Let $T$ be the linear unlabeled trellis
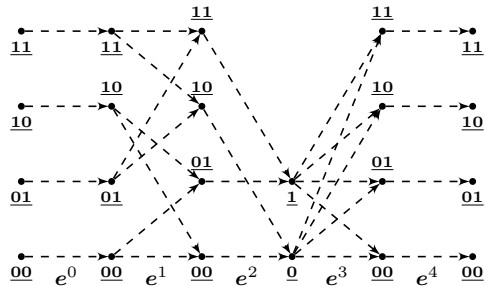
\begin{center}
 \begin{tikzpicture}[xscale=1.2,>=latex',shorten >=.9pt, shorten <=1.4pt, line width=.6pt]
  \tikzstyle{every node}=[draw,circle,fill=black,minimum size=2pt,
                        inner sep=0pt]                    
\foreach \x in {0,1,2,4,5}{
\node at (\x,0) {};};
\foreach \x in {0,1,2,4,5}{
\node at (\x,1) {};};
\foreach \x in {0,1,2,4,5}{
\node at (\x,2) {};};
\foreach \x in {0,1,2,4,5}{
\node at (\x,3) {};};
\node at (3,0) {};
\node at (3,1) {};
\tikzstyle{every node}=[]
\draw (0,0) node [below] {\tiny $\mathbf{\underline{00}}$};
\draw (0,1) node [below] {\tiny $\mathbf{\underline{01}}$};
\draw (0,2) node [below] {\tiny $\mathbf{\underline{10}}$};
\draw (0,3) node [below] {\tiny $\mathbf{\underline{11}}$};

\draw (1,0) node [below] {\tiny $\mathbf{\underline{00}}$};
\draw (1,1) node [below] {\tiny $\mathbf{\underline{01}}$};
\draw (1,2) node [above] {\tiny $\mathbf{\underline{10}}$};
\draw (1,3) node [below] {\tiny $\mathbf{\underline{11}}$};

\draw (2,0) node [below] {\tiny $\mathbf{\underline{00}}$};
\draw (2,1) node [above] {\tiny $\mathbf{\underline{01}}$};
\draw (2,2) node [above] {\tiny $\mathbf{\underline{10}}$};
\draw (2,3) node [above] {\tiny $\mathbf{\underline{11}}$};

\draw (3,0) node [below] {\tiny $\mathbf{\underline{0}}$};
\draw (3,1) node [below] {\tiny $\mathbf{\underline{1}}$};

\draw (4,0) node [below] {\tiny $\mathbf{\underline{00}}$};
\draw (4,1) node [above] {\tiny $\mathbf{\underline{01}}$};
\draw (4,2) node [above] {\tiny $\mathbf{\underline{10}}$};
\draw (4,3) node [above] {\tiny $\mathbf{\underline{11}}$};

\draw (5,0) node [below] {\tiny $\mathbf{\underline{00}}$};
\draw (5,1) node [below] {\tiny $\mathbf{\underline{01}}$};
\draw (5,2) node [below] {\tiny $\mathbf{\underline{10}}$};
\draw (5,3) node [below] {\tiny $\mathbf{\underline{11}}$};

\draw [->,dashed] (0,0) --  node[below] {\scriptsize $\bm{e}^{0}$}  (1,0);
\draw [->,dashed] (0,1) --   (1,1);
\draw [->,dashed] (0,2) --   (1,2);
\draw [->,dashed] (0,3) --   (1,3);

\draw [->,dashed] (1,0) -- node[below] {\scriptsize $\bm{e}^{1}$}   (2,0);
\draw [->,dashed] (1,0) --   (2,1);

\draw [->,dashed] (1,1) --   (2,2);
\draw [->,dashed] (1,1) --   (2,3);

\draw [->,dashed] (1,2) --   (2,0);
\draw [->,dashed] (1,2) --   (2,1);

\draw [->,dashed] (1,3) --   (2,2);
\draw [->,dashed] (1,3) --   (2,3);

\draw [->,dashed] (2,0) --  node[below] {\scriptsize $\bm{e}^{2}$}  (3,0);
\draw [->,dashed] (2,1) --   (3,1);

\draw [->,dashed] (2,2) --   (3,0);
\draw [->,dashed] (2,3) --   (3,1);

\draw [->,dashed] (3,0) --  node[below] {\scriptsize $\bm{e}^{3}$}  (4,0);
\draw [->,dashed] (3,0) --   (4,1);
\draw [->,dashed] (3,0) --   (4,2);
\draw [->,dashed] (3,0) --   (4,3);

\draw [->,dashed] (3,1) --   (4,0);
\draw [->,dashed] (3,1) --   (4,1);
\draw [->,dashed] (3,1) --   (4,2);
\draw [->,dashed] (3,1) --   (4,3);

\draw [->,dashed] (4,0) --  node[below] {\scriptsize $\bm{e}^{4}$}  (5,0);
\draw [->,dashed] (4,1) --   (5,1);
\draw [->,dashed] (4,2) --   (5,2);
\draw [->,dashed] (4,3) --   (5,3);

    \end{tikzpicture}
    \end{center}
Then one easily computes:
\begin{itemize}
\item $I_{0}(T)=I(\bm{e}^{0})=\emptyset$
\item $I_{1}(T)=I(\bm{e}^{1})=\{\{2\}\}$
\item $I_{2}(T)=I(\bm{e}^{2})= \emptyset $
\item $I_{3}(T)=I(\bm{e}^{3})=\{\{3,4,4\}\}$
\item $I_{4}(T)=I(\bm{e}^{4})= \emptyset $
\end{itemize}
So, by equation \ref{eq5} we get that $\mathcal{S}_{0}(T)=\mathcal{S}_{2}(T)=\mathcal{S}_{4}(T)=\emptyset$, $\mathcal{S}_{1}(T)=\{\{(1,2)\}\}$, 
and %NOT PRESENT
$\mathcal{S}_{3}(T)=\{\{(3,3),(3,4)\}\}$, where $\mathcal{S}_{i}(T):=\{\{(a,l)\in\mathcal{S}(T)|l=i\}\}$. Therefore $T=\bm{0}|(1,2)\otimes\bm{0}|(3,3)\otimes\bm{0}|(3,4)$. 
\end{ex}

%%%%%%%%%%%%%%%%%%%%%%%%%%%%%%%%%%%%%%%%%%%%%%%%%%%%

\section*{Acknowledgments}
%This work is part of the Ph.D. research of D. Conti, 
%%supported by the Science Foundation Ireland grant 06/MI/006, and 
%supervised by N. Boston. 
The authors thank Prof. G. David Forney for invaluable feedback and illuminating comments, and Prof. Heide Gluesing-Luerssen for stimulating discussions.

%%%%%%%%%%%%%%%%%%%%%%%%%%%%%%%%%%%%%%%%%%%%%%%%%%%%

%%%%%%%%%%%%%%%%%%%%%%%%%%%%%%%%%%%%%%%%%%%%%%%%%%%%


\begin{thebibliography}{90}
%\small



\bibitem{BCJR} L. R. Bahl, J. Cocke, F. Jelinek, and J. Raviv,   \emph{Optimal decoding of linear codes for minimizing symbol error rate},  IEEE Trans. Inform. Theory 20, pp. 284 - 287, 1974 
\bibitem{BJ} I. E. Bocharova, R. Johannesson, \emph{Tail- biting codes: Bounds and search results},  IEEE Trans. Inform. Theory 48, pp. 137-148,  2002


\bibitem{B}N. Boston, \emph{A multivariate weight enumerator for tail-biting trellis pseudocodewords},  Proc.  ``Workshop on Algebra, Combinatorics and Dynamics'', Belfast, Springer, 2009
\bibitem{CFV0}A. R. Calderbank, G. D. Forney, Jr., A. Vardy, \emph{Classification of certain tail-biting generators for the binary Golay code}, ``Codes, curves, and signals: common threads in communications'', Boston, MA, Kluwer, pp. 127-153, 1998
\bibitem{CFV}A. R. Calderbank, G. D. Forney, Jr., A. Vardy, \emph{Minimal tail-biting trellises: The Golay code and more}, IEEE Trans. Inform. Theory 45 (5), pp. 1435-1455, 1999

\bibitem{C} D. Conti, \emph{An Algebraic Development of Trellis Theory}, PhD Thesis, UCD, Ireland, 2012 
\bibitem{CB} D. Conti, N. Boston, \emph{Matrix representations of trellises and enumerating trellis 
pseudocodewords},  Allerton Conference 2011 Proceedings, 2011
\bibitem{CB2} D. Conti, N. Boston, \emph{Factoring Linear Trellises},  International Zurich Seminar
on Communications 2012 Proceedings, 2012
\bibitem{CB3} D. Conti, N. Boston, \emph{The Factorization Theorem and New Algebraic Insights into the Theory of Linear Trellises}, accepted to appear in Allerton Conference 2012 Proceedings


\bibitem{F} J. Feldman, \emph{Decoding error-correcting codes via linear programming}, PhD Thesis, MIT, 2003
\bibitem{Fo} G. D. Forney, Jr., \emph{Coset codes. Part II: Binary lattices and related codes}, IEEE Trans. Inform. Theory 34, pp. 1152-1187, 1988
\bibitem{Fonorm} G. D. Forney, Jr., \emph{Codes on graphs: Normal realizations}, IEEE Trans. Inform. Theory 47, pp. 520-548, 2001
\bibitem{FoGl}G. D. Forney, H. Gluesing-Luerssen, \emph{Codes on Graphs: Observability, Controllability and Local Reducibility}, to appear in the IEEE Trans. Inf. Theory, available at  http://arxiv.org/abs/1203.3115, 2012
\bibitem{FoGl2}G. D. Forney, H. Gluesing-Luerssen, \emph{Reducing complexity of tail-biting trellises}, 2012 IEEE International Symposium on Information Theory, available at  http://arxiv.org/abs/1202.1336, 2012




\bibitem{FKKR}G. D. Forney, Jr., R. Koetter, F. R. Kschischang, A. Reznik, \emph{On the effective weights of pseudocodewords for codes defined on graphs with cycles}, ``Codes, Systems and Graphical Models'', Springer, pp. 101-112, 2001
\bibitem{FKMT}G. D. Forney, Jr., R. Koetter, B. Marcus, S. Tuncel,  \emph{Iterative decoding of tail-biting trellises and connections with symbolic dynamics}, ``Codes, Systems and Graphical Models'', Springer, pp. 239-264, 2001

\bibitem{FoTr}G. D. Forney, M.D. Trott, \emph{The dynamics of group codes: state spaces, trellis diagrams, and canonical encoders}, IEEE Trans. Inform. Theory 39, pp. 1491-1513, 1993


\bibitem{GW}H. Gluesing-Luerssen, E. Weaver, \emph{Linear tail-biting trellises: Characteristic generators and the BCJR construction}, IEEE Trans. Inf. Theory 57, pp. 738-751, 2011
\bibitem{GW2}H. Gluesing-Luerssen, E. Weaver, \emph{Characteristic generators and Dualization of Tail-Biting Trellises}, IEEE Trans. Inf. Theory 57, pp. 7418 - 7430, 2011
\bibitem{H}G. B. Horn, \emph{Iterative decoding and pseudocodewords}, Ph.D. Thesis, California Institute of Technology, Pasadena, California, USA, 1999

\bibitem{Ko}
R. Koetter, \emph{On the representation of codes in Forney graphs}, in ``Codes, Graphs, and Systems: A Celebration of the Life and Career of G. David Forney, Jr.'' (R. E. Blahut, R. Koetter, eds.), Kluwer, pp. 425-450, 2002.
\bibitem{KV0}R. Koetter, A. Vardy, \emph{Construction of minimal tail-biting trellises}, 
Proc. IEEE Workshop on Information Theory, Killarney, Ireland, pp. 72-74, 1998
\bibitem{KV2}R. Koetter, A. Vardy, \emph{On the theory of linear trellises}, ``Information, Coding and Mathematics'', Boston, MA, Kluwer, pp. 323-354, 2002
\bibitem{KV}R. Koetter, A. Vardy, \emph{The structure of tail-biting trellises: minimality and basic principles}, IEEE Trans. Inform. Theory 49 (9), pp. 2081-2105, 2003
\bibitem{KVo}R. Koetter, P. O. Vontobel, \emph{Graph-Cover Decoding and Finite-Length Analysis of Message-Passing Iterative Decoding of LDPC Codes}, available at http://arxiv.org/pdf/cs.it/0512078.pdf, 2005
\bibitem{KsSo} F. R. Kschischang, V. Sorokine, \emph{On the trellis structure of block codes}, IEEE Trans. Inform. Theory 41 (6), pp. 1924-1937, 1995
\bibitem{KsV} F. R. Kschischang, A. Vardy, \emph{Proof of a conjecture of McEliece regarding the expansion index of the minimal trellis}, IEEE Trans. Inform. Theory 42 (6), pp. 2027-2034, 1996
\bibitem{Mc} R. J. McEliece, \emph{On the BCJR trellis for linear block codes},  IEEE Trans. Inform. Theory 42 (4), pp. 1072-1092, 1996
\bibitem{M}D. J. Muder, \emph{Minimal trellises for block codes}, IEEE Trans. Inform. Theory 34 (5), pp. 1049-1522, 1988
\bibitem{N}A. V. Nori, \emph{Unifying views of tail-biting trellis for linear block codes}, Ph.D. Thesis, Indian Institute of Science, Bangalore, India, 2005
\bibitem{RB} I. Reuven, Y. Be'ery, \emph{Tail-Biting trellises of block codes: trellis complexity and Viterbi decoding complexity}, IEICE Trans. Fundamentals, vol. E82-A, no.10, October 1999
\bibitem{SB} Y. Shani, Y. Be'ery, \emph{Linear Tail-Biting Trellises, the Square-Root Bound, and Applications for ReedÐMuller Codes}, IEEE Trans. Inform. Theory 46 (4), pp. 1514-1523, 2000
\bibitem{SB2} Y. Shani, Y. Be'ery, \emph{Lower Bounds on the State Complexity of Linear Tail-Biting Trellises}, IEEE Trans. Inform. Theory 50 (3), pp. 566-571, 2004
\bibitem{SKSR}P Shankar, P.N.A. Kumar, H. Singh, B.S. Rajan, \emph{Minimal Tail-Biting Trellises for Certain Cyclic Block Codes Are Easy to Construct},  LNCS 2076, Springer-Verlag, pp. 627Ð638, 2001
\bibitem{SoTi} G. Solomon, H. C. A. Tilborg, \emph{A connection between block and convolutional codes}, SIAM
J. Appl. Math., vol. 37, pp. 358-369, 1979

\bibitem{V}A. Vardy, \emph{Trellis structure of codes}, ``Handbook of Coding Theory'', Volume 2, Elsevier Science, 1998

\bibitem{W}N. Wiberg, \emph{Codes and decoding on general graphs}, Ph.D. dissertation, Univ. Link\"oping, Sweden, 1996


\end{thebibliography}
\end{document}